\documentclass[letterpaper,11pt]{article}
\usepackage[utf8]{inputenc}

\pdfoutput=1

\usepackage{graphicx}
\usepackage{amsmath, amsthm, amssymb, thm-restate}
\usepackage{algorithmicx}
\usepackage[table,xcdraw]{xcolor}
\usepackage[ruled,vlined,linesnumbered]{algorithm2e}
\usepackage{setspace}
\usepackage{mathtools}
\usepackage[style=alphabetic,natbib=true,maxbibnames=9, maxcitenames=9]{biblatex}

\usepackage{comment}
\usepackage[most]{tcolorbox} 
\usepackage{xfrac}
\usepackage{hyperref}
\usepackage{multirow}
\usepackage{caption}
\usepackage{bm}
\usepackage{newfloat}
\usepackage{enumitem}
\usepackage{dblfloatfix} 
\usepackage{wrapfig}
\usepackage{csquotes}

\usepackage{fancyhdr}

\usepackage{tikz}
\usetikzlibrary{calc}
\usepackage{pgfmath} 

\usepackage[margin=1in]{geometry}

\allowdisplaybreaks

\definecolor{mygreen}{RGB}{10,150,110}
\definecolor{myred}{RGB}{150,10,20}

\hypersetup{
     colorlinks=true,
     citecolor= mygreen,
     linkcolor= myred
}

\renewcommand{\epsilon}{\varepsilon}

\DeclareMathOperator{\E}{\ensuremath{\normalfont \textbf{E}}}

\newcommand{\hiddencomment}[1]{}

\newcommand{\mc}[1]{\ensuremath{\mathcal{#1}}}

\newcommand{\yes}[0]{\ensuremath{\mathsf{YES}}}
\newcommand{\no}[0]{\ensuremath{\mathsf{NO}}}
\newcommand{\yesdist}[0]{\ensuremath{\mathcal{D}_{\yes}}}
\newcommand{\nodist}[0]{\ensuremath{\mathcal{D}_{\no}}}

\newcommand{\ER}[0]{Erdős–Rényi}
\newcommand{\Var}[0]{\text{Var}}
\newcommand{\wt}[1]{\ensuremath{\widetilde{#1}}}

\newcolumntype{D}[1]{>{\displaystyle}#1}
\newcolumntype{S}[1]{>{\scriptstyle}#1}

\newcommand{\Binom}{\textnormal{Binom}}
\DeclarePairedDelimiter\card{\lvert}{\rvert}

\usepackage[noabbrev,nameinlink]{cleveref}
\crefname{lemma}{Lemma}{Lemmas}
\crefname{theorem}{Theorem}{Theorems}
\crefname{property}{Property}{Properties}
\crefname{claim}{Claim}{Claims}
\crefname{result}{Result}{Results}
\crefname{definition}{Definition}{Definitions}
\crefname{observation}{Observation}{Observations}
\crefname{proposition}{Proposition}{Propositions}
\crefname{assumption}{Assumption}{Assumptions}
\crefname{line}{Line}{Lines}
\crefname{figure}{Figure}{Figures}
\creflabelformat{property}{(#1)#2#3}
\crefname{equation}{}{}
\crefname{section}{Section}{Sections}
\crefname{appendix}{Appendix}{Appendices}
\crefname{algCounter}{Algorithm}{Algorithms}
\Crefname{algCounter}{Algorithm}{Algorithms}

\newtheorem{theorem}{Theorem}

\newtheorem{lemma}{Lemma}[section]
\newtheorem{proposition}[lemma]{Proposition}
\newtheorem{corollary}[lemma]{Corollary}

\newtheorem{definition}[lemma]{Definition}
\newtheorem{claim}[lemma]{Claim}

\newtheorem{observation}[lemma]{Observation}
\newtheorem{remark}{Remark}
\newtheorem*{remark*}{Remark}

\definecolor{mylightgray}{RGB}{230,230,230}

\algnewcommand{\IIf}[2]{\textbf{if} #1 \textbf{then} #2}
\algnewcommand{\EndIIf}{\unskip\ \algorithmicend\ \algorithmicif}

\newenvironment{graytbox}{
\par\addvspace{0.1cm}
\begin{tcolorbox}[width=\textwidth,
                  boxsep=5pt,
                  left=1pt,
                  right=1pt,
                  top=2pt,
                  bottom=2pt,
                  boxrule=0pt,
                  arc=0pt,
                  colback=mylightgray,
                  colframe=black,
                  ]%%
}{
\end{tcolorbox}
}

\newenvironment{whitetbox}{
\par\addvspace{0.1cm}
\begin{tcolorbox}[width=\textwidth,
                  boxsep=5pt,
                  left=1pt,
                  right=1pt,
                  top=2pt,
                  bottom=2pt,
                  boxrule=1pt,
                  arc=0pt,
                  colframe=black,
                  colback=white
                  ]%%
}{
\end{tcolorbox}
}

\newcounter{algCounter}

\let\oldlemma\lemma
\renewcommand{\lemma}{%
  \renewcommand{\emph}[1]{\textbf{##1}}% Change \emph to bold in lemma
  \oldlemma}

\makeatletter
\renewcommand{\paragraph}{%
  \@startsection{paragraph}{4}%
  {\z@}{10pt}{-1em}%
  {\normalfont\normalsize\bfseries}%
}
\makeatother

\makeatletter
\patchcmd{\@algocf@start}% <cmd>
  {-1.5em}% <search>
  {0pt}% <replace>
  {}{}% <success><failure>
\makeatother

\title{Tight Pair Query Lower Bounds for\\ Matching and Earth Mover's Distance \thanks{
Amir Azarmehr and Soheil Behnezhad are supported in part by NSF CAREER Award CCF-2442812 and a Google Research Award. Mohammad Roghani is supported by NSF CAREER Award CCF-2337901. Aviad Rubinstein is supported by the David and Lucille Packard Fellowship and NSF CAREER Award CCF-2337901.}}

  \author{
   Amir Azarmehr\\{\em Northeastern University} \and
   Soheil Behnezhad\\{\em Northeastern University} \and
   Mohammad Roghani\\{\em Stanford University} \and
   Aviad Rubinstein\\{\em Stanford University}}

\date{}

\begin{document}

\maketitle

\thispagestyle{empty}

\begin{abstract}
{
\setlength{\parskip}{0.2cm}
    How many adjacency matrix queries (also known as pair queries) are required to estimate the size of a \textbf{maximum matching} in an $n$-vertex graph $G$? We study this fundamental question in this paper.

    On the upper bound side, an algorithm of Bhattacharya, Kiss, and Saranurak [FOCS'23] gives an estimate that is within $\epsilon n$ of the right bound with $n^{2-\Omega_\epsilon(1)}$ queries, which is subquadratic in $n$ (and thus sublinear in the matrix size) for any fixed $\epsilon > 0$. On the lower bound side, while there has been a lot of progress in the {\em adjacency list} model, no non-trivial lower bound has been established for algorithms with adjacency matrix query access. In particular, the only known lower bound is a folklore bound of $\Omega(n)$, leaving a huge gap.

    In this paper, we present the first superlinear in $n$ lower bound for this problem. In fact, we close the gap mentioned above entirely by showing that the algorithm of [BKS'23] is optimal. Formally, we prove that for any fixed $\delta > 0$, there is a fixed $\epsilon > 0$ such that an estimate that is within $\epsilon n$ of the true bound requires $\Omega(n^{2-\delta})$ adjacency matrix queries.
    
    Our lower bound also has strong implications for estimating the earth mover's distance between distributions. For this problem, Beretta and Rubinstein [STOC'24] gave an $n^{2-\Omega_\epsilon(1)}$ time algorithm that obtains an additive $\epsilon$-approximation and works for any distance function. Whether this can be improved generally, or even for metric spaces, had remained open. Our lower bound rules out the possibility of any improvements over this bound, even under the strong assumption that the underlying distances are in a (1, 2)-metric.
    }

\end{abstract}

{
\clearpage
\hypersetup{hidelinks}
\vspace{1cm}
\renewcommand{\baselinestretch}{0.1}
\setcounter{tocdepth}{2}
\tableofcontents{}
\thispagestyle{empty}
\clearpage
}

\setcounter{page}{1}

\newpage

\section{Introduction}

With the exponential increase in the size of datasets, algorithms that run in sublinear time---that is, algorithms that examine only a small portion of their input---have become increasingly important. To allow such algorithms, one has to provide some form of a query access to the input. For graph problems, a natural way to formalize this is through {\em adjacency matrix} queries---also known as the {\em dense graph model}. In this model, an algorithm can (adaptively) query whether there exists an edge between any pair of vertices of its choice. The goal is to solve a graph problem with as few queries as possible. Our focus is particularly on a fundamental question:

\begin{quote}
    How many adjacency matrix queries are required to estimate the size of a \textbf{maximum matching} in an $n$-vertex graph $G$?
\end{quote}

Recall that a maximum matching is the largest set of edges in the graph such that no two edges share a vertex. This question sits at the center of sublinear time graph algorithms, as the complexity of various different graph problems relate directly to matchings \cite{ChenICALP20,behnezhad2021,ChenMetric-Arxiv22, chen2023query, mahabadi2025sublinearmetricsteinerforest, TSP-icalp24, steiner-tree-itcs}. Sublinear time matching algorithms have also found numerous applications in other areas. For example, sublinear time matching algorithms---particularly in the adjacency matrix model---have led to a series of recent breakthroughs in the area of dynamic graph algorithms both for maintaining the size as well as the edges of the matching \cite{Behnezhad-SODA23,bhattacharyaKSW-SODA23,BhattacharyaKS-FOCS23,AzarmehrBR-SODA24, BehnezhadG-FOCS24,assadi2025improved}.

Prior to our work, the landscape for this problem exhibited a significant gap between upper and lower bounds. We briefly overview the state-of-the-art before stating our contribution.

\paragraph{The state of the art:} On the algorithmic side, \citet{BhattacharyaKS-FOCS23} showed that an estimate that is within $\epsilon n$ of the maximum matching size can be found in $n^{2-\Omega_\epsilon(1)}$ time. For any fixed $\epsilon > 0$, this is sub-quadratic in $n$ and thus sublinear in the whole input (i.e., the $n \times n$ size adjacency matrix). If in addition to the additive error of $\epsilon n$, we allow a multiplicative error of $1/2$, then an algorithm of \citet{behnezhad2021} runs in $\widetilde{O}(n/\epsilon^3)$ time. There are also several other algorithms in the literature that achieve different trade-offs between running time and approximation ratio \cite{KapralovSODA20, BehnezhadRRS-SODA23, ChenICALP20, BhattacharyaKS-STOC23, mahabadi20250}.

On the flip side, there is a simple and folklore lower bound of $\Omega(n)$ for any multiplicative-additive approximation. This follows from taking the input graph to be either a random perfect matching or the empty graph. Distinguishing the two graphs requires finding one of the $\Theta(n)$ adjacent pairs out of $\Theta(n^2)$ pairs, which clearly requires $\Omega(n)$ queries even with randomization. While there has been a lot of progress on lower bounds in an alternative {\em adjacency list} query access model \cite{Behnezhad-RRS-STOC23,BehnezhadRR-FOCS23,behnezhad2024approximating}---which we overview extensively soon---none of this progress translates to adjacency matrix queries and the folklore $\Omega(n)$ lower bound remains the only known lower bound in this model.

\paragraph{Our contribution:} This state of affairs leaves a huge gap between lower and upper bounds for additive $\epsilon n$ errors. While the lower bound is only linear (i.e., $\Omega(n)$) the best upper bound is barely subquadratic (i.e., $n^{2-\Omega_\epsilon(1)}$ \cite{BhattacharyaKS-FOCS23}).  Our work particularly focuses on this gap and closes it completely. Namely, we show that the algorithm of \cite{BhattacharyaKS-FOCS23} is (essentially) optimal:

\begin{graytbox}
    \begin{restatable}{theorem}{mainTheorem}
    \label{thm:main-intro}
        For every $\delta > 0$ there exists $\epsilon > 0$ (i.e., $\epsilon$ is only a function of $\delta$) such that any algorithm (possibly randomized) that estimates (with probability at least 2/3) the size of maximum matching up to an additive error of $\epsilon n$ must make at least $\Omega(n^{2-\delta})$ queries to the adjacency matrix of the graph.
    \end{restatable}
\end{graytbox}

\paragraph{Connection to the adjacency list model:} Let us now discuss the {\em adjacency list} query access model for which we have a much better understanding of lower bounds for estimating the size of the maximum matching. In this model, instead of specifying two vertices, each query of the algorithm specifies a vertex $v$ and a number $i$; the response is then the ID of the $i$-th neighbor of $v$ (or $\perp$ if $v$ has less than $i$ neighbors).

Starting with \cite{Behnezhad-RRS-STOC23}, a sequence of recent results by \citet{Behnezhad-RRS-STOC23,BehnezhadRR-FOCS23,behnezhad2024approximating} established strong lower bounds for estimating the maximum matching size in the adjacency list model. Most relevant to the present paper is \cite{behnezhad2024approximating} which proves an equivalent of \cref{thm:main-intro} in the adjacency list model. 

Unfortunately this progress in the adjacency list model does not lead to any non-trivial (i.e., super-linear in $n$) lower bounds in the adjacency matrix model.  Specifically, {\em proving adjacency matrix lower bounds appears to be much more difficult than for adjacency list} for two reasons:
\begin{enumerate}
    \item[$(i)$] \textbf{The ability to query vertex-induced subgraphs:} An algorithm in the adjacency matrix model can select a small subset $U \subseteq V$ and learn the entire induced subgraph $G[U]$ with just $O(|U|^2)$ queries. Not only does this ability prevent a straightforward extension of adjacency list lower bounds to adjacency matrix ones, but in fact has led to formal separations between the two models for some problems. For example, a $(\Delta+1)$ vertex coloring in dense graphs with $\Delta = \Theta(n)$ can be found with $\widetilde{O}(n^2/\Delta) = \widetilde{O}(n)$ adjacency matrix queries but requires $\Omega(n \Delta) = \Omega(n^2)$ adjacency list queries \cite{AssadiCK19}. There are also many property testing problems that can be solved with constant adjacency matrix queries by selecting a constant size $U$ at random and querying all the pairs in $U$, see e.g. \cite{AlonFKS99,GoldreichT01}.
    \item[$(ii)$] \textbf{Discovering non-edges:} Another key difference between the two models is that each query in the adjacency list model discovers an edge, whereas adjacency matrix queries reveal both edges and non-edges in the graph. Therefore, any analysis would have to show that non-edges do not reveal too much information about the hard instance. Indeed, as we will discuss extensively in \cref{sec:techniques}, this is the main obstacle that we overcome in proving \cref{thm:main-intro}.
\end{enumerate}

\paragraph{Other implications of our result:}
There are many natural scenarios for computing matchings where the access to the underlying graph is through adjacency matrix queries. Here we briefly mention some of these scenarios and the implications of our \cref{thm:main-intro}:
\begin{itemize}[leftmargin=10pt]
    \item \textbf{Estimating the earth mover's distance:} The earth mover's distance (EMD) is the most natural measure of similarity between two distributions defined over the elements of a distance space (often a metric). For a formal definition see \cref{def:EMD}. \citet*{BerettaR24} showed that if the distributions are defined over $n$ elements, then there is an algorithm obtains an additive $\epsilon$-approximation of EMD with $n^{2-\Omega_\epsilon(1)}$ queries to the distance metric, i.e., with truly subquadratic queries.

    The algorithm of \cite{BerettaR24} does not make use of the common metric assumption and works for arbitrary distance. Whether this barely subquadratic time algorithm can be improved to say $O_\epsilon(n^{1.9})$ remained open even in the case of metric spaces (see also the related paper of \citet{andoni2023sub}). Our lower bound strongly rules out this possibility, even in the case of $(1, 2)$-metrics; see \cref{thm:EMD} in \cref{sec:EMD}. In order to obtain this result, it is important that our lower bound works in the adjacency matrix model (as opposed to adjacency lists) as it naturally captures pairwise distance queries.
    \item \textbf{Dynamic algorithms:} In a recent line of work \cite{Behnezhad-SODA23,bhattacharyaKSW-SODA23,BhattacharyaKS-FOCS23,BehnezhadG-FOCS24,assadi2025improved, AzarmehrBR-SODA24} sublinear time algorithms for maximum matching have been utilized to obtain significant improvements for the maximum matching problem in the dynamic setting. The role that sublinear time algorithms play in the works of \cite{Behnezhad-SODA23,bhattacharyaKSW-SODA23,BhattacharyaKS-FOCS23} is very different from that of \cite{BehnezhadG-FOCS24,assadi2025improved}. But, curiously, all of these results rely on adjacency matrix queries as opposed to adjacency matrix queries. Our lower bound is therefore of interest to the dynamic community. Concretely, our \cref{thm:main-intro} implies that an update-time of barely sublinear in $n$ (i.e., the bound achieved by \cite{BhattacharyaKS-FOCS23}) is best one can hope for with a black-box application sublinear time algorithms.
\end{itemize}

\section{Overview of Techniques}\label{sec:techniques}

In this section, we provide a high-level overview of the lower bound presented in \cref{thm:main-intro}. Before diving into the new techniques and ideas introduced in this paper, we first briefly review the previous constructions for adjacency list lower bounds and their key concepts. Then, we identify the challenges that arise when working with the adjacency matrix, particularly when the algorithm can query pairs corresponding to ``non-edges", and explain how we address these challenges.

\subsection{Existing Constructions: Ideas and Barriers}

We first focus on the work of \cite{Behnezhad-RRS-STOC23} and their core construction\footnote{We disregard the dummy vertices in their construction, as their primary role is to congest the results of adjacency list queries when the core construction is sparse.}. Their input distribution consists of two types of graphs: \yesdist{} and \nodist{}. There is a significant gap between the sizes of the maximum matchings in graphs drawn from these distributions, with $\yesdist{}$ having a larger matching. The key idea of the paper is to demonstrate that the algorithm cannot distinguish {\em good matching} edges in \yesdist{}. To achieve this, they prove that the queried subgraph within the core forms a tree, as the core is sufficiently sparse and the algorithm makes at most $O(n^{1.2})$ queries. Finally, they show that the queried trees from both distributions are nearly identical using a coupling argument, preventing the algorithm from differentiating between the two types, which implies that no algorithm can achieve $2/3$-approximation in $o(n^{1.2})$ time. The later work of \cite{BehnezhadRR-FOCS23} follows a similar outline but introduces a modification in the core that results in a lower bound with a different trade-off between the approximation ratio and running time. However, this lower bound still heavily relies on the fact that the queried subgraph within the core remains acyclic due to their choice of the core's degree and the imposed bound on the algorithm's running time.

\vspace{-0.5em}
\paragraph{Non-edges reveal information:} Let us zoom out and examine why these approaches are insufficient for proving a lower bound in the adjacency matrix model. In this model, the algorithm can focus all its queries within the core, as it has the power to choose which pairs of vertices to query. Furthermore, while the core is sufficiently sparse and most of the queries result in non-edges, these non-edges still provide information about the construction. For example, consider the case where we are given a graph that contains only a perfect matching (which is hidden from the algorithm). Before any queries are made, if the algorithm randomly queries a pair of vertices $(u,v)$, the probability that the pair forms an edge is exactly $1/(n-1)$. However, if the result is a non-edge, the probabilities of other pairs shift. For instance, the probability that a pair $(u,w)$ is an edge, conditioned on $(u,v)$ being a non-edge, becomes $1/(n - 2)$. While this may at first appear insignificant, it has to be noted that the algorithm is given nearly quadratically many queries, and so will learn about a huge number of non-edges. Therefore, we must carefully account for the non-edges when proving lower bounds, especially when coupling the queried subgraph in the two distributions. In particular, we need to include non-edges in the coupling argument. Notably, the state-of-the-art result for proving lower bounds in the adjacency list model \cite{behnezhad2024approximating} discovers at most $o(n)$  {\em edges} relevant edges and its arguments substantially relied on this limitation.

\subsection{Introducing Pseudo Edges}

As discussed earlier, non-edges reveal information about the construction. However, a key observation is that the amount of information revealed by non-edges depends on the density of the graph. If the construction is sparse, we expect most of the queries to be non-edges. Intuitively, this means the information revealed by these non-edges is very limited. To provide more intuition, consider the following simple example. Suppose we have three subsets of vertices, $V_1$, $V_2$, and $V_3$, where there exists a perfect matching between $V_1$ and $V_2$, a 2-regular graph between $V_2$ and $V_3$, and no edges between $V_1$ and $V_3$. Now, suppose the algorithm queries a pair $(u,v)$ and it turns out to be a non-edge. In this case, it is slightly more likely that the pair belongs to $V_1 \times V_3$ compared to other possible pairs, but the difference in probability is bounded by $O(1/n)$.

To formalize this observation, we introduce the concept of \emph{pseudo edges}. The idea is to mark some of the non-edges as pseudo edges such that all other non-edges reveal no information about the construction. In the example above, suppose that for each pair in \( V_i \times V_j \) for \( i \neq j \), we randomly choose a \(\log n\)-regular graph and mark those pairs as \textit{pseudo edges}. Note that we do not place actual edges between these pairs; they remain non-edges but are marked for the sake of analysis. Additionally, suppose the actual edges are chosen as a subgraph of the pseudo edges.  Now, if a query returns a non-edge that is not marked as a pseudo edge, it does not reveal any information about which \( V_i \) the endpoints belong to. This is because the distribution of non-edges that are not pseudo edges is identical across any two pairs of subsets. Therefore, we can ignore such queries, as they do not provide useful information, and instead focus on the pseudo edges, which are significantly fewer in number.

\subsection{The Challenge with Pseudo Edges}

Before describing how we actually use pseudo edges, we start with perhaps the ``obvious'' way of using pseudo edges and argue why these methods do not quite work. These examples are meant to illustrate why our final construction has to be somewhat involved and rather counter-intuitive.

\vspace{-0.5em}
\paragraph{Attempt 1: regular pseudo edges and regular real edges.} This idea is exactly the same as what we discussed when introducing pseudo edges. However, although non-edges that are not marked as pseudo edges are independent of vertex labels, we face another challenge that makes proving a lower bound difficult. In this approach, we must ensure that the regular graph of real edges is a subgraph of the regular graph of pseudo edges. Achieving this requires a global view of the pseudo edge graph to determine the real edges. However, this introduces correlations between the real edges, making it extremely challenging to show any lower bound.

 \vspace{-0.5em}
\paragraph{Attempt 2: \ER{} pseudo edges and regular real edges.} Instead of using a regular graph for pseudo edges, we now consider an Erdős–Rényi random graph, which is more natural when working with the adjacency matrix. Take the example with vertex sets \(V_1\), \(V_2\), and \(V_3\). For each pair in \(V_i \times V_j\) with \(i \neq j\), we mark the pair as a pseudo edge with probability \(\log n / n\), mimicking an \ER{} graph with an expected degree of \(\log n\). Additionally, suppose that real edges are selected only from among the marked pseudo edges.  As before, if a query results in a non-edge that is not marked as a pseudo edge, it provides no information about which \(V_i\) the endpoints belong to. This is because the \ER{} graph of non-pseudo edges is identical across all pairs of subsets.  

However, while non-pseudo edges remain independent of vertex labels, a new challenge arises that complicates proving a lower bound. Let \( F \) be the set of queried pairs that were identified as non-edges and were not marked as pseudo edges. Let \( L \) be an indicator random variable representing whether a given vertex \( v \) has a specific label. We know that it holds $\Pr[L \mid F] = \Pr[L]$. This is because pseudo edges are independent of vertex labels and form an \ER{} graph between different labels, regardless if they are real edges.  Now, suppose \( H \) is a subset of the actual edges that have been queried and found. Our entire argument relies on the assumption that non-pseudo edges provide no information about labels, meaning the algorithm should ignore them. This requires showing that  $\Pr[L \mid F, H] = \Pr[L \mid H]$. However, note that \( H \) and \( F \) are not necessarily independent. Essentially, our goal is to select a regular subgraph from an \ER{} graph. This process involves first examining all the edges of the \ER{} graph and then choosing a subset of them. Consequently, the decision of whether a given pair forms an edge is not independent of the decisions for other pairs.

\vspace{-0.5em}
\paragraph{Attempt 3: \ER{} pseudo edges and \ER{} real edges.} With the intuition from the previous attempt in mind, we need to develop a construction that ensures the decision of whether a given pair forms a real edge is independent of the decisions for other pairs. To achieve this goal, we choose an \ER{} graph for both pseudo edges and real edges. More specifically, consider the example with vertex sets \( V_1 \), \( V_2 \), and \( V_3 \). For each pair of vertices, we mark the pair as a pseudo edge with probability \(\log n / n\). Furthermore, each pair between \( V_1 \) and \( V_2 \) is marked as a real edge with probability \( 1 / \log n \), each pair between \( V_2 \) and \( V_3 \) is marked as a real edge with probability \( 2 / \log n \), and each pair between \( V_1 \) and \( V_3 \) is marked as a real edge with probability \( 2 / \log n \). Finally, an edge is included in the final construction if and only if it is marked as both a pseudo edge and a real edge. As a result, the expected degree will match our intended values (1 for \( V_1 \), 3 for \( V_2 \), and 2 for \( V_3 \)), consistent with the regular real-edge construction attempt. Furthermore, we achieve the desirable independence property we were seeking.

However, another challenge arises with this construction. Consider two vertices, \( v \) and \( u \), with different labels but the same expected degree. Suppose that \( u \) and \( v \) each appear in a gadget containing \( n \) other vertices. In the first gadget, \( v \) forms real edges with other vertices with probability \( 2/n \), while \( u \) does so with probability \( 1/n \). In the second gadget, \( v \)'s probability remains \( 2/n \), whereas \( u \)'s increases to \( 3/n \). Thus, both \( v \) and \( u \) have an expected degree of 4. But the variance of the degree for \( v \) and \( u \) differs. More specifically, we have $\Var(\deg(v))= 4-8/n$ and  $\Var(\deg(u))= 4-10/n$. The algorithm can use this variance to learn about the structure of the input and determine whether the instance is from \yesdist{} or \nodist{}.

\subsection{Our Actual Construction via {\em Parallel} Pseudo Edges}

To address the issue with variance, we use the following approach to bound the total variation distance between different degree distributions. Fix a pair of vertices \( u \) and \( v \). Let \( \rho \) denote the probability that this pair can form a real edge according to the gadgets used in the construction. We add \( \rho n \) parallel edges between \( (u, v) \) in the graph, which we refer to as {\em ground edges}. For each ground edge, we independently flip a coin with probability \( 1/n \) to determine whether it becomes a real edge. Note that for this fixed pair, the expected number of real edges between them is exactly \( \rho \), aligning with our intended construction.  If \( u \) and \( v \) have the same total expected degree, then the degree distribution is identical for all vertices. More formally, their degrees follow the same number of Bernoulli random variables (since their expected degrees are equal), each with probability \( 1/n \).

It is important to emphasize that our actual multigraph construction is significantly more intricate than the simplified version presented here, as it must incorporate pseudo edges and satisfy several additional properties. However, to convey the core idea while avoiding unnecessary technical details, we have chosen to present a simplified version. In the actual construction, pseudo edges are a subset of ground edges, and real edges are a subset of pseudo edges (See \Cref{fig:overview-ground}).

\begin{figure}[h]
    \centering
    \newcommand{\edgeprob}{0.55} % Set the probability of adding an edge for the second graph (0.7 for 70%)
\newcommand{\edgeprobsub}{0.2} % Set the probability of adding an edge for the third graph (0.6 for 60%)
\newcommand{\middlebonus}{0.1} % Additional probability bonus to increase edge probability for straight edges
\newcommand{\myseed}{139}
\newcommand{\mybend}{7}

\scalebox{0.75}{
\begin{tikzpicture}
    % Define the six vertices in a hexagonal shape (first graph)
    \foreach \i in {0,1,2,3,4,5} {
        \node[circle, fill=black, inner sep=2pt] (v\i) at ({60*\i}:2) {};
    }
    
    % Draw triple edges between all pairs of vertices with slight curves (first graph)
    \foreach \i in {0,1,2,3,4,5} {
        \foreach \j in {0,1,2,3,4,5} {
            \ifnum\i<\j
                % Always draw the straight edge
                \draw (v\i) -- (v\j);
                
                % Always draw the bent left edge
                \draw[bend left=\mybend] (v\i) to (v\j);
                
                % Always draw the bent right edge
                \draw[bend right=\mybend] (v\i) to (v\j);
            \fi
        }
    }
    
    % Define the six vertices in a hexagonal shape (second graph, shifted right, no edges)
    \foreach \i in {0,1,2,3,4,5} {
        \node[circle, fill=black, inner sep=2pt] (w\i) at ([shift={(5,0)}] {60*\i}:2) {};
    }

    % Add random \edgeprob of the edges from the first graph to the second graph
    \foreach \i in {0,1,2,3,4,5} {
        \foreach \j in {0,1,2,3,4,5} {
            \ifnum\i<\j
                \pgfmathsetseed{\i*\myseed+\j}
                \pgfmathsetmacro{\rand}{rnd}
                
                \ifdim \rand pt < \edgeprob pt
                    \draw (w\i) -- (w\j);
                \fi
                
                \pgfmathsetmacro{\rand}{rnd}
                \ifdim \rand pt < \edgeprob pt
                    \draw[bend left=\mybend] (w\i) to (w\j);
                \fi
                
                \pgfmathsetmacro{\rand}{rnd}
                \ifdim \rand pt < \edgeprob pt
                    \draw[bend right=\mybend] (w\i) to (w\j);
                \fi
            \fi
        }
    }

    % Define the six vertices in a hexagonal shape (third graph, shifted further right, no edges)
    \foreach \i in {0,1,2,3,4,5} {
        \node[circle, fill=black, inner sep=2pt] (u\i) at ([shift={(10,0)}] {60*\i}:2) {};
    }
    
    % Subsample the edges from the second graph
    \foreach \i in {0,1,2,3,4,5} {
        \foreach \j in {0,1,2,3,4,5} {
            \ifnum\i<\j
                \pgfmathsetseed{\i*\myseed+\j}
                \pgfmathsetmacro{\rand}{rnd}
                \pgfmathsetmacro{\probability}{\edgeprobsub + \middlebonus}
                \ifdim \rand pt < \probability pt
                    \draw (u\i) -- (u\j);
                \fi
                
                \pgfmathsetmacro{\rand}{rnd}
                \ifdim \rand pt < \edgeprobsub pt
                    \draw[bend left=\mybend] (u\i) to (u\j);
                \fi
                
                \pgfmathsetmacro{\rand}{rnd}
                \ifdim \rand pt < \edgeprobsub pt
                    \draw[bend right=\mybend] (u\i) to (u\j);
                \fi
            \fi
        }
    }
    
    % Labels
    \node at (0,-2.5) {Ground edges};
    \node at (5,-2.5) {Pseudo edges};
    \node at (10,-2.5) {Real edges};

\end{tikzpicture}
}
    \caption{An illustration of ground, pseudo, and real edges. As shown in the figure, pseudo edges form a subset of ground edges, while real edges are a subset of pseudo edges.}
    \label{fig:overview-ground}
\end{figure}
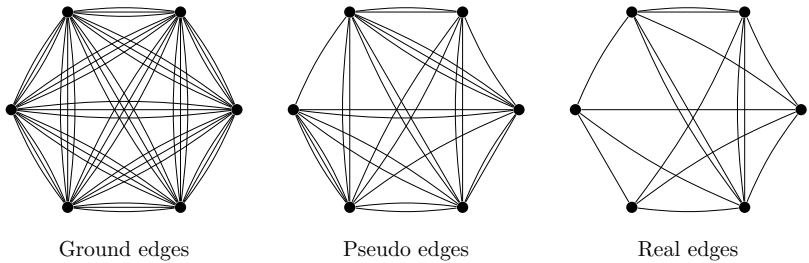

\vspace{-0.5em}
\paragraph{Algorithm cannot distinguish between multigraph and simple graph.} One caveat of using this multigraph approach is that there may be multiple real edges between a pair of vertices. Since the input to our problem is a simple graph, we need to show that, with high probability, the algorithm cannot identify pairs with multiple real edges.  Suppose the algorithm runs in \( O(n^{2-\delta}) \) time. One key observation is that our construction is sparse, with vertex degrees of roughly \( n^\sigma \) (where \(2 \sigma < \delta \)), which implies $\rho = n^{\sigma-1}$. Since there are \( n^2 \) vertex pairs in the graph and at most \( \rho n\) ground edges per pair, the total number of ground edges is at most \( \rho n^3 \). Each of these ground edges is realized as a real edge with probability \( 1/n \), implying that the total number of real edges is at most \( O(\rho n^2) = O(n^{1+\sigma}\)). We condition on the high probability event that this bound holds.  Let \( x = O(n^{1+\sigma} )\) denote the number of real edges in the graph. Using a birthday paradox argument, we can show that the number of vertex pairs with more than one real edge is approximately \( x^2/n^2 = O(n^{2\sigma})\). Since the algorithm makes at most \( O(n^{2-\delta}) \) queries and \(2 \sigma < \delta \), it cannot, with probability $1-o(1)$, query any pair that has multiple real edges. Therefore, with high probability, the algorithm cannot distinguish whether the input graph originates from this multigraph construction or from a simple graph.

\vspace{-0.5em}
\paragraph{Modifications to the previous techniques to account for pseudo edges.} Finally, it is important to highlight additional technical challenges that arise due to the introduction of pseudo edges in the construction. Since non-edges marked as pseudo edges are not independent of the labeling, we need to adjust the coupling argument to properly account for them.  Furthermore, since we aim to demonstrate that the algorithm's advantage in detecting cycles diminishes as it progresses deeper into the construction, we must show that deeper gadgets have lower density compared to higher-level gadgets even when we have pseudo edges. To achieve this, we assign different probabilities for marking pseudo edges to belong to different levels, using smaller probabilities for deeper levels. This adjustment helps reinforce the argument that the algorithm's advantage decreases as it explores deeper into the structure.

\section{Preliminaries}\label{sec:preliminaries}

\begin{proposition}[\cite{erdHos1964random, erdHos1966existence, friezepittel}]\label{prop:matching-ERgraph}
    Let $G$ be a bipartite \ER{} graph such that each part of the graph contains $n$ vertices. If the probability of the existence of each edge is at least $(\log^2 n)/n$, then there exists a perfect matching in $G$ with high probability.
\end{proposition}

\begin{proposition}[Yao's Lemma. Theorem 3 of \cite{Yao77}] \label{prop:yao}
Let $P$ be a problem over input domain $X$.
For an input distribution $\mc{D}$ over $X$, let $B_{\mc{D}, \lambda}$ be the set of deterministic algorithms that solve $P$ with probability $1 - \lambda$. Let $F_{1, \lambda}(P)$ be the distributional complexity of $P$ with parameter $\lambda$:
$$
F_{1, \lambda}(P) := \sup_{\mc{D}} \inf_{a \in B_{\mc{D}, \lambda}} c(a, \mc{D}),
$$
where $c(a, D)$ denotes the average cost of $a$ on $\mc{D}$.
Let $R_\lambda$ be the set of randomized algorithms that solve $P$ with probability $1 - \lambda$ for any input $x$. Let $F_{2, \lambda}(P)$ be the randomized complexity of $P$ with parameter $\lambda$:
$$
F_{2, \lambda}(P) := \inf_{A \in R_{\lambda}} \max_{x \in X} c(A, x),
$$
where $c(A, x)$ is the average cost of $A$ for input $x$.
Then, for $0 \leq \lambda \leq \frac{1}{2}$, it holds
$$
F_{2,\lambda}(P) \geq \frac{1}{2} F_{1, 2\lambda}(P).
$$    
\end{proposition}

\begin{proposition}[Chernoff Bound]\label{prop:chernoff}
    Consider independent Bernoulli random variables $X_1, X_2, \dots, X_n$, and define their sum as  $X = \sum_{i=1}^{n} X_i$. For any $k > 0$, the probability that $X$ deviates from its expectation can be bounded by 
    \[
        \Pr\big(|X - \mathbb{E}[X]| \geq k\big) \leq 2 \exp \left(-\frac{k^2}{3\mathbb{E}[X]}\right).
    \]
\end{proposition}

We also provide a list of parameters used in the proof in the \Cref{sec:tableofparameters}.

\section{Reduction to Earth Mover's Distance}\label{sec:EMD}

In this section, we reduce the estimation of maximum matching size to the estimation of Earth Mover's Distance (EMD, a.k.a.\ Optimal Transport, Wasserstein-1 Distance or Kantorovich–Rubinstein Distance), and obtain the same query lower bound for $\Omega(\epsilon)$-additive approximation.
This lower bound implies that the algorithm of \citet*{BerettaR24} is tight up to the $O_{\epsilon}(1)$ factor in the power, i.e.\ for any $\epsilon > 0$ they present a $n^{2-\Omega_\epsilon(1)}$-time algorithm that approximates EMD up to an $\epsilon$-additive error.

\newcommand{\EMD}{\textnormal{EMD}}

\begin{definition}[EMD]\label{def:EMD}
    Given two distributions $p$ and $q$ over a metric space $(\mc{M}, d)$, their EMD is defined as
    $$
    \EMD(p, q) = \min\left\{ \E_{x,y \sim \varphi}[d(x, y)] \mid \textnormal{$\varphi$ is a coupling of $p$ and $q$}\right\}.
    $$
\end{definition}

For this section, we normalize the distances and assume they lie in $[0, 1]$ for all pairs.
The lower bound is formalized as follows.

\begin{theorem}\label{thm:EMD}
    Let $p$ and $q$ be discrete distributions of support size $n$ on metric $(\mc{M}, d)$, such that $d: \mc{M}^2 \to [0, 1]$.
    For every $\delta > 0$ there exists $\epsilon > 0$ (i.e., $\epsilon$ is only a function of $\delta$),
    such that any algorithm with query access to $d$,
    that computes an $\epsilon$-additive approximation of $\EMD(p, q)$, requires at least $\Omega(n^{2-\delta})$ queries.
\end{theorem}

\begin{remark}
    The underlying metric in our lower bound is a $(1, 2)$-metric.
\end{remark}

Note that the lower bound on the number of queries to $d$ also provides a lower bound on the time complexity. The theorem follows directly from \cref{thm:main-intro} and the reduction below.

\begin{claim}
    Assume there exists an algorithm $\mc{A}$ that, given two distributions of support size $n$ on a normalized metric space, computes an $\epsilon/2$-additive approximation to $\EMD$ using $Q$ queries.
    Then, there exists an algorithm $\mc{A}'$ that, given a bipartite graph $G$, computes a $\epsilon n$-additive approximation to the maximum matching size using the same number of queries to the adjacency matrix.
\end{claim} 
\begin{proof}
    We define algorithm $\mc{A}'$.
    Let $A$ and $B$ be the two parts of $G$, and let $\mc{M} = A \cup B$.
    Assume, without loss of generality, that $\card{A} = \card{B} = n$.
    Let $p$ and $q$ be uniform distributions on $A$ and $B$ respectively, i.e.\ $p(u) = 1/n$ for $u \in A$ and $q(u) = 1/n$ for $u \in B$.
    Finally, let 
    $$
    d(u, v) = \begin{cases}
        1/2, & \text{if $(u, v) \in E(G)$, and} \\
        1, & \text{otherwise.}
    \end{cases}
    $$
    Note that other distances are not relevant to $\EMD(p, q)$,
    and any distance query can be answered using an adjacency matrix query between the same pair.
    Algorithm $\mc{A'}$ simply reports $2n - 2n \mc{A}(\mc{M}, d)$.

    To prove $\mc{A'}$ is computing an $\epsilon n$-additive approximation of $\mu(G)$, it suffices to show 
    $$
    \EMD(p, q) = \frac{2n - \mu(G)}{2n}.
    $$
    First, observe that a coupling $\varphi$ of $p$ and $q$, when scaled up by a factor of $n$, is a perfect fractional matching of $A$ and $B$.
    That is, after weighting the edges according to $d$, the minimum-weight perfect fractional matching has weight $n \cdot \EMD(p, q)$.
    Since the graph is bipartite, this is the same as the minimum-weight perfect (integral) matching.
    It remains to show that the minimum-weight perfect matching has weight $n - \mu(G)/2$.
    Any perfect matching that includes $x$ edges of $G$, must include $n - x$ edges outside $G$, and thus has weight
    $$
    x \cdot \frac{1}{2} + (n - x) \cdot 1 = n - \frac{x}{2}.
    $$
    Finally, note that the largest possible value for $x$ is $\mu(G)$.
    Therefore, the minimum weight perfect matching has weight $n - \mu(G)/2$, and
    \begin{equation*}
        \EMD(p, q) = \frac{n - \mu(G)/2}{n} = \frac{2n - \mu(G)}{2n}. \qedhere
    \end{equation*}
\end{proof}

\section{The Construction}

We define our construction in several steps. First, we introduce a recursive procedure to generate the gadgets used in the construction, as described in \Cref{sec:two-input,sec:recursive-struc,sec:level-1,sec:level-ell}. Next, in \Cref{sec:multigraph-cons}, we demonstrate how to construct a multigraph using these gadgets. We then establish key properties of the multigraph in \Cref{sec:properties-multi}. Following this, in \Cref{sec:final-simple-graph}, we explain how to derive the final simple graph from the constructed multigraph.

\subsection{Two Input Distributions}\label{sec:two-input}

To prove the lower bound, we construct a distribution of graphs and show that any \emph{deterministic} sublinear algorithm that computes an $\epsilon n$-additive approximation of the matching size on this distribution, requires $\Omega(n^{2-\delta})$ adjacency matrix queries. Then, it follows from Yao's min-max theorem that \emph{any (possibly randomized)} algorithm that approximates the matching size up to an $\epsilon n$ additive error requires the same number of queries.

The final distribution $\mc{D}$ is a mix of two input distributions $\mc{D} := (\yesdist + \nodist)/2$, where:
\begin{enumerate}
    \item a graph drawn from \yesdist{} contains a perfect matching with high probability, and
    \item a graph drawn from \nodist{} leaves $\Omega(\epsilon n)$ vertices unmatched.
\end{enumerate}

Observe that any algorithm that computes an $O(\epsilon n)$-additive approximation of the matching size with constant probability must be able to distinguish between \yesdist{} and \nodist{} with constant probability. Using this observation, we prove the following lemma.

\begin{restatable}{lemma}{deterministicLowerbound}\label{lem:deterministic-lower-bound}
    For every $\delta > 0$ there exists $\epsilon > 0$ (i.e., $\epsilon$ is only a function of $\delta$), take a \emph{deterministic} sublinear algorithm that given a graph $G$ drawn from $\mc{D}$, produces an $\epsilon n$-additive approximation of the matching size with probability $\frac{2}{3}$.
    That is, the algorithm computes $\tilde{\mu}(G)$ such that
    $$
    \Pr[\mu(G) - \epsilon n \leq \tilde{\mu}(G) \leq \mu(G)] \geq \frac{2}{3}
    $$
    where the probability is over the graph $G$ drawn from $\mc{D}$.
    Then, the algorithm requires $\Omega(n^{2-\delta})$ adjacency matrix queries.
\end{restatable}

Our main theorem then follows from a direct application of Yao's Lemma (\cref{prop:yao}).

\mainTheorem*
\begin{proof}
    In the notation of \cref{prop:yao}, let the cost of an algorithm be the number of adjacency matrix queries. \cref{lem:deterministic-lower-bound} states $F_{1, \frac{1}{3}} \geq \inf_{a \in B_{\mc{D}, \lambda}} c(a, \mc{D}) = \Omega(n^{2-\delta})$. Then, from \cref{prop:yao} we can conclude:
    $$
    F_{2,\frac{1}{3}} \geq \Omega\left(F_{2,\frac{1}{6}}\right) \geq \Omega\left( F_{1,\frac{1}{3}} \right) = \Omega(n^{2-\delta}).
    $$
    That is, for any randomized algorithm that computes an $\epsilon n$-additive approximation of the matching size with probability $\frac{2}{3}$, there exists an input graph $G$ such that the average (hence also the worst case) number of queries is $\Omega(n^{2-\delta})$
\end{proof}

\subsection{Recursive Structure of Graphs Drawn from Input Distribution}\label{sec:recursive-struc}

Here, we describe the recursive structure of the graphs drawn from $\mc{D}$. For now, we present the structure as gadgets between vertices. Later in \cref{sec:multigraph-cons}, we show precisely how these gadgets are used to construct the multigraph, which is then turned into the final simple graph.

In the recursive structure,
to obtain a level $\ell$ graph, we use graphs of level $\ell-1$ combined with some other gadgets that make it hard for the algorithm to distinguish edges of graphs of level $\ell-1$. 
The ultimate goal is to hide some important edges of level 1 that make a difference between the \yesdist{} and \nodist{}.
We use $\yesdist^\ell$ and $\nodist^\ell$ for the distributions of level $i$ graphs for \yesdist{} and \nodist{}, respectively. Also, we let $\mc{D}^\ell := (\yesdist^\ell + \nodist^\ell)/2$.

\subsection{Gadgets for Level 1}\label{sec:level-1}

In this section, we introduce gadgets and define the first level, i.e.\ the base case, of the recursive structure.
Each gadget is defined between two disjoint vertex sets $X$ and $Y$ by a parameter $P(X, Y)$ that controls the density of the edges.
Intuitively, the gadget mimics a bipartite \ER{} graph between $X$ and $Y$ where there exists an edge between each pair with probability $P(X, Y)$.
Based on this intuition, we define $d_\phi(u)$, the contribution of a gadget $\phi$ to the expected degree of a vertex $u$:
\[
d_\phi(u) = \begin{cases}
    P(X, Y) \card{Y}, & \qquad \textnormal{if $u \in X$,} \\
    P(X, Y) \card{X}, & \qquad \textnormal{if $u \in Y$, and} \\
    0, & \qquad \textnormal{otherwise.} \\
\end{cases}
\]
Here, $\phi$ is a gadget between $X$ and $Y$ with density parameter $P(X, Y)$.
For a set of gadgets $\Phi$, the contribution is simply defined as $d_\Phi(u) = \sum_{\phi \in \Phi} d_\phi(u)$.

The actual construction, however, is more intricate as described in \cref{sec:multigraph-cons}. There, a multigraph is constructed such that for a pair of vertices $u \in X$ and $v \in Y$ the expected number of edges between them is $P(X, Y)$. 
We use the following remark to prove some properties of the multigraph and postpone the proof to \cref{sec:properties-multi} to abstract away the details.

\begin{remark} \label{remark:gadgets-are-like-ER}
    Take two vertices $u$ and $v$ such that there is a gadget with density parameter $p$ between them. Then in the multigraph, the expected number of edges between $u$ and $v$ is $p$ (\Cref{cor:single-edge-expectation}), and there is at least one edge between them with probability $p/2$ (\Cref{cor:single-edge-prob}).
    If there are no gadgets between $u$ and $v$, then there are no edges between them.
\end{remark}

The base case of the construction consists of $6r + 2$ vertex sets which are divided into $r$ layers.
For a vertex $u$ in a set $X$, we may refer to $X$ as the label of $u$ on this level.
The sizes of these sets are fixed, but the vertices are otherwise divided between them at random.
The vertex sets are defined as follows:

\begin{itemize}
    \item $A^j_i$ for $i \in [r]$ and $j \in \{1, 2\}$. For $i < r$ and $j \in \{1, 2\}$, each subset $A^j_i$ contains $N_1$ vertices. Each of $A^1_r$ and $A^2_r$ contains $(1-\xi)N_1$ vertices.
    \item $B^j_i$ for $i \in [r]$ and $j \in \{1, 2\}$, each containing $N_1$ vertices. 
    \item $D^j_i$ \emph{(the dummy vertices)} for $i \in [r]$ and $j\in \{1, 2\}$, each containing $\zeta N_1$ vertices. 
    \item $S^j$ \emph{(the special vertices)} for $j \in \{1, 2\}$, each containing $N_1$ vertices.
\end{itemize}

Throughout the paper, we use $S$ to denote $S^1 \cup S^2$. Likewise, we use $A_i$, $B_i$, and $D_i$ to denote $A_i^1 \cup A_i^2$, $B_i^1 \cup B_i^2$, and $D_i^1 \cup D_i^2$, respectively. Also, we let $n_1$ be the total number of vertices in a level-$1$ graph.

Now, we present the gadgets that are the same for the $\yesdist^1$ and $\nodist^1$ distributions.
It consists of (1) sparse gadgets between $S^j$ and $B^j_1$, (2) a dense gadget between $A_i^j$ and $B_i^j$ in each layer $i < r$, (3) a sparse gadget between $B^i_j$ and $A_{i-1}^j$ that connects the layers $i \leq r$ and $i - 1$, and (4) the gadgets to the dummy vertices.

\begin{align*}
    &P(S^j, B^j_1) = \frac{\log^2 n}{N_1}    &\forall\quad j \in \{1, 2\},\\
    &P(B^j_i, A^j_i) = \frac{d_1}{N_1}     &\forall\quad j \in \{1, 2\}, \quad 1 \leq i < r,\\
    &P(B_i^j, A_{i-1}^j) = \frac{\log^2 n}{N_1}     &\forall\quad j \in \{1, 2\}, \quad 1 < i \leq r ,\\
\end{align*}

The gadgets to the dummy vertices are more involved.
Each vertex set $A_i^j$ (resp.\ $B_i^j$ and $D_i^j$) on layer $i$
has gadgets to the dummies $D_k^{3-j}$ (resp.\ $D_k^j$ and $D_k^{j-3}$) for all layers $k \leq i$. There are some gadgets between the dummy vertices of the same layer to control the total degree of each dummy vertex.
$$
\begin{array}{D{l} @{\hspace{1cm}} S{r}}
    P(A_i^j, D_k^{3-j}) = \frac{\gamma d_1}{\zeta N_1}    &\forall\quad j \in \{1, 2\}, \quad 1 < i \leq r, \quad k < i,\\
    P(A_i^j, D_i^{3-j}) = \frac{(r-i+1)\gamma d_1}{\zeta N_1}    &\forall\quad j \in \{1, 2\}, \quad 1 \leq i \leq r,\\
    P(B_i^j, D^j_k) = \frac{\gamma d_1}{\zeta N_1}     &\forall\quad j \in \{1, 2\}, \quad 1 < i \leq r, \quad k < i, \\
    P(B_i^j, D^j_i) = \frac{(r-i+1)\gamma d_1}{\zeta N_1}     &\forall\quad j \in \{1, 2\}, \quad 1 \leq i \leq r,\\
    P(D^j_i, D^{3-j}_k) = \frac{\gamma d_1}{\zeta N_1}     &\forall\quad 1\leq i \leq r, \quad 1\leq k \leq r, \quad i \neq k, \quad j \in \{1,2 \},\\
    P(D^j_i, D^{3-j}_i) = \frac{d_1 + \gamma d_1 + \log^2 n - (4r - 4i + 2 - \xi)\gamma d_1 / \zeta}{\zeta N_1}     &\forall\quad 1\leq i \leq r, \quad j \in \{1,2 \}.
\end{array}
$$

Finally, we describe the gadgets that are different in $\yesdist^1$ and $\nodist^1$.
In the \yes{} case, there is a sparse gadget between $A_r^1$ and $A_r^2$.
In the \no{} case, there is no such gadget. There are gadgets between $A_r^j$ and $B_r^j$, and between $B_r^1$ and $B_r^2$ the parameters of which are changed to control the degrees and make up for the discrepancy between the two cases.
More precisely, we have the following gadgets \textbf{only} in $\yesdist^1$:
\begin{align*}
    &P(A^j_r, B^j_r) = \frac{d_1}{N_1}   \qquad\qquad  &\forall\quad j \in \{1, 2\},\\
    &P(A^1_r, A^2_r) = \frac{\log^2 n}{(1-\xi)N_1},\\
    &P(B^1_r, B^2_r) = \frac{\xi d_1}{N_1},\\
\end{align*}
On the other hand, we have the following gadgets \textbf{only} in $\nodist^1$:
\begin{align*}
    &P(A^j_r, B^j_r) = \frac{d_1 + \log^2 n}{N_1}   \qquad\qquad  &\forall\quad j \in \{1, 2\},\\
    &P(B^1_r, B^2_r) = \frac{\xi d_1 - (1-\xi)\log^2 n}{N_1},\\
\end{align*}

Now, we establish some useful properties of the graphs $\mc{D}^1$. The formal proofs of some of the claims can be found in the appendix, as they follow from simple observations and algebra. 
Firstly, we characterize the expected degree of the vertices:
all special vertices have the same expected degree,
and all the non-special vertices have the same expected degree.
Recall that a gadget between vertex sets $X$ and $Y$ with parameter $p$
contributes $p\card{Y}$ to the expected degree of each vertex in $X$.
The proof is deferred to \cref{sec:omitted-proofs}.

\begin{restatable}{claim}{expectedDegreeBase}
Let $\Phi_\yes^1$ and $\Phi_\no^1$ be the gadgets introduced on level $1$,
and $\Phi^1$ be either one of them.
It holds that,
    \begin{itemize}
        \item[(i)] $d_{\Phi^1}(v) =  \log^2 n$, for $v \in S$
        \item[(ii)] $d_{\Phi^1}(v) = d_1 + r\gamma d_1 + \log^2 n$, for $v \notin S$.
    \end{itemize}
\end{restatable}

\begin{observation}\label{obs:bipartiteness}
    Any multigraph drawn from $\mc{D}^1$ is bipartite.  
\end{observation}
\begin{proof}
    Consider the following partitioning of the vertices to the following two disjoint parts:
    \begin{itemize}
        \item[(1)] $\left(\bigcup_{i=1}^r A^1_i \right) \cup \left(  \bigcup_{i=1}^r B^2_i \right) \cup \left(  \bigcup_{i=1}^r D^1_i \right) \cup S^1$, and
        \item[(2)] $\left(\bigcup_{i=1}^r A^2_i \right) \cup \left(  \bigcup_{i=1}^r B^1_i \right) \cup \left(  \bigcup_{i=1}^r D^2_i \right) \cup S^2$.
    \end{itemize}
    According to the gadgets, there are no edges inside each part, which concludes the proof.  
\end{proof}

\begin{claim}\label{clm:perfect-matching-ind-subgraphs}
    In the multigraph construction, all the following hold with high probability:
    \begin{itemize}
        \item There exists a perfect matching between $A^j_i$ and $B^j_{i+1}$ for all $1 \leq i < r$ and $j \in \{1, 2\}$ in each graph drawn from $\mc{D}^1$,
        \item There exists a perfect matching between $B^j_1$ and $S^j$ for all $j \in \{1, 2\}$ in each graph drawn from $\mc{D}^1$,
        \item There exists a perfect matching between $D^1_i$ and $D^2_i$ for all $1 \leq i \leq r$ in each graph drawn from $\mc{D}^1$,
        \item There exists a perfect matching between $A^1_r$ and $A^2_r$ in each graph drawn from $\yesdist^1$.
    \end{itemize}
\end{claim}
\begin{proof}
Consider the induced subgraph between $A_i^j$ and $B_{i+1}^j$. The induced subgraph is an \ER{} graph such that each edge exists with a probability of at least $\Omega\left((\log^2 n)/N_1\right)$ (by \Cref{remark:gadgets-are-like-ER}). Also, we have $ (\log^2 n)/N_1 >(\log^2 N_1)/N_1$. Moreover, we have $|A_i^j| = |B_i^j| = N_1$. Thus, by using \Cref{prop:matching-ERgraph}, there exists a perfect matching in the induced subgraph between $A_i^j$ and $B_{i+1}^j$ with high probability. With the exact similar approach, we can prove that there exists a perfect matching in each of the induced subgraphs in the claim statement with high probability.
\end{proof}

\begin{lemma} \label{lem:base-matching-size}
    In the multigraph construction, the size of the maximum matching in $\yesdist^1$ and $\nodist^1$ is characterized as follows:
    \begin{itemize}
        \item[(i)] If $G \sim \yesdist^1$, then $\mu(G) = n_1 / 2$ with high probability.
        \item[(ii)] If $G \sim \nodist^1$, then $\mu(G) \leq n_1 / 2 - N_1/2$.
    \end{itemize}
\end{lemma}
\begin{proof}
    Let us condition on the high probability event of \Cref{clm:perfect-matching-ind-subgraphs}. Now suppose that $G$ is drawn from $\yesdist^i$. Therefore,
    \begin{align*}
        \mu(G) &\geq \mu(A_r^1, A_r^2) + \sum_{j\in\{1,2\}} \mu(S^j, B_1^j) + \sum_{j\in\{1,2\}, i < r} \mu(A_i^j, B_{i+1}^j) + \sum_{i \leq r} \mu(D_i^1, D_i^2)\\
        & = (1-\xi)N_1 + rN_1(2+\zeta) = n_1 / 2.
    \end{align*}
    Also, the total number of vertices in the graph is $n_1$, thus the inequality is tight, which completes the proof of (i).

    Now suppose that $G$ is drawn from $\nodist^i$. Note that vertices of
    \begin{align*}
        \left(\bigcup_{i\leq r, j \in \{1,2\}} B_i^j \right) \cup \left( \bigcup_{i\leq r, j \in \{1,2\}} D_i^j\right),
    \end{align*}
    form a vertex cover of $G$. 
    Observe that the size of any vertex cover must be larger than the size of the maximum matching because for each edge in the maximum matching, 
    at least one of the endpoints must be in the vertex cover.
    Therefore, 
    \begin{align*}
         \mu(G)  \leq \sum_{i \leq r, j\in \{1,2\}}^r |B^j_i| + \sum_{i\leq r, j\in \{1,2\}}^r |D^j_i| = 2rN_1 + 2r\zeta N_1 = \frac{n_1}{2} - (1-\xi - r\zeta )N_1  < \frac{n_1 - N_1}{2},
    \end{align*}
    which completes the proof of (ii).
\end{proof}

\subsection{Gadgets for Level $\ell$}\label{sec:level-ell}

Now, we describe the higher levels of the recursive structure.
For a level $1 < \ell \leq L$, constructing a graph of $\mc{D}^\ell$ involves drawing multiple instances from $\mc{D}^{\ell - 1}$.
Let $N_\ell = n_{\ell-1}/(2\zeta)$ be a parameter that controls the number of vertices in level $\ell$ of the construction, and $d_\ell$ be a parameter that controls the degree of vertices in level $\ell$ of the construction. Recall $n_{\ell - 1}$ is the total number of vertices for a graph of level $\ell - 1$.

Each graph in $\mc{D}^\ell$ consists of $8r + 2$ vertex sets which are divided into $r$ layers.
For a vertex $u$ in a set $X$, we may refer to $X$ as the label of $u$ on this level.
The sizes of these sets are fixed, but the vertices are otherwise divided between them at random.
We have the following disjoint subset of vertices:
\begin{itemize}
    \item $A^j_i$ for $i \in [r]$ and $j \in \{1, 2\}$, each containing $N_\ell$ vertices. 
    \item $B^j_i$ for $i \in [r]$ and $j \in \{1, 2\}$ each containing $N_\ell$ vertices. 
    \item $D^j_i$ \emph{(the dummy vertices)} for $i \in [r]$ and $1\leq j 
 \leq 4$, each containing $\zeta N_\ell$ vertices. 
    \item $S^j$ \emph{(the special vertices)} for $j \in \{1, 2\}$ each containing $N_\ell$ vertices.
\end{itemize}

We define $A_i$ equal to $A_i^1 \cup A_i^2$, and similarly, $B_i = B_i^1 \cup B_i^2$, $S = S^1 \cup S^2$, and $D_i = D_i^1 \cup D_i^2 \cup D_i^3 \cup D_i^4$.

The structure of the edges is somewhat similar to the base case.
Aside from the gadgets to the dummies, the main difference is that the sparse gadgets are replaced with instances of $\yesdist^{\ell-1}$ and $\nodist^{\ell - 1}$.
For two vertex sets $X$ and $Y$ of size $N_{\ell} = n_{\ell - 1}/\zeta$,
we use \enquote{drawing $1/\zeta$ disjoint graphs from $\yesdist^{\ell - 1}$ between $X$ and $Y$} to refer to the following procedure:
(1) divide the vertices of $X$ into $1/\zeta$ sets $X_1, \ldots, X_{1/\zeta}$ of equal size at random, similarly for $Y$, 
(2) for each $i$, draw a graph from $\yesdist^{\ell - 1}$ and randomly map the vertices of one part to $X_i$ and the vertices of the other part to $Y_i$, and
(3) construct the corresponding gadgets between $X_i$ and $Y_i$ according to the graph drawn from $\yesdist^{\ell - 1}$

First, we describe the parts that are the same in $\yesdist^\ell$ and $\nodist^\ell$.
We recursively define the following parts:
\begin{itemize}
    \item draw $1/\zeta$ disjoint graphs from $\yesdist^{\ell-1}$ between $A_i^j$ and $B_{i+1}^j$ for each $1\leq i < r$ and $j \in \{1,2\}$.
    \item draw $1/\zeta$ disjoint graphs from $\yesdist^{\ell-1}$ between $B_1^j$ and $S^j$ for each $j \in \{1,2\}$.
    \item draw one graph from $\yesdist^{\ell-1}$ between $D_i^j$ and $D_i^{5-j}$ for each $1\leq i \leq r$ and $j \in \{1, 2\}$.
\end{itemize}

There are dense gadgets between $A_i^j$ and $B_i^j$ for every layer $i \in [r]$.
\begin{align*}
    &P(B^j_i, A^j_i) = \frac{d_\ell}{N_\ell}     &\qquad\forall\quad j \in \{1, 2\}, \quad 1 \leq i \leq r
\end{align*}

We also define the following gadgets to the dummy vertices.
Each vertex set $A_i^j$ or $B_i^j$ has gadgets to the dummies of levels up to $i$,
and there are some gadgets between the dummy vertices.
$$
\begin{array}{D{l} S{r}}
    P(B_i^j, D^{j'}_k) = \frac{\gamma d_\ell}{2\zeta N_\ell}     &\forall\quad j \in \{1, 2\}, \quad j' \in \{j, j+2\}, \quad 1 < i \leq r, \quad k < i, \\
    P(B_i^j, D^{j'}_i) = \frac{(r-i+1)\gamma d_\ell}{2\zeta N_\ell}     &\forall\quad j \in \{1, 2\}, \quad j' \in \{j, j+2\}, \quad 1 \leq i \leq r,\\
    P(A_i^j, D_k^{j'}) = \frac{\gamma d_\ell}{2\zeta N_\ell}    &\forall\quad j \in \{1, 2\},  \quad j' \in \{3-j, 5-j\}, \quad 1 < i \leq r, \quad k < i,\\
    P(A_i^j, D_i^{j'}) = \frac{(r-i+1)\gamma d_\ell}{2\zeta N_\ell}    &\forall\quad j \in \{1, 2\}, \quad j' \in \{3-j, 5-j\}, \quad 1 \leq i \leq r,\\
    P(D^j_i, D^{j'}_k) = \frac{\gamma d_\ell}{2\zeta N_\ell}     & \forall\quad 1\leq i \leq r, \quad 1\leq k \leq r, \quad i \neq k, \quad j \in \{1,3 \}, \quad j' \in \{2,4 \}\\
    P(D^j_i, D^{j+1}_i) = \frac{d_\ell + \gamma d_\ell - (2r - 2i + 1)\gamma d_\ell / \zeta}{\zeta N_\ell}     &\forall\quad 1\leq i \leq r, \quad j \in \{1,3 \}
\end{array}
$$

Finally, we present the gadgets that are different in $\yesdist^\ell$ and $\nodist^\ell$. We have the following gadget \textbf{only} in $\yesdist^\ell$:
\begin{itemize}
    \item $1/\zeta$ disjoint graphs drawn from $\yesdist^{\ell-1}$ between $A_r^1$ and $A_r^2$.
\end{itemize}
On the other hand, we have the following gadget \textbf{only} in $\nodist^\ell$:
\begin{itemize}
    \item $1/\zeta$ disjoint graphs drawn from $\nodist^{\ell-1}$ between $A_r^1$ and $A_r^2$.
\end{itemize}

\begin{restatable}{claim}{expectedDegreeHigherLevels}
    Let $\Phi_\yes^\ell$ and $\Phi_\no^\ell$ be the gadgets introduced on level $\ell$,
    and $\Phi^\ell$ be either one of them.
    \begin{itemize}
        \item[(i)] $d_{\Phi^\ell}(v) =  0$, for $v \in S$,
        \item[(ii)] $d_{\Phi^\ell}(v) = d_\ell + r\gamma d_\ell$, for $v \notin S$.
    \end{itemize}
\end{restatable}

\begin{observation}\label{obs:bipartiteness-ell}
   Any multigraph drawn from $\mc{D}^\ell$ is bipartite.  
\end{observation}
\begin{proof}
    Observe that all the recursively produced parts are bipartite.
    Therefore, the following is a bipartition of the graph:
    \begin{itemize}
        \item[(1)] $\left(\bigcup_{i=1}^r A^1_i \right) \cup \left(  \bigcup_{i=1}^r B^2_i \right) \cup \left(  \bigcup_{i=1}^r D^1_i \right) \cup \left(  \bigcup_{i=1}^r D^3_i \right) \cup S^1$, and
        \item[(2)] $\left(\bigcup_{i=1}^r A^2_i \right) \cup \left(  \bigcup_{i=1}^r B^1_i \right) \cup \left(  \bigcup_{i=1}^r D^2_i \right) \cup \left(  \bigcup_{i=1}^r D^4_i \right) \cup S^2$. \qedhere
    \end{itemize}
\end{proof}

\begin{lemma}\label{lem:matching-level-ell}
    In the multigraph construction, the size of the maximum matching in $\yesdist^\ell$ and $\nodist^\ell$ is characterized as follows:
    \begin{itemize}
        \item[(i)] For $G \sim \yesdist^\ell$, it holds $\mu(G) = n_\ell / 2$ with high probability.
        \item[(ii)] If $G \sim \nodist^\ell$, then $\mu(G) \leq n_\ell / 2 - N_1/2$.
    \end{itemize}
\end{lemma}
\begin{proof}
    For $(i)$, we prove $G \sim \yesdist^\ell$ has a perfect matching with high probability by induction on $\ell$.
    The base case, $\ell = 1$, has been proven in \cref{lem:base-matching-size}.
    By the induction hypothesis, there exists a perfect matching in each recursively drawn graph from $\yesdist^{\ell - 1}$ with high probability.
    Therefore, by the union bound, there exists a perfect matching 
    between $B_1^j$ and $S^j$ for $j \in \{1, 2\}$;
    between $A_i^j$ and $B_{i+1}^j$ for $1 \leq i < r$ and $j \in \{1, 2\}$;
    between $A_r^1$ and $A_r^2$; and
    between $D_i^j$ and $D_i^{5 - j}$ for $j \in \{1, 2\}$.
    Putting these matchings together results in a perfect matching for the whole graph.
    Observe that using the union bound in the induction is valid because over the course of the induction up to $\ell = L$, we recursively take the union bound over $\left(\frac{8r + 2}{\zeta}\right)^L = O_\epsilon(1)$ events.

    For $(ii)$, we prove by induction on $\ell$ that $G \sim \nodist^\ell$ has a vertex cover of size $n_\ell/2 - N_1/2$.
    The base case, $\ell = 1$, has been proven in \cref{lem:base-matching-size}.
    Assuming the induction hypothesis for $\ell - 1$.
    There exists a vertex cover $X$ of size $\frac{1}{\zeta}(n_{\ell - 1} - N_1/2)$ for the induced graph between $A_r^1$ and $A_r^2$.
    The following set is a vertex cover:
    $$
    C := X \cup \left(\bigcup_{i\leq r, j \in \{1,2\}} B_i^j \right) \cup \left( \bigcup_{i\leq r, j \in \{1,2, 3, 4\}} D_i^j\right),
    $$
    as any edge outside of $G[A_r^1; A_r^2]$ is adjacent to $B$ or $D$.
    Now, we calculate the size of this vertex cover: 
    \begin{align*}
    \card{C} 
    &= \card{X} \cup \sum_{i\leq r, j \in \{1,2\}} \card{B_i^j} + \sum_{i\leq r, j \in \{1, 2, 3, 4\}} \card{D_i^j} \\
    &= \frac{1}{\zeta}(n_{\ell - 1}/2 - N_1/2) + 2rN_{\ell} + 4r\zeta N_\ell \\
    &\leq (2r + 1/2 + 4r\zeta )N_\ell - N_1/2 \\
    &\leq n_\ell/2 - N_1/2 \\
    \end{align*}
    where the last inequality follows from $4r\zeta \leq \frac{1}{2}$ and $n_\ell/2 = (2r + 1 + 2r\zeta)N_\ell$.
    This concludes the proof of $(ii)$ and the claim.
\end{proof}

\begin{lemma}\label{lem:matching-size-final-graph}
    Let $\epsilon = (\delta/10)^{300/\delta^2}$. Any algorithm that estimates the size of the maximum matching of the multigraph that is drawn from the input distribution with $\epsilon n$ additive error must be able to distinguish whether it belongs to \yesdist{} or \nodist{}.
\end{lemma}
\begin{proof}
By \cref{lem:matching-level-ell},
the maximum matching size in \yesdist{} and \nodist{} differ by $N_1/2$.
Therefore, any algorithm that estimates the size of the maximum matching within a $(N_1/4)$-additive error, must distinguish between the two distributions.
Now, we express $N_1$ in terms of $n = n_L$.

For the first level, we have 
$$
n_1 = \card{A} + \card{B} + \card{D} + \card{S} = (2r - 2\xi)N_1 + 2rN_1 + 2r\zeta N_1 + 2N_1 = (4r-2r\xi+2r\zeta + 2)N_1.
$$

For the next levels $\ell > 1$, it holds 
$$
N_\ell = n_{\ell - 1}/\zeta 
\qquad \textnormal{and} \qquad
n_{\ell} 
= 2rN_\ell + 2rN_\ell + 4r\zeta N_\ell + 2N_\ell
= (4r + 4r\zeta + 2)N_\ell
$$
Therefore,
$$
n_{\ell} = \frac{(4r + 4r\zeta + 2)}{\zeta}n_{\ell-1}.
$$
Then, it can be shown by induction
$$
n_L = \left(\frac{(4r + 4r\zeta + 2)}{\zeta}\right)^{L-1}n_1
= \left(\frac{(4r + 4r\zeta + 2)}{\zeta}\right)^{L-1} (4r-2r\xi+2r\zeta + 2)N_1.
$$
Plugging in the parameters 
$$
L = 4/\delta, 
\qquad r = (10/\delta)^{L+1}, 
\qquad \zeta = 1/r^2, \qquad \textnormal{and} \qquad 
\xi = 1/r^4,
$$
we get
$$
n_L \leq (10r^3)^L N_1 \leq (10 / \delta)^{3L^2}N_1 \leq (10/\delta)^{300/\delta^2}N_1.
$$

Putting everything together,
if we let $\epsilon = (\delta/10)^{300/\delta^2}$,
then any algorithm that computes an $\epsilon n$-additive approximation,
computes an $(N_1/4)$-additive approximation,
and thus distinguishes between \yesdist{} and \nodist{}.
\end{proof}

\subsection{Constructing the Multigraph}\label{sec:multigraph-cons}

In this subsection, we demonstrate how to construct the multigraph using the gadgets introduced earlier. We first prove the following claim, which is essential for the construction of the multigraph.

\begin{claim}\label{clm:unique-level-edge}
    For every pair of vertices $u$ and $v$, considering all levels, there exists at most one gadget between vertex sets $X$ and $Y$ such that $u \in X$ and $v \in Y$.
\end{claim}
\begin{proof}
    The claim is proved by induction on $\ell$.
    For the base case $\ell = 1$, the claim holds as the vertices are divided into disjoint sets, each gadget is defined between two of these sets, and there is at most one gadget between any two sets.
    For $\ell > 1$, the claim holds because 
    (1) the gadgets on level $\ell$ are defined between disjoint vertex sets, and
    (2) the recursively defined gadgets of level smaller than $\ell$ are either between $A_i^j$ and $B_{i+1}^j$, $B_1^j$ and $S^j$, $D_i^j$ and $D_i^{5- j}$, or $A_r^1$ and $A_r^2$.
    Therefore, they do not share pairs with each other or with gadgets of level $\ell$.
\end{proof}

Given this fact, for any two vertices $u$ and $v$ we can define a parameter $p^{(u, v)}$, which is equal to the density parameter associated with the gadget between $u$ and $v$, or zero if no such gadget exists.

\subsubsection*{Warm-up: ``level-free'' construction}
We begin with a simplified ``level-free'' construction and then generalize to include levels. 
En route to our multigraph construction, we construct three helper multigraphs, the {\em ground graph}, the {\em labelled-ground graph}, and the {\em pseudo graph}. It will be helpful for the analysis that the ground and pseudo graphs  are independent of the vertex labels, so their structure does not reveal any information about the labels. 
This scheme also helps obtain the same distribution for the degrees of the vertices, so that the algorithm cannot easily distinguish between two vertices, e.g.\ based on the distribution of the degrees of their neighbors.

All multigraphs are defined over the same set of vertices, and their respective edge sets satisfy the following inclusions:
$$
\text{Ground edges} \supset \text{Labelled-ground edges}, \text{pseudo edges} \supset \text{Real edges}.
$$
In particular, we will eventually let the real multigraph simply be the intersection of the pseudo and the labelled-ground edge sets:
$$
\text{Real edges} = \text{Labelled-ground edges} \cap \text{Pseudo edges}.
$$

\paragraph{Ground edges:}
The ground graph is completely deterministic. 
For every pair of vertices $u$ and $v$, we add exactly $\rho n$ parallel ground edges between $u$ and $v$ (independently of their labels). 

\paragraph{Labelled-ground edges:}
The labelled-ground edges are deterministic conditioned on the (random) labels of the vertices.
For any pair of vertices $u$ and $v$,  let $p^{(u,v)}$ be density parameter of the gadget between $u$ and $v$ (and zero if no such gadget exists). We let  $p^{(u,v)} n$ of the  ground edges between $u,v$ be labelled-ground edges. Observe that if two vertices have the same expected degree $d$ in the gadget construction, they have the same degree $d n$ in the labelled-ground graph.

\paragraph{Pseudo edges:}
Every ground edge is a pseudo edge independently with probability $1/n$.

\subsubsection*{Full multigraph construction with levels}
Our actual construction closely follows the level-free construction, with the modification that each labelled-ground and pseudo edge is assigned a level. Then, the level-$\ell$ real edges are defined as:
$$
\text{Level-$\ell$ real edges} = \text{Level-$\ell$ labelled-ground edges} \cap \text{Level-$(\le \ell)$ pseudo edges}.
$$
Note that we require level exactly $\ell$ for labelled-ground edges, but take any level $\le \ell$ for pseudo edges.

The ground edges are not associated with any levels and are defined exactly as before, i.e.\ there is $\rho n$ of them between any pair.
Then, for any two vertices $u$ and $v$, we add $p^{(u,v)} n \cdot (\rho/\rho_{\ell})$ of the ground edges between them to the level-$\ell$ labelled-ground edges, where $\ell$ is the level of the gadget defined between $u$ and $v$ (if any).

\begin{figure}
    \centering
    \newcommand{\edgeprob}{0.55} % Set the probability of adding an edge for the second graph (0.7 for 70%)
\newcommand{\edgeprobsub}{0.2} % Set the probability of adding an edge for the third graph (0.6 for 60%)
\newcommand{\middlebonus}{0.1} % Additional probability bonus to increase edge probability for straight edges
\newcommand{\myseed}{139}
\newcommand{\mybend}{7}

\scalebox{0.75}{
\begin{tikzpicture}
    % Define the six vertices in a hexagonal shape (first graph)
    \foreach \i in {0,1,2,3,4,5} {
        \node[circle, fill=black, inner sep=2pt] (v\i) at ({60*\i}:2) {};
    }
    
    % Draw triple edges between all pairs of vertices with slight curves (first graph)
    \foreach \i in {0,1,2,3,4,5} {
        \foreach \j in {0,1,2,3,4,5} {
            \ifnum\i<\j
                % Always draw the straight edge
                \draw (v\i) -- (v\j);
                
                % Always draw the bent left edge
                \draw[bend left=\mybend] (v\i) to (v\j);
                
                % Always draw the bent right edge
                \draw[bend right=\mybend] (v\i) to (v\j);
            \fi
        }
    }
    
    % Define the six vertices in a hexagonal shape (second graph, shifted right, no edges)
    \foreach \i in {0,1,2,3,4,5} {
        \node[circle, fill=black, inner sep=2pt] (w\i) at ([shift={(5,0)}] {60*\i}:2) {};
    }

    % Add random \edgeprob of the edges from the first graph to the second graph
    \foreach \i in {0,1,2,3,4,5} {
        \foreach \j in {0,1,2,3,4,5} {
            \ifnum\i<\j
                \pgfmathsetseed{\i*\myseed+\j}
                % Randomly decide whether to add the straight edge (with probability \edgeprob)
                \pgfmathsetmacro{\rand}{rnd}
            
                \ifdim \rand pt < \edgeprob pt
                    \draw (w\i) -- (w\j);
                \fi
                
                % Randomly decide whether to add the bent left edge (with probability \edgeprob)
                \pgfmathsetmacro{\rand}{rnd}
                \ifdim \rand pt < \edgeprob pt
                    \draw[bend left=\mybend] (w\i) to (w\j);
                \fi
                
                % Randomly decide whether to add the bent right edge (with probability \edgeprob)
                \pgfmathsetmacro{\rand}{rnd}
                \ifdim \rand pt < \edgeprob pt
                    \draw[bend right=\mybend] (w\i) to (w\j);
                \fi
            \fi
        }
    }

    % Define the six vertices in a hexagonal shape (third graph, shifted further right, no edges)
    \foreach \i in {0,1,2,3,4,5} {
        \node[circle, fill=black, inner sep=2pt] (u\i) at ([shift={(10,0)}] {60*\i}:2) {};
    }
    
    % Subsample the edges from the second graph (only from the existing edges in the second graph)
    \foreach \i in {0,1,2,3,4,5} {
        \foreach \j in {0,1,2,3,4,5} {
            \ifnum\i<\j
                % Randomly decide whether to add edges
                \pgfmathsetseed{\i*\myseed+\j}
                \pgfmathsetmacro{\rand}{rnd}
                \pgfmathsetmacro{\probability}{\edgeprobsub + \middlebonus}
                \ifdim \rand pt < \probability pt
                    \draw (u\i) -- (u\j);
                \fi
                
                \pgfmathsetmacro{\rand}{rnd}
                \ifdim \rand pt < \edgeprobsub pt
                    \draw[bend left=\mybend] (u\i) to (u\j);
                \fi
                
                \pgfmathsetmacro{\rand}{rnd}
                \ifdim \rand pt < \edgeprobsub pt
                    \draw[bend right=\mybend] (u\i) to (u\j);
                \fi
            \fi
        }
    }
    
    % Define the six vertices for the fourth graph
    \foreach \i in {0,1,2,3,4,5} {
        \node[circle, fill=black, inner sep=2pt] (x\i) at ([shift={(10,-6)}] {60*\i}:2) {}; 
    }

    % In the fourth graph, draw smooth cycle around the left three vertices and label as Y with space around the points
    \draw[thick] plot[smooth cycle, tension=1] coordinates {($(x5) + (-0.3,-0.4)$) ($(x0) + (0.3,0)$) ($(x1) + (-0.3,0.4)$)};
    \node at (12.2,-7.5) {Y};
    
    % In the fourth graph Draw smooth cycle around the right three vertices and label as X with space around the points
    \draw[thick] plot[smooth cycle, tension=1] coordinates {($(x2) + (0.3,0.4)$) ($(x3) + (-0.3,0)$) ($(x4) + (0.3,-0.4)$)};
    \node at (7.8,-7.5) {X};

    % In the fourth graph, draw straight edges between left and right vertices
    \foreach \i in {5,0,1} {
        \foreach \j in {2,3,4} {
            \draw (x\i) -- (x\j);
        }
    }

    % Define the six vertices for the fifth graph (Level-(L-1) real edges)
    \foreach \i in {0,1,2,3,4,5} {
        \node[circle, fill=black, inner sep=2pt] (y\i) at ([shift={(15,-3)}] {60*\i}:2) {}; 
    }
    
    % Labels
    \node at (0,-2.5) {Ground edges};
    \node at (5,-2.5) {Level-$L$ Pseudo edges};
    \node at (10,-2.5) {Level-$(L-1)$ Pseudo edges};
    \node at (10,-8.5) {Level-$(L-1)$ Labeled-ground edges};
    \node at (15,-5.5) {Level-$(L-1)$ Real edges};
    
    % X and Y cycles for the fifth graph
    \draw[thick] plot[smooth cycle, tension=1] coordinates {($(y5) + (-0.3,-0.4)$) ($(y0) + (0.3,0)$) ($(y1) + (-0.3,0.4)$)};
    \node at (17,-5) {Y};
    
    \draw[thick] plot[smooth cycle, tension=1] coordinates {($(y2) + (0.3,0.4)$) ($(y3) + (-0.3,0)$) ($(y4) + (0.3,-0.4)$)};
    \node at (13,-5) {X};

    % Draw edges in the fifth graph if they are in both the third and fourth graphs
    \foreach \i in {5,0,1} {
        \foreach \j in {2,3,4} {
            \pgfmathsetseed{\j*\myseed+\i}
            \ifnum\i<\j
                \pgfmathsetseed{\i*\myseed+\j}
            \fi
            \pgfmathsetmacro{\rand}{rnd}
            \pgfmathsetmacro{\probability}{\edgeprobsub + \middlebonus}
            \ifdim \rand pt < \probability pt
                \draw (y\i) -- (y\j);
            \fi
        }
    }

\end{tikzpicture}
}
    \caption{An illustration of the levelled construction for a gadget on level $L-1$.
    The ground edges and pseudo edges do not depend on the vertex labels.
    Level-$L$ pseudo edges are a subset of the ground edges,
    and level-$\ell$ pseudo edges are a subset of the Level-$(\ell+1)$ pseudo edges.
    The level-$\ell$ Laballed-ground edges are determined based on the vertex labels and the level-$\ell$ gadgets between them, here the $X$-$Y$ gadget on level $L-1$.
    The level-$\ell$ real edges are the intersection of level-$\ell$ pseudo edges and level-$\ell$ labelled-ground edges.
    }
    \label{fig:enter-label}
\end{figure}
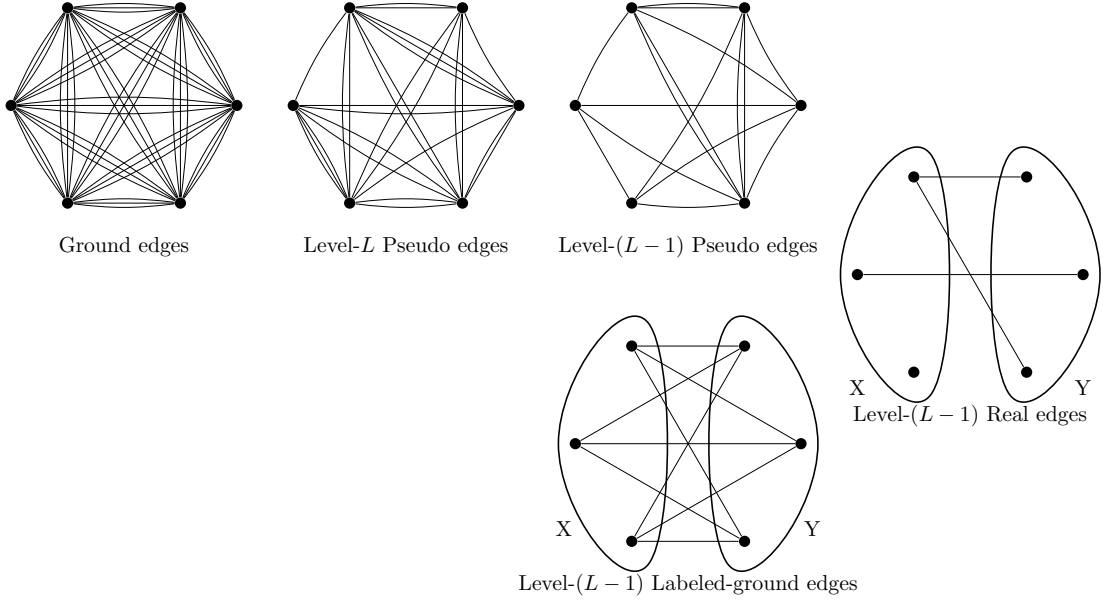

Finally, any ground edge is marked as a pseudo edge of level $L$ with probability $1/n$. 
If it is marked as a pseudo edge of level $L$, 
then it is marked as a pseudo edge of level $L-1$ with probability $\rho_{L-1}/\rho_L$, and so on. 
More specifically, if the edge is marked as a pseudo edge of level $\ell$, it will be marked as a pseudo edge of level $\ell-1$ with probability $\rho_{\ell-1}/\rho_\ell$. 
Observe that at this point the edge can become a real edge if it is also a level-$\ell$ labelled-ground edge.

We let $\deg^n_G(v)$, $\deg^p_G(v)$, and $\deg^r_G(v)$ be the number of non-pseudo edges, pseudo edge, real edges of vertex $v$ in graph $G$, respectively.

\subsection{Properties of the Constructed Multigraph}\label{sec:properties-multi}

\begin{claim} \label{clm:pseudo-prob}
    It holds
    \[
    \textnormal{Level-$\ell$ pseudo edges} \supset \textnormal{Level-$(\ell-1)$ pseudo edges},
    \]
    and for any ground edge, the probability of becoming a level-$\ell$ pseudo edge is $\rho_\ell/(n\rho)$.
\end{claim}
\begin{proof}
    The first part of the claim follows from the definition.
    The second part can be proven by induction on $\ell$.
    It holds for $\ell = L$ since $\rho_L = \rho$ and the probability of being a level-$L$ pseudo edge is $\frac{1}{n} = \frac{\rho_L}{\rho n}$.
    For $\ell < L$, by the induction hypothesis, the edge becomes a level-$(\ell + 1)$ pseudo edge with probability $\frac{\rho_{\ell + 1}}{\rho n}$.
    Therefore, it becomes a level-$\ell$ pseudo edge with probability
    $$
    \frac{\rho_{\ell + 1}}{\rho n} \cdot \frac{\rho_\ell}{\rho_{\ell + 1}} = \frac{\rho_{\ell}}{\rho n},
    $$
    which concludes the proof.
\end{proof}

\begin{claim} \label{clm:single-edge-dist}
    For any two vertices $u$ and $v$, the distribution of the number of edges between them is $\Binom\left( p^{(u, v)} n \cdot (\rho/\rho_\ell) , \rho_\ell/n\rho \right)$, where $\ell$ is the level of the gadget between $u$ and $v$. If there is no such gadget the number of edges is simply zero.
\end{claim}
\begin{proof}
    Recall, we added $p^{(u, v)}\cdot(\rho/\rho_\ell)$ level-$\ell$ labelled-ground edges between $u$ and $v$.
    By \cref{clm:pseudo-prob}, each of them has a probability $\rho_\ell/(n\rho)$ of becoming a level-$\ell$ pseudo edge.
    Those edge then constitute the real edges between $u$ and $v$.
    This concludes the proof.
\end{proof}

\begin{corollary} \label{cor:single-edge-expectation}
    For two vertices $u$ and $v$, the expected number of edges between them is $p^{(u, v)}$.
\end{corollary}

\begin{corollary} \label{cor:single-edge-prob}
    For two vertices $u$ and $v$, the probability that there exists an edge between them is at least $ p^{(u, v)}/2$.
\end{corollary}
\begin{proof}
    The probability of there being no edges is:
    $$
    (1 - \rho_\ell/n\rho)^{ p^{(u, v)}n\cdot(\rho/\rho_\ell)} \leq
    e^{-p^{(u, v)}}.
    $$
    Therefore, since $ p^{(u, v)} \leq 1$, the probability that there is an edge is at least
    \begin{equation*}
    1 - e^{-p^{(u, v)}} \geq  p^{(u, v)}\left(1 - \frac{1}{e}\right) \geq  p^{(u, v)} / 2. \qedhere
    \end{equation*}
\end{proof}

\begin{observation} \label{obs:non-pseudo-same-prob}
    Let $e$ and $e'$ be two different edges in the ground multigraph (possibly with the same endpoints). Then, the probability that $e$ is not in the pseudo graph is equal to the probability that $e'$ is not in the pseudo graph. 
\end{observation}
\begin{proof}
    The proof follows from the fact that each edge is a pseudo edge independently with probability $1/n$.
\end{proof}

\begin{lemma} \label{lem:deg-dist-general}
    Take a vertex $u$ in the multigraph.
    Let $\Phi$ be a subset of level-$\ell$ gadgets,
    and let $H$ be the subgraph including the ground edges associated with $\Phi$.
    Then, $\deg_H^r(u) = \Binom(d_\Phi(u) \cdot n\rho/\rho_\ell, \rho_\ell/n\rho)$.
    That is, the distribution of $\deg_H^r(u)$, the real degree of $u$ from edges of $H$,
    is solely determined by $d_\Phi(u)$.
\end{lemma}
\begin{proof}
    We utilize \cref{clm:single-edge-dist}.
    Let $\Phi' \subseteq \Phi$ be gadgets that have $u$ on one side.
    That is $\Phi' = \{\phi_1, \phi_2, \ldots, \phi_k\}$, where gadget $\phi_i$ with density parameter $P_i$ is between vertex sets $X_i$ and $Y_i$, and $u \in X_i$.
    Then, we have (below we use $\deg^r_H(u, Y)$ for a vertex set $Y$, to denote the number of real edges of $H$ between $u$ and $Y$):
    \begin{align*}
        \deg_H^r(u)
        &= \sum_{i} \deg^r_H(u, Y_i) \notag \\
        &= \sum_{i} \sum_{v \in Y_i} \deg^r_H(u, v) \notag \\
        &= \sum_i \sum_{v \in Y_i} \Binom(p^{(u, v)} \cdot n\rho/\rho_\ell, \rho_\ell/n\rho) \tag{by \cref{clm:single-edge-dist}} \\
        &= \sum_i \Binom(\card{Y_i} P_i \cdot n\rho/\rho_\ell, \rho_\ell/n\rho) \\
        &= \sum_i \Binom(d_{\phi_i}(u) \cdot n\rho/\rho_\ell, \rho_\ell/n\rho) 
        \tag{$d_{\phi_i}(u) = \card{Y_i}P_i$} \\
        &= \Binom(\sum_i d_{\phi_i}(u) \cdot n\rho/\rho_\ell, \rho_\ell/n\rho) \\
        &= \Binom(d_\Phi(u) \cdot n\rho/\rho_\ell, \rho_\ell/n\rho) \qedhere \\
    \end{align*}
\end{proof}

\begin{corollary}\label{cor:same-deg-dist-l}
    Take a vertex $u$ in the multigraph, 
    and let $H$ be the subgraph of ground edges associated with $\Phi^\ell$, the gadgets of level $\ell$.
    Then, the distribution of $\deg_H^r(u)$ is solely determined by $d_{\Phi^\ell}(u)$.
\end{corollary}
\begin{proof}
    Follows from a direct application of $\cref{lem:deg-dist-general}$ for $\Phi = \Phi^\ell$.
\end{proof}

\begin{corollary}\label{cor:same-deg-dist-pseudo}
    Take a vertex $u$ in the multigraph,
    and let $H$ be the subgraph of ground edges associated with $\Phi^\ell$, the gadgets of level $\ell$.
    Then, the distribution of the number of non-real pseudo edges of $H$ adjacent to vertex $u$, i.e.\ $\deg_H^p(u) - \deg_H^r(u)$, is solely determined by $d_{\Phi^\ell}$(u).
\end{corollary}
\begin{proof}
    The distribution of $\deg_H^r(u)$ is determined by $d_{\Phi^\ell}(u)$, and the distribution of $\deg_H^p(u)$ is determined by $\ell$ and is independent of the labels and gadgets.
\end{proof}

To wrap the section up,
we prove a lemma that states
the part of the graph queried by the algorithm is essentially a simple graph.

\begin{lemma} \label{lem:cant-tell-multigrpah}
    Let $\mc{A}$ be an algorithm that makes $O(n^{2-\delta})$ queries. Then, with high probability, it does not query any pair $(u, v)$ such that there is more than one pseudo edge between them in $G$.
\end{lemma}
\begin{proof}
    Take any query $(u, v)$ that the algorithm makes.
    Recall that the number of pseudo edges between them has distribution $\Binom(\rho n, \frac{1}{n})$, and is independent of everything else in the graph (including the parts the algorithm has queried).
    Therefore, the probability that $(u, v)$ has at most one pseudo edge is
    \begin{align*}
        \left(1 -\frac{1}{n}\right)^{\rho n}
        + \rho n \cdot \frac{1}{n}\left(1 -\frac{1}{n}\right)^{\rho n - 1} 
        &\geq \left(1 + \rho\right)\left(1 -\frac{1} {n}\right)^{\rho n} \\
        &\geq (1 + \rho)\left(1 - \rho + \frac{\rho^2}{2} - O(\rho/n) \right) \\
        &= 1 - O(\rho^2).
    \end{align*}
    It follows that the probability of having more than one pseudo edge is at most $O(\rho^2)$.
    By taking the union bound over all the $O(n^{2-\delta})$ queries,
    the probability that the algorithm queries a pair with more than one pseudo edge is at most
    \begin{equation*}
    O(\rho^2 \cdot n^{2-\delta}) = O(n^{\delta/5 - 2} \cdot n^{2-\delta})
    = O(n^{-4\delta/5}) = o(1). \qedhere
    \end{equation*}
\end{proof}

\subsection{Multigraph Query Model and Reduction to Simple Graphs} \label{sec:final-simple-graph}
\label{sec:not-observable-difference}

Recall that the final algorithm aims to approximate the maximum matching in the real multigraph.
A natural model for adjacency matrix queries for multigraphs
is querying the number of edges between two vertices $u$ and $v$.
We make the algorithm stronger as follows:
Upon querying a pair $(u, v)$, the algorithm finds out (1) how many pseudo edges there are between them, and (2) how many of them are real edges.
In later sections, we prove the lower bound for this model.
Note that solving the problem in this model is only easier than the natural multigraph model.

\begin{remark}
    In the proof of the lower bound, due to \cref{lem:cant-tell-multigrpah}, we can assume that the algorithm upon querying a pair $(u, v)$ either (1) sees no edges between them,
    (2) sees a pseudo edge that is not a real edge, or
    (3) sees a pseudo edge that is also a real edge.
\end{remark}

Furthermore, the problem is not easier for simple graphs.
Consider the following reduction.
Given an algorithm $\mc{A}$ for simple graphs, we obtain an algorithm for multigraphs:
Run $\mc{A}$ on the multigraph; if a query between two vertices $u$ and $v$ reports more than one edge, simply report one edge to $\mc{A}$.
This reduction essentially removes the multiplicity of the edges.
Therefore, it does not change the size of the maximum matching,
and the output of $\mc{A}$ is still an $\epsilon n$-additive approximation.

In the following sections, we prove that if the algorithm makes $\Omega(n^{2-\delta})$ queries, it cannot distinguish between input graphs drawn from \yesdist{} and \nodist{}. Our proof is based on a coupling between \yesdist{} and \nodist{}. The ground edges that the algorithm finds which are not pseudo or real edge do not affect the algorithm's ability to distinguish between the two cases \Cref{obs:non-pseudo-same-prob}. This is because the probability that a ground edge is a non-pseudo edge is independent of the labels of the vertices by \Cref{obs:non-pseudo-same-prob}. Consequently, such edges provide no useful information for distinguishing between the two distributions. As a result, we focus on coupling the pseudo and real edges between \yesdist{} and \nodist{}. So for the rest of the proof, we assume that ground edges that are not marked as pseudo do not exist in the graph and when we refer to edges, we mean pseudo and real edges. Whenever we want to mention a ground edge, we explicitly mention that. Moreover, for the rest of the proof when we refer to level $\ell$ edge, we mean pseudo or real edges that belong to level $\ell$ or lower. Also, when we use $A_i$ (resp. $B_i$, and $D_i$) without superscript, we mean $A_i^1$ and $A_i^2$ (the same for both $B_i$ and $D_i$). Further, when in the proof we are considering level $\ell$, by label of the vertex we mean the label of the vertex in the same level.

\section{Losing Advantage in the Highest Level}

The key distinction between graphs drawn from $\yesdist^L$ and $\nodist^L$ lies in the subgraph between $A_r^1$ and $A_r^2$ at the highest level. Specifically, in $\yesdist^L$ , this subgraph is sampled from $\yesdist^{L-1}$ , while in $\nodist{}^L$, it is sampled from $\nodist^{L-1}$. Therefore, any algorithm aiming to differentiate between $\yesdist^L$ and $\nodist^L$ must detect differences within this subgraph.

In this subsection, we derive an upper bound on the number of level $L-1$ edges the algorithm can recognize as belonging to this critical subgraph. To formalize this, we introduce the concept of the distinguishability of an edge within the subgraph between $A_r^1$ and $A_r^2$.

When the algorithm queries a typical edge, due to our choices of $d_L$ and $d_{L-1}$, parameters that control the degrees of vertices at levels $L$ and $L-1$, the probability that this edge belongs to the subgraph between $A_r^1$ and $A_r^2$ is approximately $d_{L-1}/{d_L} = n^{\sigma_{L-1} - \sigma_L}$ if the edge is a real edge. Also, for a pseudo edge, because of our choices of $\rho_L$ and $\rho_{L-1}$, the probability that this pseudo edge belongs to the subgraph at the lower level is approximately $\rho_{L-1}/\rho_L = n^{\sigma_{L-1} - \sigma_L}$.

We say an edge is distinguishable when the algorithm can infer a bias in this probability, conditioned on the queried subgraph.

\begin{definition}[$p_e^{inner}$ and Distinguishability of an Edge]\label{def:distinguishibility-top-level}
Let $e$ be a real edge or pseudo edge queried by the algorithm, and let $p_e^{inner}$ denote the probability that if $e$ is a level $L-1$ edge that belongs to subgraph between $A_r^1$ and $A_r^2$, conditioned on all queries made by the algorithm so far and assuming either input distribution. We say edge $e$ is  \textit{distinguishable} if $p_e^{\text{inner}} > 10 n^{\sigma_{L - 1} - \sigma_L}$.
\end{definition}

\begin{lemma}\label{lem:total-edges-discovered}
    With high probability, the total number of edges that the algorithm discovers is at most $O(n^{1-\delta+\sigma_L})$.
\end{lemma}
\begin{proof}
Since there are $\rho n$ ground edges between every two vertices and the algorithm makes $O(n^{2-\delta})$ pair queries, the total number of ground edges the algorithm queries is $O(\rho n^{3-\delta})$. Moreover, each of these ground edges is classified as either a pseudo or real edge with an independent probability of at most $1/n$, as they must be marked as pseudo edges of level $L$, which occurs with probability $1/n$. Therefore, in expectation, the algorithm finds $O(\rho n^{2-\delta})$ such edges.  

Let $X_i$ be the indicator random variable for the event that the $i$-th queried ground edge is a pseudo or real edge. Let $X = \sum X_i$ denote the total number of such edges found by the algorithm. We have $E[X_i] \leq 1/n$, and $E[X] \leq O(\rho n^{2 - \delta})$. Furthermore, the random variables $X_i$ are independent. Therefore, applying the Chernoff bound, we have:
\begin{align*}    
\Pr\left[ |X - E[X]| \geq 2 \sqrt{E[X] \log n} \right] \leq 2 \exp \left( -\frac{(2 \sqrt{E[X] \log n})^2}{3 E[X]} \right) < \frac{1}{n}.
\end{align*}
This implies that with probability at least $1 - 1/n$, the total number of pseudo or real edges discovered by the algorithm is $O(\rho n^{2 - \delta} )$. Plugging $\rho = n^{\sigma_L - 1}$ completes the proof.
\end{proof}

\begin{definition}[Direction of an Edge]\label{def:directing}
Let $(u, v)$ be an edge queried by the algorithm. Suppose the algorithm has already discovered some edges of $u$, but has not discovered any edges of $v$. When we refer to the \textit{direction} of the edge $(u, v)$, we mean that the edge is considered to go from $u$ to $v$. If the algorithm has already discovered edges for both $u$ and $v$, or has not discovered any edges for either, we assign the direction of the edge randomly.
\end{definition}

\begin{claim}\label{clm:max-in-deg-top}
    Each vertex in the queried subgraph has at most $3\sqrt{\log n}$ indegree with high probability.
\end{claim}
\begin{proof}
    Let $v$ be an arbitrary vertex. Suppose that $v$ has at least one incoming edge. Also, let $\hat{V}$ be the set of vertices in the graph for which the algorithm finds at least one edge. According to the way we direct the edges in \Cref{def:directing}, the rest of the incoming edges of $v$ are queried at the time when the other endpoint already has an existing edge. Therefore, there are $|\hat{V}|$ possible pairs between $\hat{V}$ and $v$, each containing $\rho n$ ground edges, and each being a pseudo or real edge with probability at most $1/n$. Thus, the expected number of edges between $\hat{V}$ and $v$ is:  
\begin{align*}
E[X] = |\hat{V}|\cdot \rho = \wt{O}(n^{2\sigma_L - \delta}),
\end{align*}
since $|\hat{V}| = O(n^{1-\delta+\sigma_L})$ by \Cref{lem:total-edges-discovered}.  Let $X_i$ be the indicator random variable for the event that the $i$-th ground edge between $\hat{V}$ and $v$ is a pseudo or real edge, and let $X = \sum X_i$ denote the total number of such edges. We have $E[X] = O(n^{2\sigma_L - \delta}) < 1$ for large enough $n$. Furthermore, the random variables $X_i$ are independent. Applying the Chernoff bound, we obtain:
\begin{align*}
\Pr\left[ |X - E[X]| \geq 3 \sqrt{E[X] \log n} \right] \leq 2 \exp \left( -\frac{(3 \sqrt{E[X] \log n})^2}{3 E[X]} \right) < \frac{1}{n^2}.
\end{align*}
Thus, with probability at least $1 - 1/n^2$, the total number of incoming pseudo or real edges to $v$ is $3\sqrt{\log n}$. Applying the union bound over all vertices concludes the proof.
\end{proof}

\begin{definition}[Shallow Subgraph]
For a vertex $v$, we define $v$'s shallow subgraph as the set of vertices that are reachable from $v$ via directed paths of length at most $10 \log n$ in the queried subgraph. We denote $v$'s shallow subgraph by $T(v)$.
\end{definition}

\begin{lemma}\label{lem:belong-shallow-count-vertex}
Each vertex belongs to at most $\widetilde{O}(1)$ shallow subgraphs with high probability.
\end{lemma}
\begin{proof}
Let $v$ be an arbitrary vertex in the graph. Let $V_i$ be the set of vertices at a distance $i$ from $v$ for $i \in [10 \log n]$ using edges in the reverse direction. We will show that, with high probability, $|V_i| \leq 3i\sqrt{\log n}$ using induction.

For the base case of the induction, $i = 1$, the claim holds by \Cref{clm:max-in-deg-top}. Suppose the claim holds for all $i' < i$. Suppose that each vertex in $V_{i-1}$ has at least one incoming edge. Note that this only increases the size of $V_i$. Let $u \in V_{i-1}$. Also, let $\hat{V}$ be the set of vertices that the algorithm has found at least one edge in the queried subgraph. If $u$ has more than one incoming edge, it should be between a vertex that already has an edge because of the way we defined the direction of edges in \Cref{def:directing}. Therefore, there are $|\hat{V}|$ possible pairs between $\hat{V}$ and $u$, each containing $\rho n$ ground edges, and each being a pseudo or real edge with probability at most $1/n$. Thus, the expected number of edges between $\hat{V}$ and $u$ is $|\hat{V}|\cdot \rho = O(n^{2\sigma_L - \delta})$,
since $|\hat{V}| = O(n^{1-\delta+\sigma_L})$ by \Cref{lem:total-edges-discovered}.  Additionally, $|V_{i-1}| \leq 3(i - 1)\sqrt{\log n}$ by induction hypothesis. Hence, the expected number of edges between $\hat{V}$ and $V_{i-1}$ is $\wt{O}(n^{2\sigma_L - \delta})$.

Let $X_i$ be the indicator random variable for the event that the $i$-th ground edge between $\hat{V}$ and $V_{i-1}$ is a pseudo or real edge, and let $X = \sum X_i$ denote the total number of such edges. We have $E[X] = \wt{O}(n^{2\sigma_L - \delta}) < 1$ for large enough $n$. Furthermore, the random variables $X_i$ are independent. Applying the Chernoff bound, we obtain:
\begin{align*}
\Pr\left[ |X - E[X]| \geq 3 \sqrt{E[X] \log n} \right] \leq 2 \exp \left( -\frac{(3 \sqrt{E[X] \log n})^2}{3 E[X]} \right) < \frac{1}{n^2}.
\end{align*}
Thus, with probability at least $1 - 1/n^2$, the total number of incoming pseudo or real edges to $v$ is $3\sqrt{\log n}$. Moreover, we assume that each vertex in $V_{i-1}$ contains at least one incoming edge. Hence, we have $|V_i| \leq |V_{i-1}| + 3\sqrt{\log n} = \leq 3i\sqrt{\log n}$ which completes the induction step. Therefore, we have $|V_i| \leq 3i\sqrt{\log n}$ for all $i \leq 10 \log n$. As a result, the total number of shallow subgraphs that contain $v$ is bounded by
\begin{align*}
    1 + \sum_{i=1}^{10\log n} |V_i| \leq 1 + \sum_{i=1}^{10\log n}3i\sqrt{\log n} \leq (10 \log n)\cdot (30 \log n \sqrt{\log n}),
\end{align*}
which completes the proof.
\end{proof}

\begin{corollary}\label{cor:belong-shallow-count-edge}
    Each edge belongs to at most $\widetilde{O}(1)$ shallow subgraphs with high probability.
\end{corollary}
\begin{proof}
    For directed edge $(u,v)$, since $u$ is in at most $\wt{O}(1)$ shallow subgraphs by \Cref{lem:belong-shallow-count-vertex}, then $(u,v)$ is in at most $\wt{O}(1)$ shallow subgraphs.
\end{proof}

\begin{definition}[Spoiler Vertex]\label{def:spoiler-vertex}
A vertex $u$ is called a spoiler vertex if an edge $(u, v)$ is discovered by the algorithm at a time when both $u$ and $v$ already have a non-zero degree.
\end{definition}

\begin{claim}\label{clm:spoiler-count}
    There are at most $O(n^{1-2\delta + 3\sigma_L})$ spoiler vertices with high probability.
\end{claim}
\begin{proof}
    Note that each query between two vertices that already have non-zero degree in the queried subgraph, resulting in a pseudo or real edge, creates two spoiler vertices. By \Cref{lem:total-edges-discovered}, there are at most $O(n^{1 - \delta + \sigma_L})$ vertices with non-zero degree. Therefore, the total number of such vertex pairs is at most $O(n^{2 - 2\delta + 2\sigma_L})$. If the algorithm queries all these pairs, the total number of ground edges queried is at most $O(\rho n^{3 - 2\delta + 2\sigma_L})$, where each forms a pseudo or real edge with probability at most $1/n$. Consequently, the expected number of edges discovered by the algorithm between pairs of vertices with a non-zero degree is at most $O(\rho n^{2 - 2\delta + 2\sigma_L}) = O(n^{1 - 2\delta + 3\sigma_L})$. Therefore, using Chernoff bound, we can show that with high probability, there are at most $O(n^{1 - 2\delta + 3\sigma_L})$ spoiler vertices.
\end{proof}

\begin{definition}[Spoiled Vertex]\label{def:spoiled-vertex}
A vertex $v$ is called a spoiled vertex if its shallow subgraph contains any of the following:
\begin{itemize}
    \item[(i)] a spoiler vertex; or
    \item[(ii)] at least $n^{\delta - 2\sigma_L}$ vertices.
\end{itemize}
\end{definition}

\begin{observation}\label{obs:tree-structure-of-unspoiled}
    Let $v$ be a vertex that is not spoiled. Then, the shallow subgraph of $v$ forms a rooted tree with at most $n^{\delta - 2\sigma_L}$ vertices. Additionally, for every edge $(u,w)$ in the shallow subgraph of $v$, when the algorithm queries this edge, vertex $w$ is a singleton.
\end{observation}
\begin{proof}
    When the algorithm encounters an edge $(u,w)$ in the shallow subgraph of $v$, vertex $w$ must be a singleton, as per \Cref{def:spoiler-vertex} and \Cref{def:spoiled-vertex}. This ensures that the shallow subgraph of $v$ is indeed a rooted tree.
\end{proof}

\begin{lemma}\label{lem:spoiled-count}
    There are at most $O(n^{1-2\delta + 4\sigma_L})$ spoiled vertices with high probability.
\end{lemma}
\begin{proof}
    According to \Cref{def:spoiled-vertex}, at least one of the two conditions is required for a vertex to be spoiled. For condition (i), by \Cref{clm:spoiler-count}, there are $O(n^{1-2\delta + 3\sigma_L})$ spoiler vertices. Moreover, each vertex is in $\wt{O}(1)$ shallow subgraphs by \Cref{lem:belong-shallow-count-vertex}, which implies that there are at most $O(n^{1-2\delta + 4\sigma_L})$ vertices that have the condition (i) for large enough $n$.

    For condition (ii), by \Cref{cor:belong-shallow-count-edge}, each edge in the queried subgraph appears in $\wt{O}(1)$ shallow subgraphs. Thus, for non-isolated vertices in the queried subgraph, by \Cref{lem:total-edges-discovered}, we have $\sum_v |T(v)| \leq \wt{O}(n^{1 - \delta + \sigma_L})$. Consequently, the total number of vertices whose shallow subgraph contains more than $n^{\delta - 2\sigma_L}$ vertices is at most $O(n^{1 - 2\delta + 4\sigma_L})$.
\end{proof}

\begin{definition}[Spoiled Edge]\label{def:spoiled-edge}
    Let $(u,v)$ be a directed edge. We say $(u,v)$ is a spoiled edge if $u \in A_r$ and at least one of the following condition holds:
    \begin{itemize}
        \item[(i)] $v$ is a spoiled vertex; or
        \item[(ii)] $u$ has at least $n^{\sigma_L}/3$ spoiled neighbors in the queried subgraph.
    \end{itemize}
\end{definition}

\begin{lemma}\label{lem:spoiled-edge-count}
    There are at most $O(n^{1-2\delta + 5\sigma_L})$ spoiled edges with high probability.
\end{lemma}
\begin{proof}
    Each vertex in the queried subgraph has an indegree of $\wt{O}(1)$ with high probability by \Cref{clm:max-in-deg-top}. Also, by \Cref{lem:spoiled-count}, we have at most $O(n^{1-2\delta + 4\sigma_L})$ spoiled vertices. Therefore, there are at most $O(n^{1-2\delta + 5\sigma_L})$ edges that satisfy condition (i) of \Cref{def:spoiled-edge}.

    Further, if a vertex $u$ satisfies condition (ii), it must have at least $n^{\sigma L}/4$ outgoing edges $(u, w)$ where $w$ is spoiled.
    However, from \Cref{clm:max-in-deg-top}, each spoiled vertex has indegree of $\wt{O}(1)$. Also, the total number of spoiled vertices is $O(n^{1 - 2\delta + 4\sigma L})$ which implies that the number of edges that satisfy condition (ii) of \Cref{def:spoiled-edge} is at most $O(n^{1 - 2\delta + 5\sigma L})$.
\end{proof}

We defer the proof of the following lemma to a later section, as it is involved, lengthy, and mostly independent of the flow of this section.

\begin{lemma}\label{lem:coupling-high-level}
Let $v$ be a vertex that is not spoiled and belongs to $\{A_r, B_r, D_r\}$. Let $\mc{L}(v)$ and $\mc{L}'(v)$ represent an arbitrary label for $v$ from $\{A_r, B_r, D_r\}$ and the entire queried subgraph all available labels at level $L$, excluding the shallow subgraph of $v$ and isolated vertices. Then, we have:  
\begin{align*}
\Pr[T(v) \mid \mc{L}(v)] \leq \left(1 + O(n^{\sigma_L - \delta})\right)^{|T(v)|} \cdot \Pr[T(v) \mid \mc{L}'(v)].
\end{align*}
\end{lemma}

\begin{lemma}\label{lem:non-spoiled-edge-non-bias}
    Let $e$ be a directed edge that is not spoiled, then it holds that $p_e^{\text{inner}} \leq 10 n^{\sigma_{L - 1} - \sigma_L}$.
\end{lemma}
\begin{proof}
Let $e = (u,v)$ be the directed edge (directed from $u$ to $v$). First, if $u \notin A_r$, we have  $p_e^{\text{inner}} = 0$. Also, if $e$ is a pseudo edge, by the construction, the probability that we mark $e$ as level $L-1$ is $\rho_{L-1}/\rho_L = n^{\sigma_{L-1} - \sigma_L}$ independent at random from other edges.

Thus, suppose that $u \in A_r$ and $e$ is a real edge for the rest of the proof. According to the construction, $u$ has at least $6n^{\sigma_L}/7$ real edges such that the other endpoint has label $A_r \cup B_r$. Also, $u$ has at most $\wt{O}(1)$ incoming edge by \Cref{clm:max-in-deg-top}. Since $e$ is not a spoiled edge, $u$ has at most $n^{\sigma_L}/3$ neighbors in the queried subgraph that are spoiled. Let $V_u$ be the set of neighbors of $u$ using real edges in the original graph such that their label is $A_r \cup B_r$ and either they are singleton or direct children of $u$ in the queried subgraph. By the above argument, we have $|V_u| \geq n^{\sigma_L}/2$. Also, note that $v \in V_u$. Now we provide an upper bound on the probability that $e$ belongs to level $L-1$, or in other words, $v \in A_r$.  We prove this upper bound using \Cref{lem:coupling-high-level}.

Consider a labeling profile $\mathcal{P}$ of all vertices $V_u$ such that $\mathcal{P}(v) = A_r$. By the construction of our input distribution, since $u \in A_r$, at most $O(d_{L-1}) = O(n^{\sigma_{L-1}})$ vertices in $V_u$ belong to $A_r$. We generate $\Omega(n^{\sigma_L})$ new profiles $\mathcal{P}'$ where $\mathcal{P}'(v) \neq A_r$. For each vertex $w \in V_u$ with $\mathcal{P}(w) = B_r$, we create a new profile $\mathcal{P}'$ by setting $\mathcal{P}'(z) = \mathcal{P}(z)$ for $z \notin \{v, w\}$, $\mathcal{P}'(w) = A_r$, and $\mathcal{P}'(v) = B_r$.  

By \Cref{lem:coupling-high-level}, the probability of querying the same shallow subgraphs $T(v)$ and $T(w)$ in the new labeling profile changes only by factors of  
\begin{align*}
\left(1 + O(n^{\sigma_L - \delta}) \right)^{|T(v)|} \quad \text{and} \quad \left(1 + O(n^{\sigma_L - \delta}) \right)^{|T(w)|},
\end{align*}  
respectively. Since $v$ and $w$ are not spoiled vertices, \Cref{def:spoiled-vertex} gives $|T(v)|, |T(w)| \leq n^{\delta - 2\sigma_L}$. Therefore, the probability of generating profiles $\mathcal{P}$ and $\mathcal{P}'$ differs by at most  
\begin{align*}
\left(1 + O(n^{\sigma_L - \delta}) \right)^{|T(v)|} \cdot \left(1 + O(n^{\sigma_L - \delta}) \right)^{|T(w)|} 
&\leq \left(1 + O(n^{\sigma_L - \delta}) \right)^{2n^{\delta - 2\sigma_L}} \\
&\leq 1 + o(1).
\end{align*}  
Now, construct a bipartite graph $H = (P_1, P_2, E_P)$ of labeling profiles, where $P_1$ consists of all profiles $\mathcal{P}$ with $\mathcal{P}(v) = A_r$, and $P_2$ contains all profiles $\mathcal{P}'$ with $\mathcal{P}'(v) = B_r$. Add an edge between $\mathcal{P} \in P_1$ and $\mathcal{P}' \in P_2$ if $\mathcal{P}$ can be transformed into $\mathcal{P}'$ through the process described earlier.  

For any profile $\mathcal{P} \in P_1$, we have $\deg_H(\mathcal{P}) \geq |V_u|/2 \geq n^{\sigma_L}/4$ since at least $|V_u|/2$ vertices in $V_u$ belong to $B_r$. In contrast, for any profile $\mathcal{P}' \in P_2$, $\deg_H(\mathcal{P}') \leq 2n^{\sigma_{L-1}}$ because, by the input distribution, at most $2d_{L-1} = 2n^{\sigma_{L-1}}$ vertices $w$ in $V_u$ satisfy $\mathcal{P}'(w) = A_r$. Therefore,  
\begin{align*}
p_{e}^{\text{inner}} 
&\leq (1 + o(1)) \cdot \frac{|P_1|}{|P_2|} \\
&\leq (1 + o(1)) \cdot \frac{2n^{\sigma_{L-1}}}{n^{\sigma_L}/4} \\
&\leq (1 + o(1)) \cdot 8n^{\sigma_{L-1} - \sigma_L} \\
&\leq 10n^{\sigma_{L-1} - \sigma_L},
\end{align*}  
which completes the proof.  
\end{proof}

\begin{corollary}\label{cor:biased-edge-count}
    There are at most $O(n^{1 - 2\delta + 5\sigma_L})$ edges $e$ such that $p_e^{\text{inner}} > 10 n^{\sigma_{L - 1} - \sigma_L}$.
\end{corollary}
\begin{proof}
    The proof follows by combining \Cref{lem:spoiled-edge-count} and \Cref{lem:non-spoiled-edge-non-bias}.
\end{proof}

\section{Unbiased Edges Results in Small Connected Components}

In this section, we demonstrate that as the recursive construction progresses to deeper levels, the algorithm finds it increasingly difficult to form large connected components using inner-level edges. By \Cref{cor:biased-edge-count}, at the highest level of the construction, the algorithm can identify at most $O(n^{1 - 2\delta + 5\sigma_L})$ edges in the queried subgraph as belonging to the inner level with a probability greater than $10n^{\sigma_{L-1} - \sigma_L}$. For the sake of analysis, we assume the algorithm can perfectly distinguish these edges.  

However, for all remaining queried edges, the probability of belonging to the inner level is significantly lower due to the degree choices formalized in \Cref{lem:non-spoiled-edge-non-bias}. Specifically, each additional queried edge has at most $O(n^{\sigma_{L-1} - \sigma_L})$ probability of being an inner-level edge.  

Our objective is to formalize a similar result to \Cref{cor:biased-edge-count} for each level in the hierarchy in the next two sections. Intuitively, \Cref{lem:advantage-from-top} shows that as the construction descends through the levels, the number of edges the algorithm can confidently identify as inner-level edges diminishes. To proceed, we first extend \Cref{def:distinguishibility-top-level} to all levels of the hierarchy.

\begin{definition}[$p_e^{\ell-inner}$ and Distinguishability of an Edge]\label{def:distinguishibility}
Let $e$ be a real edge or pseudo edge queried by the algorithm, and let $p_e^{\ell-inner}$ denote the probability that if $e$ is a level $\ell-1$ edge that belongs to subgraph between $A_r^1$ and $A_r^2$, conditioned on all queries made by the algorithm so far and assuming either input distribution. We say edge $e$ is \textit{distinguishable} if $p_e^{\ell-inner} > 10 n^{\sigma_{\ell - 1} - \sigma_\ell}$. We use $E^{\text{inner}}_\ell$ to denote the set of edges for which $p_e^{\ell\text{-inner}} > 10n^{\sigma_{\ell-1} - \sigma_\ell}$.
\end{definition}

We define function $g(\ell)$ for $0\leq \ell \leq L$ as follows:
\begin{align*}
    g(\ell) =(L-\ell+2)\cdot \delta - 5\left(\sum_{i=\ell}^{L} \sigma_{i}/\sigma_{i+1}\right) - 5\left(\sum_{i=\ell}^{L-1}\sigma_i\right),
\end{align*}
where $\sigma_0 = 0$ and $\sigma_{L+1} = 1$ for the purpose of defining this function.

\begin{restatable}{observation}{gFunctionProperties}\label{obs:g-function-properties}
    The following statements are true regarding function $g$:
    \begin{itemize}
        \item[(i)] $g(\ell - 1) = g(\ell) + \delta - 5\sigma_{\ell - 1}/\sigma_\ell - 5\sigma_{\ell-1}$ for $\ell \in (1, L]$,
        \item[(ii)] $1-g(\ell - 1) - 3\sigma_{\ell-1} = 1 - g(\ell) - \delta + 5\sigma_{\ell-1}/\sigma_\ell + 2\sigma_{\ell-1}$ for $\ell \in (1, L]$,
        \item[(iii)] $g(1) > 2$,
        \item[(iv)] $1 - g(\ell) \neq 0$ for all $\ell \in [L]$.
    \end{itemize}
\end{restatable}

We defer the proof of \Cref{obs:g-function-properties} to the appendix as it follows from straightforward algebra. When we consider the level $\ell$ graph and we refer to the inner level, we specifically mean the union of 
subgraph of real edges between $A^1_r$ and $A^2_r$ at level $\ell - 1$, and pseudo edges of level $\ell - 1$ between $A^1_r$ and $A^2_r$. In this section, we denote the edges of the inner level as {\em inner edges}, while all other edges are referred to as {\em outer edges}. We prove that the algorithm cannot grow a large component of inner edges. We prove the following two lemmas using induction on $\ell$. For the base case of $\ell = L$ in \Cref{lem:advantage-from-top}, we already proved the claim in the previous section (\Cref{cor:biased-edge-count}).  To prove \Cref{lem:small-square-component} for a fixed $\ell$, we first apply the bound from \Cref{lem:advantage-from-top} at level $\ell$. We then use this result to establish \Cref{lem:advantage-from-top} for $\ell - 1$. In this section, our primary focus is on proving \Cref{lem:small-square-component} using \Cref{lem:advantage-from-top}. The remainder of this subsection is dedicated to detailing the steps involved in proving \Cref{lem:small-square-component}.

\begin{lemma}\label{lem:advantage-from-top} 
The following statements hold with high probability:
\begin{itemize} 
\item[(i)] If $1 - g(\ell) > 0$, then the number of edges $e$ such that $p_e^{\ell-inner} > 10n^{\sigma_{\ell-1} - \sigma_\ell}$ is at most $O(n^{1 - g(\ell)})$ with high probability. 
\item[(ii)] If $1 - g(\ell) < 0$, then there are no edges $e$ such that $p_e^{\ell-inner} > 10n^{\sigma_{\ell-1} - \sigma_\ell}$ with probability $1 - O(n^{1 - g(\ell)})$.
\item[(iii)] If $1 - g(\ell) < 0$, the number of edges $e$ satisfying $p_e^{\ell-inner} > 10n^{\sigma_{\ell-1} - \sigma_\ell}$ is at most $\widetilde{O}(1)$ with high probability.
\end{itemize}
\end{lemma}

\begin{lemma}\label{lem:small-square-component} 
Let $C_1, C_2, \ldots, C_c$ denote the underlying undirected connected components of inner edges, where each component contains at least one edge from $E^{inner}_\ell$. Then, the following statements hold: \begin{itemize} 
\item[(i)] If $1 - g(\ell) > 0$, then with high probability,
\begin{align*}
    \sum_{i=1}^c |C_i| \leq O(n^{1 - g(\ell) + 5\sigma_{\ell-1}/\sigma_\ell}).
\end{align*}
\item[(ii)] If $1 - g(\ell) < 0$, then we have $c = 0$ with probability $1 - O(n^{1 - g(\ell)})$.
\item[(iii)] If $1 - g(\ell) < 0$, then with high probability,
\begin{align*}
\sum_{i=1}^c |C_i| \leq \widetilde{O}(n^{5\sigma_{\ell-1}/\sigma_{\ell}})
\end{align*}
\end{itemize} 
\end{lemma}

\begin{restatable}{claim}{maxInnerEdgesDiscovered}\label{clm:max-inner-edges}
    There are at most $O(n^{1-\delta+\sigma_{\ell-1}})$ inner edges in the queried subgraph with high probability.
\end{restatable}

We defer the proof of \Cref{clm:max-inner-edges} to the appendix, as it follows a very similar approach to that of \Cref{lem:total-edges-discovered}.

\begin{claim}\label{clm:longest-black-directed-path}
    Consider all queried inner edges, excluding those in $E^{\text{inner}}_\ell$. The length of the longest directed path consisting of inner edges is less than $5/\sigma_\ell$.
\end{claim}
\begin{proof}
    Let $u$ be an arbitrary vertex. We first claim that the probability of a directed path of length $i$ starting from $u$ and ending at a vertex $v$ is at most $n^{i(\sigma_{\ell-1} - \sigma_\ell)/2}$. We prove this claim by induction on $i$.

For the base case $i = 1$, if no edge exists between $u$ and $v$, the probability is trivially 0. If there is an edge, \Cref{lem:advantage-from-top} implies that this edge is an inner edge with probability at most $10n^{\sigma_{\ell-1} - \sigma_\ell} < n^{(\sigma_{\ell-1} - \sigma_\ell)/2}$.

Assume the claim holds for all lengths less than $i$. By \Cref{clm:max-in-deg-top}, the in-degree of $v$ in the entire queried subgraph (including all edges) is at most $3 \sqrt{\log n}$. Let $\{v_1, v_2, \ldots, v_k\}$ denote the vertices with directed edges to $v$. If there is a directed inner path of length $i$ to $v$, then there must exist a path of length $i - 1$ to some $v_j$ and an inner edge from $v_j$ to $v$. Let $I^i_w$ represent the event that there exists a directed inner path of length $i$ to vertex $w$. Using a union bound:

\begin{align*}
    \Pr[I^i_v] &\leq \sum_{j=1}^k \Pr[I^{i-1}_{v_j}] \cdot \Pr[(v_j, v) \text{ is inner}] \\
    &\leq \sum_{j=1}^k n^{(i-1)(\sigma_{\ell-1} - \sigma_\ell)/2} \cdot 10n^{\sigma_{\ell-1} - \sigma_\ell} & (\text{By induction hypothesis and \Cref{lem:advantage-from-top}}) \\
    &\leq 3\sqrt{\log n} \cdot n^{(i-1)(\sigma_{\ell-1} - \sigma_\ell)/2} \cdot 10n^{\sigma_{\ell-1} - \sigma_\ell} & (k \leq 3\sqrt{\log n} \text{ by \Cref{clm:max-in-deg-top}})\\
    &\leq 30 \sqrt{\log n} \cdot n^{(i+1)(\sigma_{\ell-1} - \sigma_\ell)/2} \\
    &\leq n^{i(\sigma_{\ell-1} - \sigma_\ell)/2} & (n \text{ sufficiently large enough}),
\end{align*}
which completes the induction.

Next, we show that with high probability, no directed inner path of length $5/\sigma_\ell$ exists in the graph. The probability of such a path between any two vertices $u$ and $v$ is bounded by $n^{5(\sigma_{\ell-1} - \sigma_\ell)/(2\sigma_\ell)}$. Applying a union bound over all vertex pairs:
\begin{align*}
    \Pr\left[\exists \text{ directed inner path of length $5/\sigma_\ell$}\right] 
    &\leq n^2 \cdot n^{5(\sigma_{\ell-1} - \sigma_\ell)/(2\sigma_\ell)} \\
    &\leq n^{-1/4}. & (\text{Since } \sigma_{\ell-1} < \frac{\sigma_\ell}{10}).\\
\end{align*}
Thus, with high probability, the longest directed inner path has a length of at most $5/\sigma_\ell - 1$. 
\end{proof}

\begin{claim}\label{lem:descendants-upper-bound}
Consider all queried inner edges excluding those in $E^{inner}_\ell$. Each vertex has at most $n^{5\sigma_{\ell-1} / \sigma_\ell}$ descendants reachable via directed inner edges with high probability. Furthermore, for each vertex, the total number of inner edges to its descendants is at most $n^{5\sigma_{\ell-1} / \sigma_\ell}$.
\end{claim}

\begin{proof}
Suppose that we condition on the absence of inner paths of length $5/\sigma_\ell$ by \Cref{clm:longest-black-directed-path}. Also, each vertex has at most $n^{\sigma_{\ell-1}}$ inner edges in total, even those not queried. Consequently, the number of vertices reachable from $u$ within a distance of $5/\sigma_\ell$ is bounded by $n^{5\sigma_{\ell-1} / \sigma_\ell}$.

\end{proof}

\begin{definition}[Strongly Connected Component] 
    Let \( G \) be a directed graph. A subset of vertices \( C \) is called a strongly connected component of \( G \) if it is a maximal set of vertices such that, for every pair of vertices \( u, v \in C \), there exists a directed path from \( u \) to \( v \) and a directed path from \( v \) to \( u \).  
\end{definition}

Consider the strongly connected component decomposition of the directed inner edges queried by the algorithm excluding edges in $E^{\text{inner}}_\ell$. From each component with zero in-degrees, we select a representative vertex (an arbitrary vertex in the component). Denote the set of these vertices by $R$.

\begin{lemma}\label{lem:intersect-desendants-prob}
    Consider all queried inner edges excluding those in $E^{\text{inner}}_\ell$. Let $v \in R$. The probability that there exists a vertex $u \in R \setminus \{v\}$ such that the descendants of $u$ intersect with those of $v$ is at most $O(n^{5\sigma_{\ell-1} / \sigma_\ell - \delta})$.
\end{lemma}
\begin{proof}
By \Cref{lem:descendants-upper-bound}, vertex $v$ has at most $n^{5\sigma_{\ell-1} / \sigma_\ell}$ inner descendants. This provides an upper bound on the number of descendants of $v$ that are reachable via directed inner edges. By \Cref{clm:max-inner-edges}, there are at most $O(n^{1-\delta + \sigma_{\ell-1}})$ vertices with inner edges. At any point during the execution of the algorithm, there are at most
\begin{align*}
    X &= O\left(n^{1-\delta + \sigma_{\ell-1} + 5\sigma_{\ell-1}/\sigma_\ell}\right)
\end{align*}
pairs consisting of a descendant of $v$ and a vertex that has at least one inner edge. Since each pair has $\rho n$ ground edges between them, the total number of ground edges between these pairs is at most $\rho nX$. 

For any ground edge between a pair to become a level $\ell - 1$ real edge, it must be marked as a pseudo edge of level $\ell - 1$. The probability that a given ground edge between a pair is marked as a pseudo edge of level $\ell - 1$ is
\begin{align*}
    \frac{1}{n} \cdot \left( \frac{\rho_{\ell-1}}{\rho} \right),
\end{align*}
independently at random for each of the ground edges. To compute the probability that there is an intersection between the descendants of $v$ and the descendants of any other vertex $u$, we multiply the number of ground edges by the probability that each of them is a pseudo edge of level $\ell - 1$. Thus, the total probability of having an intersection is bounded by
\begin{align*}
    (\rho n X) \cdot \frac{1}{n} \cdot \left( \frac{\rho_{\ell-1}}{\rho} \right) = \rho_{\ell - 1} X = O(n^{5\sigma_{\ell - 1} / \sigma_\ell - \delta})
\end{align*}
\end{proof}

\begin{corollary}
    Consider all queried inner edges excluding those in $E^{inner}_\ell$. Let $k < 50\sigma_\ell / \sigma_{\ell-1}$ and let $v_1, \ldots, v_k$ be $k$ arbitrary vertices in $R$. The probability that there exists a vertex $u \in R \setminus \{v_1, \ldots, v_k\}$ such that the descendants of $u$ intersect with the descendants of any vertex in $\{v_1, \ldots, v_k\}$ is at most $O(n^{5\sigma_{\ell-1} / \sigma_\ell - \delta})$.
\end{corollary}
\begin{proof}
    The proof follows the same reasoning as that of \Cref{lem:intersect-desendants-prob}, noting that $k$ is a constant.
\end{proof}

\begin{claim}\label{clm:small-component-in-intersect-tree}
    Consider all queried inner edges excluding those in $E^{inner}_\ell$. Let $C$ be an arbitrary connected component formed by these edges. It holds that $|C| \leq O(n^{5\sigma_{\ell-1} / \sigma_\ell})$ with high probability. Furthermore, each connected component contains at most $O(|C|)$ inner edges.
\end{claim}
\begin{proof}
    Define a graph $H_R$ with vertex set $R$, where an edge $(u, v)$ exists in $H_R$ if the inner descendants of $u$ and $v$ intersect. We will show that, with high probability, the largest connected component in $H_R$ has a size of at most $10/\delta$. By \Cref{lem:descendants-upper-bound}, each vertex has at most $n^{5\sigma_{\ell-1}/\sigma_\ell}$ inner descendants (and inner edges to its descendants), so proving this claim suffices to complete the proof.

    Consider the following exploration process starting from any vertex $v \in R$. Initialize a set $S = \{v\}$. At each step, examine the edges between vertices in $S$ and those in $R \setminus S$. If there exists an edge $(w, z)$ where $w \in S$ and $z \in R \setminus S$, add $z$ to $S$ and continue. The process halts when either no such edge exists or when $|S| > 10/\delta$. 

    Let $X_i$ denote the event that an edge between $S$ and $R \setminus S$ is revealed at step $i$. By \Cref{lem:intersect-desendants-prob}, the probability of this event is bounded as:
    \begin{align*}
        \Pr[X_i \mid X_1, \ldots, X_{i-1}] \leq O(n^{5\sigma_{\ell-1} / \sigma_\ell - \delta + \sigma_{\ell-1}}).
    \end{align*}
    Let $Y_v$ represent the event that the process reaches $|S| > 10/\delta$. The probability of this happening is:
    \begin{align*}
        \Pr[Y_v] &= \prod_{i=1}^{10/\delta} \Pr[X_i \mid X_1, X_2, \ldots, X_{i-1}] \\
        &\leq O\left( (n^{5\sigma_{\ell-1} / \sigma_\ell - \delta + \sigma_{\ell-1}})^{10/\delta} \right) \\
        &= O\left(\frac{1}{n^5}\right).
    \end{align*}
    Finally, applying a union bound over all vertices $v \in R$, the probability that any connected component in $H_R$ exceeds size $10/\delta$ is at most $O(n^{-4})$. Therefore, with high probability, the largest connected component in $H_R$ has size at most $10/\delta$.
\end{proof}

\begin{corollary}\label{cor:small-component-in-intersect-tree-number}
    Consider all queried inner edges excluding those in $E^{inner}_\ell$. Let $C$ be any connected component formed by the intersection of descendants of $k$ vertices in $R$. Then, with high probability, $k \leq 10/\delta$.
\end{corollary}

\begin{proof}
    The result follows directly from the proof of \Cref{clm:small-component-in-intersect-tree}.
\end{proof}

We are ready to prove \Cref{lem:small-square-component}.

\begin{proof}[Proof of \Cref{lem:small-square-component}]
Assume an adversary selects the edges in $E^{inner}_\ell$ to connect the components. Let $\widehat{\mathcal{C}} = \{\widehat{C_1}, \widehat{C_2}, \ldots, \widehat{C_{c'}}\}$ denote the connected components prior to adding the edges in $E^{inner}_\ell$. By \Cref{clm:small-component-in-intersect-tree}, with high probability, each component satisfies $|\widehat{C_i}| = O(n^{5\sigma_{\ell-1} / \sigma_\ell})$ for all $i \in [c']$.

Now, let $C_1, C_2, \ldots, C_c$ represent the connected components formed after adding the edges in $E^{inner}_\ell$ and discarding components that do not contain any of these edges. Each edge in $E^{inner}_\ell$ can merge at most two components from $\widehat{\mathcal{C}}$, so the number of components in $\widehat{\mathcal{C}}$ connected by at least one edge from $E^{inner}_\ell$ is bounded by $O(|E^{inner}_\ell|)$. Now we prove each statement of the lemma separately:
\begin{itemize}
    \item[(i):] If $1 - g(\ell) > 0$, statement (i) of \Cref{lem:advantage-from-top} implies that $|E^{inner}_\ell| \leq O(n^{1 - g(\ell)})$ with high probability. Since each initial component has size $|\widehat{C_i}| = O(n^{5\sigma_{\ell-1} / \sigma_\ell})$, we conclude that:
\begin{align*}
    \sum_{i=1}^c |C_i| &\leq O(n^{1 - g(\ell) + 5\sigma_{\ell-1}/\sigma_\ell}),
\end{align*}
which completes the proof of the first statement.
    \item[(ii):] If $1 - g(\ell) < 0$, then by statement (ii) of \Cref{lem:advantage-from-top}, with probability $1 - O(n^{1 - g(\ell)})$, the size of $E^{inner}_\ell$ is zero which completes the proof of the second statement.
    \item[(iii):] If $1 - g(\ell) < 0$, $|E^{inner}_\ell| = \widetilde{O}(1)$ with high probability according to the statement (iii) of \Cref{lem:advantage-from-top}. Combining this with the fact that $|\widehat{C_i}| = O(n^{5\sigma_{\ell-1} / \sigma_\ell})$, the last statement follows immediately.
\end{itemize}
\end{proof}

\begin{corollary}\label{cor:total-number-edges-inner-comp}
    Let $C_1, C_2, \ldots, C_k$ be the underlying undirected connected components of inner edges, where each component contains at least one edge from $E^{inner}_\ell$. Let $E(C_i)$ denote the edge set of component $C_i$. Then, the following statements hold:
    \begin{itemize}
        \item[(i)] If $1 - g(\ell) > 0$, then with high probability,  
        \begin{align*}
            \sum_{i=1}^k |E(C_i)| \leq \widetilde{O}\left(n^{1 - g(\ell) + 5\sigma_{\ell-1}/\sigma_\ell}\right).
        \end{align*}
        \item[(ii)] If $1 - g(\ell) < 0$, then $k = 0$ with probability $1 - O(n^{1 - g(\ell)})$. 
        \item[(iii)] If $1 - g(\ell) < 0$, then with high probability,  
        \begin{align*}
            \sum_{i=1}^k |E(C_i)| \leq \widetilde{O}\left(n^{5\sigma_{\ell-1}/\sigma_\ell}\right).
        \end{align*}
    \end{itemize}
\end{corollary}
\begin{proof}
    Note that $|E(C_i)| \leq O(\log n) \cdot |C_i|$. This is because each vertex has an incoming degree of at most $3 \sqrt{\log n}$ in the entire queried subgraph by \Cref{clm:max-in-deg-top}. By merging the results from \Cref{lem:small-square-component} with the fact that $|E(C_i)| \leq O(\log n) \cdot |C_i|$, we derive each statement accordingly.
\end{proof}

\section{Smaller Connected Components Lead to Fewer Unbiased Inner Edges}

In this section, we leverage \Cref{lem:small-square-component} to demonstrate that shrinking the size of connected components makes it increasingly difficult for the algorithm to detect inner edges. For clarity and consistency, we extend the notions of spoiler vertices, spoiled vertices, spoiled edges, and shallow subgraphs from the warm-up section, adapting the terminology accordingly for level $\ell$.

\begin{definition}[$\ell$-Shallow Subgraph]\label{def:shallow-subgraph-extend}
Let the outer and inner edges be defined with respect to levels $\ell$ and $\ell - 1$ of the construction hierarchy. For a vertex $v$, the $\ell$-shallow subgraph of $v$ is the set of vertices that are reachable from $v$ by directed paths of length at most $10 \log n$ using only inner edges in the queried subgraph. We denote this subgraph as $T^\ell(v)$.
\end{definition}

\begin{restatable}{lemma}{vertexInShallowCountExtended}\label{lem:vertex-in-shallow-count-extended}
    Each vertex belongs to at most $\widetilde{O}(1)$ $\ell$-shallow subgraphs with high probability.
\end{restatable}

\begin{restatable}{corollary}{edgeInShallowCountExtended}\label{cor:edge-in-shallow-count-extended}
    Each inner edge belongs to at most $\widetilde{O}(1)$ $\ell$-shallow subgraphs with high probability.
\end{restatable}

We postpone the proofs of the above lemma and corollary to the appendix, as they are similar to those of \Cref{lem:belong-shallow-count-vertex} and \Cref{cor:belong-shallow-count-edge}.

\begin{definition}[$\ell$-Spoiler Vertex]\label{def:spoiler-vertex-ell}
For $1 < \ell \leq L$, let $\widehat{E}$ denote the set of inner edges that belong to a connected component containing at least one edge from $E^{inner}_\ell$. Let $u$ be a vertex in a connected component that contains at least one edge from $E^{inner}_\ell$. We say a vertex $u$ is an {\em $\ell$-spoiler} if an edge $(u, v)$ is discovered by the algorithm at a time when both $u$ and $v$ already have non-zero degree.
\end{definition}

\begin{claim}\label{clm:ell-spoiler-count}
    Suppose that $1-g(\ell-1)-2\sigma_{\ell-1} \geq 0$. Then, there are at most $O(n^{1-g(\ell-1)-2\sigma_{\ell-1}})$ $\ell$-spoiler vertices with high probability.
\end{claim}
\begin{proof}
    Let $C_1, C_2, \ldots, C_k$ be the underlying undirected connected components of inner edges, where each component contains at least one edge from $E^{inner}_\ell$. Denote by $\widehat{E}$ the set of inner edges in these connected components. By statement (i) of \Cref{cor:total-number-edges-inner-comp}, we have
\begin{align*}
    |\widehat{E}| \leq O\left(n^{1 - g(\ell) + 5\sigma_{\ell-1}/\sigma_\ell}\right).
\end{align*}
Now, consider the scenario where the algorithm queries a pair of vertices and discovers an edge in $\widehat{E}$. There are at most $O\left(n^{1 - \delta + \sigma_{\ell - 1}}\right)$ vertex pairs where one vertex is an endpoint of the queried edge, and the other already has an incident inner edge in the queried subgraph. Note that each query between two vertices that already have a non-zero degree in the queried subgraph, and one of the endpoint is in component where an edge of $E^{inner}_\ell$, results in either a pseudo or real edge of level $\ell-1$, creating two $\ell$-spoiler vertices by \Cref{def:spoiler-vertex-ell}. Therefore, the total number of such vertex pairs is bounded by
\begin{align*}
    O\left(|\widehat{E}| \cdot n^{1 - \delta + \sigma_{\ell-1}}\right).
\end{align*}
If edges are revealed between all these pairs, the total number of ground edges would be at most
\begin{align*}
    O\left(|\widehat{E}| \cdot \rho n^{2 - \delta + \sigma_{\ell-1}}\right),
\end{align*}
where each pair forms a pseudo or real edge of level $\ell - 1$ with probability at most $\rho_{\ell - 1} / (\rho n)$. Hence, the expected number of edges discovered by the algorithm between such pairs is bounded by $O\left(|\widehat{E}| \cdot \rho_{\ell - 1} n^{1 - \delta + \sigma_{\ell-1}}\right)$. Applying the Chernoff bound, we can show that with high probability, the number of $\ell$-spoiler vertices is at most $O\left(|\widehat{E}| \cdot \rho_{\ell - 1} n^{1 - \delta + \sigma_{\ell-1}}\right)$. Combining this with the fact that $\rho_{\ell - 1} = O\left(n^{1 - \sigma_{\ell-1}}\right)$ and the bound on $\widehat{E}$, the total number of $\ell$-spoiler vertices is bounded by
\begin{align*}
    \wt{O}\left(n^{1 - g(\ell) + 5\sigma_{\ell-1}/\sigma_\ell - \delta + 2\sigma_{\ell-1}}\right),
\end{align*}
with high probability. Finally, by statement (ii) of \Cref{obs:g-function-properties}, we have,
\begin{align*}
    1 - g(\ell) - \delta + 5\sigma_{\ell-1}/\sigma_\ell + 2\sigma_{\ell-1} = 1-g(\ell - 1) - 3\sigma_{\ell-1},
\end{align*}
which implies that the total number of $\ell$-spoiler vertices is at most $O(n^{1 - g(\ell-1) - 2\sigma_{\ell-1}})$.
\end{proof}

\begin{claim}\label{clm:ell-spoiler-count-prob}
    Suppose that $1-g(\ell-1)-2\sigma_{\ell-1} < 0$. Then, the probability of having an $\ell$-spoiler vertex is at most $O(n^{1-g(\ell-1)-2\sigma_{\ell - 1}})$. Also, there are at most $\wt{O}(1)$ $\ell$-spoiler vertices with high probability.
\end{claim}
\begin{proof}
    Let $C_1, C_2, \ldots, C_k$ be the underlying undirected connected components of inner edges, where each component contains at least one edge from $E^{inner}_\ell$. Denote by $\widehat{E}$ the set of inner edges in these connected components. We prove the claim by considering the following two cases (note that $1-g(\ell) \neq 0$ by statement (iv) \Cref{obs:g-function-properties}):

\begin{itemize}
    \item \textbf{$1 - g(\ell) > 0$:}  
    In this case, by statement (i) of \Cref{cor:total-number-edges-inner-comp}, we have
    \begin{align*}
        |\widehat{E}| \leq \widetilde{O}\left(n^{1 - g(\ell) + 5\sigma_{\ell-1}/\sigma_\ell}\right),
    \end{align*}
    with high probability. With the exact same proof as the proof of \Cref{clm:ell-spoiler-count}, the probability of having a $\ell$-spoiler vertices is bounded by
\begin{align*}
    \wt{O}\left(n^{1 - g(\ell) + 5\sigma_{\ell-1}/\sigma_\ell - \delta + 2\sigma_{\ell-1}}\right).
\end{align*}
Finally, by statement (ii) of \Cref{obs:g-function-properties}, we have,
\begin{align*}
    1 - g(\ell) - \delta + 5\sigma_{\ell-1}/\sigma_\ell + 2\sigma_{\ell-1} = 1-g(\ell - 1) - 3\sigma_{\ell-1},
\end{align*}
which implies that the probability of having a $\ell$-spoiler vertices is at most $O(n^{1 - g(\ell-1) - 2\sigma_{\ell-1}})$.
\item \textbf{$1 - g(\ell) < 0$:} In this case, by statement (ii) of \Cref{cor:total-number-edges-inner-comp}, we have $k=0$ with probability $1 - O(n^{1-g(\ell)})$ which implies that there is no $\ell$-spoiler vertex. Now suppose that we condition on $k>0$. Then, by statement (iii) of \Cref{cor:total-number-edges-inner-comp}, we have
    \begin{align*}
        |\widehat{E}| \leq \widetilde{O}\left(n^{5\sigma_{\ell-1}/\sigma_\ell}\right).
    \end{align*}
    With the exact same proof as the proof of \Cref{clm:ell-spoiler-count}, the probability of having an $\ell$-spoiler vertex is bounded by  
\begin{align*}
    \widetilde{O}\left(n^{5\sigma_{\ell-1}/\sigma_\ell + 2\sigma_{\ell-1} - \delta}\right).
\end{align*}
Since the probability of having a component containing an edge from $E^{inner}_\ell$ is $O\left(n^{1 - g(\ell)}\right)$, the overall probability of having an $\ell$-spoiler vertex is upper bounded by  
\begin{align*}
    O\left(n^{1 - g(\ell)}\right) \cdot \widetilde{O}\left(n^{5\sigma_{\ell-1}/\sigma_\ell + 2\sigma_{\ell-1} - \delta}\right) = \wt{O}\left(n^{1 - g(\ell) + 5\sigma_{\ell-1}/\sigma_\ell - \delta + 2\sigma_{\ell-1}}\right).
\end{align*}
Thus, combining with statement (ii) of \Cref{obs:g-function-properties}, the probability of having an $\ell$-spoiler vertex is at most $O\left(n^{1 - g(\ell-1) - 2\sigma_{\ell - 1}}\right)$.  
\end{itemize}
In both cases, since the expected number of $\ell$-spoiler vertices is smaller than 1, then, using Chernoff bound, with high probability there are at most $\wt{O}(1)$ $\ell$-spoiler vertices which concludes the proof.
\end{proof}

\begin{definition}[$\ell$-Spoiled Vertex]\label{def:spoiled-vertex-ell}
For $1 < \ell \leq L$, let $\widehat{E}$ denote the set of inner edges that belong to a connected component containing at least one edge from $E^{inner}_\ell$. Let $u$ be a vertex in a connected component that contains at least one edge from $E^{inner}_\ell$. We say a vertex $v$ is called an {\em $\ell$-spoiled} vertex if its shallow subgraph contains any of the following: 
\begin{itemize}
    \item[(i)] an $\ell$-spoiler vertex; or
    \item[(ii)] at least $n^{\delta - 2\sigma_{\ell-1}}$.
\end{itemize}
\end{definition}

\begin{observation}\label{obs:tree-structure-of-unspoiled-ell}
    Let $v$ be a vertex that is not $\ell$-spoiled. Then, the $\ell$-shallow subgraph of $v$ forms a rooted tree with at most $n^{\delta - 2\sigma_{\ell - 1}}$ vertices. Additionally, for every edge $(u,w)$ in the $\ell$-shallow subgraph of $v$, when the algorithm queries this edge, vertex $w$ is a singleton.
\end{observation}
\begin{proof}
    Proof is the same as \Cref{obs:tree-structure-of-unspoiled} by adapting \Cref{def:spoiler-vertex-ell} and \Cref{def:spoiled-vertex-ell}.
\end{proof}

\begin{lemma}\label{lem:ell-spoiled-count}
    Suppose that $1-g(\ell-1)-2\sigma_{\ell-1} \geq 0$. Then, there are at most $O(n^{1-g(\ell-1)-\sigma_{\ell-1}})$ $\ell$-spoiled vertices with high probability.
\end{lemma}

\begin{proof}
    Let $C_1, C_2, \ldots, C_k$ denote the underlying undirected connected components of inner edges, where each component contains at least one edge from $E^{inner}_\ell$. Additionally, let $\widehat{V}$ represent the set of vertices in these components, and let $E(C_i)$ be the set of edges in component $C_i$. By statement (i) of \Cref{cor:total-number-edges-inner-comp}, we have:  
\begin{align*}
    \sum_{i=1}^c |E(C_i)| \leq \widetilde{O}\left(n^{1 - g(\ell) + 5\sigma_{\ell - 1}/\sigma_\ell}\right).
\end{align*}
Applying \Cref{cor:edge-in-shallow-count-extended}, it follows that
\begin{align*}
    \sum_{u \in \widehat{V}} |T^\ell(u)| \leq \widetilde{O}(1) \cdot \sum_{i=1}^k |E(C_i)| \leq \widetilde{O}\left(n^{1 - g(\ell) + 5\sigma_{\ell - 1}/\sigma_\ell}\right).
\end{align*}
Hence, the number of vertices for which $|T^\ell(v)| > n^{\delta - \sigma_{\ell - 1}}$ is at most  
\begin{align*}
    \frac{\sum_{u \in \widehat{V}} |T^\ell(u)|}{n^{\delta - 2\sigma_{\ell - 1}}} 
    &\leq \frac{\widetilde{O}\left(n^{1 - g(\ell) + 5\sigma_{\ell - 1}/\sigma_\ell}\right)}{n^{\delta - 2\sigma_{\ell - 1}}} \\
    &\leq \widetilde{O}\left(n^{1 - g(\ell) - \delta + 5\sigma_{\ell - 1}/\sigma_\ell + 2\sigma_{\ell - 1}}\right) \\
    &\leq O\left(n^{1 - g(\ell - 1) - 2\sigma_{\ell - 1}}\right) & (\text{statement (ii) of \Cref{obs:g-function-properties}}),
\end{align*}  
which implies that there are at most $O(n^{1 - g(\ell - 1) - 2\sigma_{\ell - 1}})$ vertices that satisfy condition (ii) of \Cref{def:spoiled-vertex-ell}.

On the other hand, by \Cref{clm:ell-spoiler-count}, there are at most $O(n^{1-g(\ell-1)-2\sigma_{\ell-1}})$ $\ell$-spoiler vertices. By \Cref{lem:vertex-in-shallow-count-extended}, each vertex is in at most $\wt{O}(1)$ $\ell$-shallow subgraphs. Thus, there are at most $O(n^{1-g(\ell-1)-\sigma_{\ell-1}})$ vertices that satisfy condition (i) of \Cref{def:spoiled-vertex-ell} which completes the proof.
\end{proof}

\begin{lemma}\label{lem:ell-spoiled-count-prob}
    Suppose that $1-g(\ell-1)-2\sigma_{\ell-1} < 0$. Then, the probability of having a $\ell$-spoiled vertex is at most $O(n^{1-g(\ell-1)-2\sigma_{\ell - 1}})$. Also, there are at most $\wt{O}(1)$ $\ell$-spoiled vertices with high probability.
\end{lemma}

\begin{proof}
    By \Cref{clm:ell-spoiler-count-prob}, the probability of having an $\ell$-spoiler vertex is at most $O(n^{1-g(\ell-1)-2\sigma_{\ell-1}})$ which implies that the probability of satisfying condition (i) of \Cref{def:spoiled-vertex-ell} is upper bounded by $O(n^{1-g(\ell-1)-2\sigma_{\ell-1}})$. Also, there are at most $\wt{O}(1)$ $\ell$-spoiler vertices with high probability. By \Cref{lem:vertex-in-shallow-count-extended}, each vertex is in at most $\wt{O}(1)$ $\ell$-shallow subgraphs. Hence, there are at most $\wt{O}(1)$ vertices that satisfy condition (i) of \Cref{def:spoiled-vertex-ell} with high probability.

    Let $C_1, C_2, \ldots, C_k$ denote the underlying undirected connected components of inner edges, where each component contains at least one edge from $E^{inner}_\ell$. Additionally, let $\widehat{V}$ represent the set of vertices in these components, and let $E(C_i)$ be the set of edges in component $C_i$. We prove the claim by considering the following two cases (note that $1-g(\ell) \neq 0$ by statement (iv) \Cref{obs:g-function-properties}):
    \begin{itemize}
        \item $1-g(\ell) > 0$: by statement (i) of \Cref{cor:total-number-edges-inner-comp}, we have  
        \begin{align*}
            \sum_{i=1}^k |E(C_i)| \leq \widetilde{O}\left(n^{1 - g(\ell) + 5\sigma_{\ell-1}/\sigma_\ell}\right),
        \end{align*}
        with high probability. Then, by \Cref{cor:edge-in-shallow-count-extended}, we obtain  
\begin{align*}
    \sum_{u \in \widehat{V}} |T^\ell(u)| \leq \widetilde{O}(1) \cdot \sum_{i=1}^k |E(C_i)| \leq \widetilde{O}\left(n^{1 - g(\ell) + 5\sigma_{\ell - 1}/\sigma_\ell}\right).
\end{align*}
Further, 
    \begin{align*}
        1-g(\ell) + 5\sigma_{\ell-1}/\sigma_{\ell} &= 1 - g(\ell-1) + \delta - 5\sigma_{\ell-1} & (\text{By statement (i) of \Cref{obs:g-function-properties}}) \\
        & = \left(1 - g(\ell - 1) - 2\sigma_{\ell-1} \right) + \left(\delta - 3\sigma_{\ell-1} \right)\\
        & < \delta - \sigma_{\ell-1} & (\text{Since } 1-g(\ell-1)-2\sigma_{\ell-1} < 0),
    \end{align*}
    which implies that there is no component with an edge of $E^{inner}_\ell$ with size $n^{\delta - \sigma_{\ell-1}}$ with high probability. Therefore, no vertex satisfies condition (ii) of \Cref{def:spoiled-vertex-ell} with high probability. 
        \item $1-g(\ell) < 0$: by statement (iii) of \Cref{cor:total-number-edges-inner-comp}, we have  
\(\sum_{i=1}^c |E(C_i)| \leq \widetilde{O}(n^{5\sigma_{\ell-1}/\sigma_{\ell}})\) with high probability. Since \(\delta - 2\sigma_{\ell-1} > 5\sigma_{\ell-1}/\sigma_\ell\), it follows that, with high probability, no component containing an edge from \(E^{\text{inner}}_\ell\) has size \(n^{\delta - 2\sigma_{\ell-1}}\), which implies that with high probability, no vertex satisfies condition (ii) of \Cref{def:spoiled-vertex-ell}.
    \end{itemize}
\end{proof}

\begin{claim}\label{clm:no-cycle-in-small-components}
    Consider the connected components $C'_1, C'_2, \dots, C'_{k'}$ formed by inner edges that do not contain any edges from $E^{\text{inner}}_\ell$. With probability $1 - O(n^{-\delta + 2\sigma_{\ell - 1} + 10\sigma_{\ell-1}/\sigma_\ell})$, all such components are acyclic, i.e., they form trees.
\end{claim}

\begin{proof}
    By \Cref{clm:small-component-in-intersect-tree}, each component $C'_i$ satisfies $|C'_i| \leq O(n^{5\sigma_{\ell-1}/\sigma_\ell})$ with high probability for all $i$. Also, since the total number of inner edges is at most $O(n^{1-\delta + \sigma_{\ell - 1}})$ by \Cref{clm:max-inner-edges}, it follows that $k' \leq O(n^{1-\delta + \sigma_{\ell - 1}})$. Hence, we obtain the bound  
    \begin{align*}
    \sum_{i=1}^{k'} |C_i'|^2 \leq O(n^{1-\delta + 2\sigma_{\ell - 1} + 10\sigma_{\ell-1}/\sigma_\ell}).
    \end{align*}
    Fix a component $C_i'$. Now, consider the process of adding edges to this component one by one in the order they are queried by the algorithm. To form a cycle, at some point during this process, there must be an inner edge within a component that has not yet been queried. Each newly added edge merges two existing components. If the sizes of these components are $x$ and $y$, then apart from the inner edge discovered by the algorithm, there are potentially $xy$ other new pairs within the newly created component that could contain an inner edge which could potentially create a cycle. Ultimately, the total number of such pairs in a component is $O(|C_i'|^2)$. 

    Since there are $\rho n$ ground edges between each of these pairs, and each is marked as a pseudo edge of level $\ell-1$ with probability $\rho_{\ell-1}/(\rho n)$, the expected number of psuedo edges of level $\ell - 1$ in the component, ignoring those found by the algorithm, is at most $O(\rho_{\ell - 1} |C_i'|^2)$. Summing over all components, this expectation is given by  
    \begin{align*}
    \sum_{i=1}^{k'} O(\rho_{\ell - 1}|C_i'|^2) = O(n^{-\delta + 2\sigma_{\ell - 1} + 10\sigma_{\ell-1}/\sigma_\ell}).
    \end{align*}  
    This completes the proof.
\end{proof}

For the remainder, we assume that every connected component of inner edges without an edge from $E^{inner}_\ell$ is acyclic. By \Cref{clm:no-cycle-in-small-components}, the probability of this assumption failing is 
\begin{align*}
    O(n^{-\delta + 2\sigma{\ell - 1} + 10\sigma_{\ell-1}/\sigma_\ell}) = o(1).
\end{align*}    
Furthermore, since the hierarchy has a constant number of levels, this holds for all levels with probability $1 - o(1)$.

\begin{definition}[$\ell$-Spoiled Edge]\label{def:ell-spoiled-edge}
    Let $(u,v)$ be a directed inner edge. We say $(u,v)$ is a spoiled edge if $u \in A_r$ in level $\ell - 1$ of hierarchy and at least one of the following condition holds:
    \begin{itemize}
        \item[(i)] $v$ is an $\ell$-spoiled vertex; or
        \item[(ii)] $u$ has at least $n^{\sigma_{\ell-1}}/3$ $\ell$-spoiled neighbors in the queried subgraph.
    \end{itemize}
\end{definition}

\begin{lemma}\label{lem:ell-spoiler-edge-bound}
    The following two statements hold regarding the number of $\ell$-spoiled edges:
    \begin{itemize}
        \item[(i)] If $1 - g(\ell - 1) > 0$, then the number of $\ell$-spoiled edges is at most $O(n^{1-g(\ell - 1)})$ with high probability.
        \item[(ii)] If $1 - g(\ell - 1) < 0$, then no $\ell$-spoiled edge exists with probability $1-O(n^{1-g(\ell-1)})$.
        \item[(iii)] If $1 - g(\ell - 1) < 0$, then the number of $\ell$-spoiled edge is at most  $\wt{O}(1)$ with high probability.
    \end{itemize}
\end{lemma}
\begin{proof}
    Let $\widehat{E}$ be the set of $\ell$-spoiled edges. We prove each statement separately:
    \begin{itemize}
    \item[(i):] First consider the case that $1 - g(\ell - 1) -  2\sigma_{\ell - 1} \geq 0$. By \Cref{lem:ell-spoiled-count}, there are at most $O(n^{1-g(\ell - 1) -\sigma_{\ell - 1}})$ $\ell$-spoiled vertices with high probability. Furthermore, each vertex has an indegree of $\wt{O}(1)$ by \Cref{clm:max-in-deg-top} with high probability. Thus, there are at most $O(n^{1-g(\ell - 1)})$ edges that satisfy condition (i) of \Cref{def:ell-spoiled-edge}.

    On the other hand, if vertex $u$ satisfies condition (ii) of \Cref{def:ell-spoiled-edge}, it must have at least $n^{\sigma_{\ell - 1}}/4$ outgoing edges $(u,w)$ where $w$ is $\ell$-spoiled. But using \Cref{clm:max-in-deg-top}, each $\ell$-spoiled vertex has an indegree of $\wt{O}(1)$. Combining with the fact that the total number of $\ell$-spoiled vertices is $O(n^{1-g(\ell - 1) -\sigma_{\ell - 1}})$ implies that the total number of edges satisfy condition (ii) of \Cref{def:ell-spoiled-edge} is at most $O(n^{1-g(\ell - 1)})$.

    \item[(ii):] If $1-g(\ell - 1) < 0$ then we have $1 - g(\ell-1) - 2\sigma_{\ell - 1} < 0$. Then, the probability of having a $\ell$-spoiled vertex is at most $O(n^{1-g(\ell-1)-2\sigma_{\ell - 1}})$ by \Cref{lem:ell-spoiled-count-prob}. Since $2\sigma_{\ell - 1} > 0$ and according to \Cref{def:ell-spoiled-edge},  the probability of having an $\ell$-spoiled edge is at most $O(n^{1-g(\ell-1)})$.
    \item[(iii):] If $1-g(\ell - 1) < 0$ then we have $1 - g(\ell-1) - 2\sigma_{\ell - 1} < 0$. Then, by \Cref{lem:ell-spoiled-count-prob}, there are at most $\wt{O}(1)$ $\ell$-spoiled vertices with high probability. Therefore, with high probability, there are at most $\wt{O}(1)$ $\ell$-spoiled edges.
\end{itemize}
\end{proof}

\begin{lemma}\label{lem:small-component-focs-modification}
    Let $(u,v)$ be a directed inner edge within a connected component $C$ that does not contain any edge from $E^{inner}_\ell$. Suppose $u \in A_r$ and $v$ belongs to $\{A_r, B_r, D_r\}$ at level $\ell-1$ of the hierarchy. Define $\overline{C}$ as the component containing $v$ after removing the edge $(u,v)$. Let $\mc{L}(v)$ and $\mc{L}'(v)$ represent an arbitrary label for $v$ from $\{A_r, B_r, D_r\}$ and the entire queried subgraph of inner edges, excluding $\overline{C}$. Then, we have:
    \begin{align*}
        \Pr[\overline{C} \mid \mc{L}(v)] \leq \left(1 + O(n^{\sigma_{\ell-1}-\delta}) \right)^{|\overline{C}|} \cdot \Pr[\overline{C} \mid \mc{L}'(v)].
    \end{align*}
\end{lemma}

\begin{lemma}\label{lem:large-component-focs-modification}
    Let $v$ be a vertex that is not $\ell$-spoiled and belongs to a connected component that contains at least one edge from $E^{\text{inner}}_\ell$. Suppose $v \in \{A_r, B_r, D_r\}$ at level $\ell-1$ of the hierarchy. Let $\mc{L}(v)$ and $\mc{L}'(v)$ represent an arbitrary label for $v$ from $\{A_r, B_r, D_r\}$ and the entire queried subgraph of inner edges, excluding the $\ell$-shallow subgraph of $v$. Then, we have:
    \begin{align*}
        \Pr[T^\ell(v) \mid \mc{L}(v)] \leq \left(1 + O(n^{\sigma_{\ell-1}-\delta}) \right)^{|T^\ell(v)|} \cdot \Pr[T^\ell(v) \mid \mc{L}'(v)].
    \end{align*}
\end{lemma}

The proofs of these two lemmas are postponed to a later section as they are intricate, lengthy, and largely independent of the discussion in this section.

\begin{lemma}\label{lem:non-spoiled-edge-non-bias-ell}
    Let $e$ be a directed inner edge that is not $\ell$-spoiled. Then, it holds that $p_e^{(\ell-1)-inner} \leq 10 n^{\sigma_{\ell - 2} - \sigma_{\ell-1}}$.
\end{lemma}
\begin{proof}
Let $e = (u,v)$ be a directed inner edge from $u$ to $v$. First, if $u \notin A_r$, we have $p^{(\ell-1)-inner}_e = 0$. Also, if $e$ is a pseudo edge of level $\ell - 1$, then by construction, the probability that we mark $e$ as level $\ell - 2$ is $\rho_{\ell - 2}/\rho_{\ell - 1} = n^{\sigma_{\ell - 2} - \sigma_{\ell-1}}$ independently from other edges.

Hence, suppose that $u \in A_r$ and $e$ is a real edge for the rest of the proof. According to the construction, $u$ has at least $6n^{\sigma_{\ell - 1}}/7$ inner edges that are real edges such that the other endpoint has label $A_r \cup B_r$. Also, $u$ has at most $\wt{O}(1)$ incoming edge by \Cref{clm:max-in-deg-top}. Since $e$ is not an $\ell$-spoiled edge, $u$ has at most $n^{\sigma_{\ell-1}}/3$ neighbors in the queried subgraph that are $\ell$-spoiled. Let $V_u$ be the set of neighbors of $u$ using inner edges that are real edges in the original graph such that their label is $A_r \cup B_r$ and either they are singleton or direct children of $u$ in the queried subgraph. By the above argument, we have $|V_u| \geq n^{\sigma_{\ell - 1}}/2$. Also, note that $v \in V_u$. Now we provide an upper bound on the probability that $e$ belongs to level $\ell-2$, or in other words, $v \in A_r$.  We prove this upper bound using \Cref{lem:small-component-focs-modification} and \Cref{lem:large-component-focs-modification}.

Consider a labeling profile $\mathcal{P}$ of all vertices $V_u$ such that $\mathcal{P}(v) = A_r$. By the construction of our input distribution, since $u \in A_r$, at most $O(d_{\ell-2}) = O(n^{\sigma_{\ell-2}})$ vertices in $V_u$ belong to $A_r$. We generate $\Omega(n^{\sigma_{\ell - 1}})$ new profiles $\mathcal{P}'$ where $\mathcal{P}'(v) \neq A_r$. For each vertex $w \in V_u$ with $\mathcal{P}(w) = B_r$, we create a new profile $\mathcal{P}'$ by setting $\mathcal{P}'(z) = \mathcal{P}(z)$ for $z \notin \{v, w\}$, $\mathcal{P}'(w) = A_r$, and $\mathcal{P}'(v) = B_r$. 

Since $v$ and $w$ are not $\ell$-spoiled vertices, they either belong to a component with no edges in $E^{inner}_\ell$ or have an $\ell$-shallow subgraph that satisfies the conditions in \Cref{def:spoiled-vertex-ell}. In both cases, the probability of querying the same shallow subgraph or connected component in the new labeling profile remains the same up to a factor of
\begin{align*}
    \left(1 + O(n^{\sigma_{\ell - 1}-\delta}) \right)^{n^{\delta - 2\sigma_{\ell-1}}},
\end{align*}
by \Cref{lem:large-component-focs-modification} and \Cref{lem:small-component-focs-modification} since either $|T^\ell(v)| \leq n^{\delta-2\sigma_{\ell-1}}$  (resp. $|T^\ell(w)| \leq n^{\delta-2\sigma_{\ell-1}}$) or the component that $v$ or $w$ belongs to has size of at most $O(n^{5\sigma_{\ell-1}/\sigma_{\ell}})$ which is smaller than $O(n^{\delta-2\sigma_{\ell-1}})$ for large enough $n$.  Therefore, the probability of generating profiles $\mathcal{P}$ and $\mathcal{P}'$ differs by at most  
\begin{align*}
\left(1 + O(n^{\sigma_{\ell - 1} - \delta}) \right)^{n^{\delta - 2\sigma_{\ell-1}}} \cdot \left(1 + O(n^{\sigma_{\ell - 1} - \delta}) \right)^{n^{\delta - 2\sigma_{\ell-1}}}
&\leq \left(1 + O(n^{\sigma_{\ell - 1} - \delta}) \right)^{2n^{\delta - 2\sigma_L}} \\
&\leq 1 + o(1).
\end{align*}  
Now, construct a bipartite graph $H = (P_1, P_2, E_P)$ of labeling profiles, where $P_1$ consists of all profiles $\mathcal{P}$ with $\mathcal{P}(v) = A_r$, and $P_2$ contains all profiles $\mathcal{P}'$ with $\mathcal{P}'(v) = B_r$. Add an edge between $\mathcal{P} \in P_1$ and $\mathcal{P}' \in P_2$ if $\mathcal{P}$ can be transformed into $\mathcal{P}'$ through the process described earlier.  

For any profile $\mathcal{P} \in P_1$, we have $\deg_H(\mathcal{P}) \geq |V_u|/2 \geq n^{\sigma_{\ell - 1}}/4$ since at least $|V_u|/2$ vertices in $V_u$ belong to $B_r$. In contrast, for any profile $\mathcal{P}' \in P_2$, $\deg_H(\mathcal{P}') \leq 2n^{\sigma_{\ell-2}}$ because, by the input distribution, at most $2d_{\ell-2} = 2n^{\sigma_{\ell-2}}$ vertices $w$ in $V_u$ satisfy $\mathcal{P}'(w) = A_r$. Therefore,  
\begin{align*}
p_{e}^{\text{inner}} 
&\leq (1 + o(1)) \cdot \frac{|P_1|}{|P_2|} \\
&\leq (1 + o(1)) \cdot \frac{2n^{\sigma_{\ell-2}}}{n^{\sigma_{\ell - 1}}/4} \\
&\leq (1 + o(1)) \cdot 8n^{\sigma_{\ell-2} - \sigma_{\ell - 1}} \\
&\leq 10n^{\sigma_{\ell-2} - \sigma_{\ell - 1}},
\end{align*}  
which concludes the proof. 
\end{proof}

Now we finish the proof of \Cref{lem:advantage-from-top}.

\begin{proof}[Proof of \Cref{lem:advantage-from-top}]
    We need to prove the lemma for $\ell - 1$ using \Cref{lem:small-square-component}. According to \Cref{lem:non-spoiled-edge-non-bias-ell}, which relies on \Cref{lem:small-square-component} for $\ell$, if an edge is not an $\ell$-spoiled edge, then we have $p_e^{(\ell-1)-inner} \leq 10 n^{\sigma_{\ell - 2} - \sigma_{\ell-1}}$. Thus, by applying each statement of \Cref{lem:ell-spoiler-edge-bound}, we obtain the proof for the corresponding item in the lemma.
\end{proof}

\section{Identical Distribution for Acyclic Subgraphs}

In this section, we prove \Cref{lem:coupling-high-level}, \Cref{lem:large-component-focs-modification}, and \Cref{lem:small-component-focs-modification}. Our proof is inspired by the results of \cite{BehnezhadRR-FOCS23, behnezhad2024approximating}. The main difference is that, in our setting, we need to incorporate pseudo edges into the coupling argument. Furthermore, instead of a regular graph, we consider an \ER{} graph, which necessitates slight adjustments to the proof.  We will demonstrate that the same result holds under this construction. To do so, we step by step revisit the original proof, making the necessary modifications to account for pseudo edges. For now, suppose we are at a fixed level $\ell$ in the hierarchy. We begin by introducing the relevant notation and establishing the essential tools required to achieve the main objective of this section. Throughout the remainder of the section, assume that $d = d_\ell/d_{\ell-1} = \Theta(n^{\sigma_{\ell} - \sigma_{\ell-1}})$. Also, when we say a vertex belongs to layer $i$, we mean it belongs to $A_i \cup B_i \cup D_i$. Moreover, we consider vertices of $S$ as layer 0.

\begin{definition}[Special Edge]\label{def:special-edge}
    We define an edge $(u,v)$ as \emph{special} if any of the following conditions hold:
    \begin{itemize}
        \item $(u,v)$ belongs to gadget between $S$ and $B_1$,
        \item $(u,v)$ belongs to gadgets between $B_i$ and $A_{i-1}$ for $i \in (1, r]$.
        \item $(u,v)$ belongs to a gadget that exists only in $\yesdist^\ell$ or $\nodist^\ell$,
        \item $(u,v)$ belongs to a gadget between $D_i^{j}$ and $D_i^{j'}$ for $i \in [r]$, where $j \in \{1, 3\}$ and $j' \in \{2, 4\}$. At the base level, we consider edges between $D_i^1$ and $D_i^2$ for $i \in [r]$.
    \end{itemize}
\end{definition}

\begin{definition}[Special Pseudo Edge of Layer $i$]\label{def:special-pseudo-edge}
Let \( T \) be a rooted tree with root \( u \), where \( u \in \{A_{r}, B_{r}, D_{r} \} \). Consider a vertex \( v \) and suppose there exists a path in \( T \) leading to \( v \) such that it contains no mixer vertices (as defined in \Cref{def:mixer-vertex}) and has exactly \( k \) special crossings. Additionally, assume that \( v \notin B_{r-k} \).  Now, consider the pseudo edges between \( v \) and \( A_{i-1} \). By construction, there are no real edges connecting \( v \) to vertices in \( A_{i-1} \). However, in expectation, there are \( \rho_\ell N_\ell \) pseudo edges of level \( \ell \) between \( v \) and vertices in \( A_{i-1} \). Each pseudo edge between \( v \) and \( A_{i-1} \) is classified as a special pseudo edge of layer \( i \) with probability \( \log n / \rho_{\ell} \). Thus, in expectation, \( v \) has \( \log n \) special pseudo edges of layer \( i \) in \( A_{i-1} \).  Furthermore, if \( v \in B_{r-k} \), the expected number of pseudo edges between \( v \) and vertices in layers no lower than \( r-k \) is \( \log n \) greater than for other labels in layer no lower than $r-k$. Therefore, we randomly select \( \log n \) edges in expectation to be marked as special crossings.  
\end{definition}

\begin{definition}[Special Crossing]\label{def:special-crossing}
    Let $T$ be a rooted tree with root $u$. Let $e$ be an edge in the tree. We call edge $e$ as {\em special crossing} if any of the following hold:
    \begin{itemize}
        \item $e$ is a special edge by \Cref{def:special-edge},
        \item $e$ is a special pseudo edge of layer $i$ for some $i \in [r]$ by \Cref{def:special-pseudo-edge}.
    \end{itemize}
\end{definition}

\begin{definition}[Mixer Vertex]\label{def:mixer-vertex}
Let $T$ be a rooted tree with root $u$, where $u \in \{A_{r}, B_{r}, D_{r} \}$. Consider a vertex $v$ in $T$, and suppose there are $k$ special crossings along the path from $u$ to $v$. We say $v$ is a {\em mixer vertex} if and only if $k < r - 1$ and $v \in \bigcup_{i=1}^{r-k-1} (A_i, B_i, D_i)$.
\end{definition}

Note that the definition of a mixer vertex relies on special pseudo edges, while the definition of special pseudo edges depends on mixer vertices. However, each of these definitions uses the other only up to a distance of \( i \) from the root of the tree to establish the definition for distance \( i+1 \).

\begin{claim}\label{clm:special-crossings-and-layers-reachable}
    Let $T$ be a rooted tree with root $u$, where $u \in \{A_{r}, B_{r}, D_{r} \}$. Let $v$ be a vertex in the tree where $v$ belongs to layer $i$. Any path from $u$ to vertex $v$ that does not pass through a mixer vertex must contain at least $r - i$ special crossings.  
\end{claim}

\begin{proof} 
We claim that any path reaching a vertex $v$ in a layer smaller than or equal $i$, without passing through a mixer vertex, must contain at least $r - i$ special crossings. We prove this claim by induction.  For the base case $r = i$, the statement is trivial. Now, suppose the claim holds for all $j > i$, and we aim to prove it for $i$. Consider the first edge in the path where the algorithm reaches a vertex in layer $i$ or smaller for the first time—let this vertex be $w$, reached via edge $e$. We analyze two possible scenarios:  
\begin{itemize}
    \item $e$ is a real edge: By the construction of the gadgets, either $e$ is a special edge between $B_{i+1}$ and $A_i$, or $w$ belongs to $D_i$, in which case $w$ is a mixer vertex by \Cref{def:mixer-vertex}.  
    \item $e$ is a pseudo edge: If $e$ is not a special pseudo edge of layer $i$, then $w$ must be a mixer vertex according to \Cref{def:mixer-vertex}.  
\end{itemize}  
In both cases, when the algorithm reaches a vertex in layer $i$ or lower for the first time, it requires at least one additional special crossing. Applying the induction hypothesis for $i+1$ completes the proof.     
\end{proof}

\begin{lemma}\label{lem:mixer-vertex-in-tree}  
Let $T$ be a rooted tree queried by the algorithm, with its root in $\{A_r, B_r, D_r\}$. With probability at least $1 - \wt{O}(|V(T)|/d^{r-1})$, any path from the root to a vertex in $T$ that does not cross a mixer vertex contains at most $r - 2$ special crossings.  
\end{lemma}
\begin{proof}
We show that with probability $\wt{O}(1/d^{r-1})$, any path found by the algorithm that contains at least $r-1$ special crossings does not pass through a mixer vertex. Let $i$ denote the index of a mixer vertex located in layer $i$.  Suppose there exists an oracle that, whenever the algorithm discovers a path with at least $r-1$ special crossings, performs one of the following actions:  
\begin{itemize}
    \item confirms that the path does not contain any mixer vertices, or  
    \item reveals the mixer vertex with the smallest index along the path.  
\end{itemize}  
This oracle is introduced purely for analytical purposes; the algorithm itself does not have access to this additional power. However, we demonstrate that even with this advantage, it remains difficult for the algorithm to find a path with $r-1$ special crossings without encountering a mixer vertex on the path.

Now, consider the first path discovered by the algorithm that contains $r-1$ special crossings. Let us examine the moment during the query process when the algorithm detects $r-2$ special crossings for the first time. Suppose that, up to this point, the path does not contain a mixer vertex with index $1$. This implies, by \Cref{clm:special-crossings-and-layers-reachable}, that the path does not reach any vertex in layer $1$ or below.  At this step, when the algorithm queries the next edge (ignoring queries that do not result in either a pseudo edge or a real edge), the probability of encountering another special crossing is $\wt{O}(1/d)$, based on the construction and the definition of special crossings in \Cref{def:special-crossing}. However, the probability of encountering a mixer vertex with index $1$ is $\Theta(1)$, according to the same construction. Consequently, the probability that the path reaches another special crossing before encountering a mixer vertex with index $1$ is $\wt{O}(1/d)$.  

For the sake of analysis, we introduce an additional power for the algorithm: if the oracle has already revealed half of the mixer vertices of a given index $i$ that are direct children of a vertex $v$, then the oracle will either:  
\begin{itemize}
    \item reveals a mixer vertex closer to the root with an index larger than $i$, if such a mixer vertex exists, or  
    \item construct a path with $r-1$ special crossings that avoids mixer vertices altogether, immediately terminating the process.  
\end{itemize}  
From the previous argument, we conclude that an $\wt{O}(1/d)$ fraction of paths does not have a mixer vertex with index $1$. Furthermore, when the algorithm finds $\Omega(d)$ direct children of a vertex that are mixer vertices with index $1$, the oracle reveals the next mixer vertex with a higher index. As a result, the proportion of paths discovered by the algorithm that do not contain a mixer vertex of index $1$ remains at most an $\wt{O}(1/d)$ fraction of all paths.  Further, observe that once a mixer vertex $z$ is revealed by the oracle, any further queries below $z$ become not useful for constructing a path with $r-1$ special crossings while avoiding mixer vertices. This is because, for any path that crosses $z$, the mixer vertex with the highest index that the oracle can reveal will always be $w$ itself.  

Let us consider the set of paths that do not pass through a mixer vertex of index $1$. By a similar argument, an $\wt{O}(1/d)$ fraction of these paths also avoids encountering a mixer vertex of index $2$. Specifically, the probability that a path reaches the $(r-2)$-th special crossing before encountering a mixer vertex of index $2$ is $\wt{O}(1/d)$. As a result, the fraction of paths that do not contain a mixer vertex of index $2$ or lower is at most $\wt{O}(1/d^2)$.  

Applying the exact same argument, the probability that a path avoids all mixer vertices with an index up to $i$ is at most $\wt{O}(1/d^i)$. Therefore, the fraction of paths that do not contain any mixer vertex is bounded by $\wt{O}(1/d^{r-1})$. Since the total number of paths originating from the root is at most $O(|V(T)|)$, it follows that with probability at least $1 - \wt{O}(|V(T)| / d^{r-1})$, every path with more than $r-2$ special crossings must pass through at least one mixer vertex.  
\end{proof}

\begin{lemma}[Modification of Lemma 6.7 of \cite{BehnezhadRR-FOCS23} with Pseudo Edges]\label{lem:same-tree-different-labels}
    Consider the highest level of hierarchy. Let $T$ be a rooted tree that is queried by the algorithm where the root of the tree is in $\{A_r, B_r, D_r \}$. Also, suppose that we condition on having a mixer vertex on all paths that contain at least $r-1$ special crossings. Then, the probability of seeing the same tree is equal for all possible roots in $\{A_r, B_r, D_r \}$ up to $(1 + O(n^{\sigma_L - \delta}))^{|T|}$ multiplicative factor.
\end{lemma}
\begin{proof}
    The proof follows the same outline as Lemma 6.7 in \cite{BehnezhadRR-FOCS23}, with a modification to account for pseudo edges. It is based on a one-to-one coupling argument: for each tree queried by the algorithm, if \( l_1 \) is the label of the tree, the algorithm will observe the same tree with almost the same probability if it starts from label \( l_2 \).

    First, note that according to our definition of special crossings \Cref{def:special-crossing}, if we query a pair of vertices \( (u,v) \) and another pair \( (w,z) \), the probability of each being a special crossing is the same. This is because, at any given moment, each vertex has \( 2\log n \) expected special crossing neighbors, assuming we condition on the fact that no labels have been revealed yet in the coupling.

    Let \( L_i = \{A_i, B_i, D_i\} \) for each \( i \). Also, define  
    \[G_i = G\left[\bigcup_{j=i}^{r} L_j\right].\]
    
    Let \( u \) and \( v \) be two distinct vertices in \( G_i \). An important property of our input distribution is that if we query a neighbor of \( u \) and \( v \) via a real edge, and the queried edge is not a special edge, then the probability that the queried neighbor belongs to \( G_i \) is the same for both \( u \) and \( v \). Similarly, if we query a neighbor of \( u \) and \( v \) via a pseudo edge, and the queried edge is not a special pseudo edge, then the probability that the queried neighbor belongs to \( G_i \) remains the same for both \( u \) and \( v \). These properties follow directly from the structure of the input construction, as formalized in \Cref{cor:same-deg-dist-l} and \Cref{cor:same-deg-dist-pseudo}.  

For a vertex \( u \) in a tree \( T \), if there is no mixer vertex on its path to the root, we define the {\em progress} of \( u \), denoted as \( p_u \), which represents the number of special crossings along the path from the root to \( u \). Under the assumptions of the lemma statement, we have  
\[
0 \leq p_u < r - 1, \quad \text{for all } u \in T.
\]
We claim that for a vertex $u$ where there is no mixer vertex on its path to route, we have $u \in V(G_{r-p_u})$. This is because, according to the construction and \Cref{def:special-crossing}, the only way to go down one layer is through special crossings or mixer vertices.

    Let $\mc{L}_1$ be a labeling of vertices in $T$ where the root is labeled with $l_1$. We construct the labeling $\mc{L}_2$ from top to the bottom of the tree such that the label of root is $l_2$. Consider the edge $(u,v)$ with direction from $u$ to $v$ when considering the rooted tree. So at this point we have a label for $u$ in $\mc{L}_2$ and we want to chose the label of $v$ in $\mc{L}_2$. We maintain the invariant that progress of each vertex is the same in both labeling $\mc{L}_1$ and $\mc{L}_2$. We consider five possible cases for edge $e$ for our coupling step:
    \begin{itemize}
        \item Edge $e$ is special crossing: in this case, we know $v$ is not labeled $S$ in $\mc{L}_1$ as in order to reach $S$ vertex, at least $r-1$ special crossings are needed or the path contains mixer vertex. Thus, suppose that $v$ is not an $S$ vertex. In labeling $\mc{L}_2$, we assume that edge $(u,v)$ is a special crossing and we choose the label of $v$ accordingly based on label of $u$ in $\mc{L}_2$. Since each vertex has $2\log n$ expected number of special crossings ($\log n$ for special edges and $\log n$ for special pseudo edges), no matter what is the label of $u$ in $\mc{L}_1$ and $\mc{L}_2$, the distribution of their special edges are the same. Also, the invariant remains valid since, in this case, \( p_v = p_u + 1 \), and the progress of vertex \( u \) is the same in both $\mc{L}_1$ and $\mc{L}_2$ .

        \item Vertex $v$ is a mixer vertex in $\mc{L}_1$ and $e$ is a real edge: in this case, label of $v$ must be among $\bigcup_{i=1}^{r-p_u-1} D_i$ as the only possible gadgets that add real edges between vertices of $G_{r-p_u}$ and vertices of layers lower than $p_u$ are those to vertices in $\bigcup_{i=1}^{r-p_u-1} D_i$. Since the distribution of neighbors of vertices in $G_{r-p_u}$ to $\bigcup_{i=1}^{r-p_u-1} D_i$ are the same, we assign the same label in $\mc{L}_1$ to $\mc{L}_2$ for vertex $v$. Moreover, for the subtree rooted at \( v \), we assume that all labels remain the same since the label of \( v \) is identical at this point in both \( \mc{L}_1 \) and \( \mc{L}_2 \). Furthermore, the invariant remains valid since we have not used any special crossings in labeling $\mc{L}_2$.
        
        \item Vertex $v$ is a mixer vertex in $\mc{L}_1$ and $e$ is a pseudo edge: in this case, we assign the same label in $\mc{L}_1$ to $\mc{L}_2$ for vertex $v$. The reason is that the distribution of neighbors of vertices in $G_{r-p_u}$ to any type of label in layer lower than $p_u$ is the same if we ignore special pseudo edges. Moreover, for the subtree rooted at \( v \), we assume that all labels remain the same since the label of \( v \) is identical at this point in both \( \mc{L}_1 \) and \( \mc{L}_2 \). Furthermore, the invariant remains valid since we have not used any special crossings in labeling $\mc{L}_2$.
        
        \item Edge $e$ is a real edge which is not special and belongs to $G_{r-p_u}$: since the distribution of neighbors of vertices in induced subgraph of $G_{r-p_u}$ using real edges is the same if we ignore special crossings, we randomly choose one of the real edge of $u$ as the possible label for $v$ in $\mc{L}_2$. Furthermore, the invariant remains valid since we have not used any special crossings in labeling $\mc{L}_2$.
        
        \item Edge $e$ is a pseudo edge which is not special and belongs to $G_{r-p_u}$: since the distribution of neighbors of vertices in induced subgraph of $G_{r-p_u}$ using pseudo edges is the same if we ignore special crossings, we randomly choose one of the pseudo edges of $u$ as the possible label for $v$ in $\mc{L}_2$. Furthermore, the invariant remains valid since we have not used any special crossings in labeling $\mc{L}_2$.
    \end{itemize}

    It is also important to note that the number of non-singleton vertices in level $L$ is at most $O(n^{\sigma_L - \delta})$ since there are at most $O(n^{1 + \sigma_L - \delta})$ non-singleton vertices by \Cref{lem:total-edges-discovered}. Since the total number of steps is at most \( |T| \), the probability that the new labeling still forms a forest differs by a multiplicative factor of at most $(1 + O(n^{\sigma_L - \delta}))^{|T|}$.
\end{proof}

\begin{proof}[Proof of \Cref{lem:coupling-high-level}]
    By \Cref{obs:tree-structure-of-unspoiled}, since $v$ is not spoiled, then, the shallow subgraph of $v$ is rooted. First, we bound the failure probability of the event in \Cref{lem:mixer-vertex-in-tree}. More specifically, we upper bound the probability of having a path with at least $r-1$ special crossings without any mixer vertex. Since $v$ is not a spoiled vertex, we have $|T(v)| \leq n^{\delta - 2\sigma_L}$. Therefore,  the probability of having a path with $r-1$ crossing without a mixer vertex is at most
    \begin{align*}
        \wt{O}\left( \frac{|T(v)|}{d^{r-1}} \right) &= \wt{O}\left( \frac{n^{\delta - 2\sigma_L}}{d^{r-1}} \right) & (\text{Since }|T(v)| \leq n^{\delta - 2\sigma_L})\\
        & = \wt{O}\left( \frac{n^{\delta - 2\sigma_L}}{n^{(\sigma_L - \sigma_{L-1})\cdot (r-1)}} \right) & (\text{Since }d = \Theta(n^{\sigma_L - \sigma_{L-1}}))\\
        & = \wt{O}\left( \frac{n^{\delta - 2\sigma_L}}{n^{r\sigma_L/4}} \right) & (\text{Since }\sigma_L - \sigma_{L-1} \geq \sigma_L / 2 \text{ and } r-1 \geq r/2)\\
        & = \wt{O}\left( n^{\delta - 2\sigma_L - 10 /(4\delta)} \right) & (\text{Since } r \sigma_L \geq 10/\delta)\\
        & = O\left( n^{\delta - 2\sigma_L - 2} \right)  & (\text{Since } \delta \leq 1 \text{ and } n \text{ is sufficiently large enough})\\
        & = O(n^{-1})
    \end{align*}

    Note that we condition on the labels of all vertices except those in the shallow subgraph of vertex $v$. However, in the coupling described in \Cref{lem:same-tree-different-labels}, there is no such conditioning on vertex labels. Since the total number of vertices whose labels we are conditioning on is $O(n^{1-\delta + \sigma_{L}})$, the probability shift in each step of the coupling in \Cref{lem:same-tree-different-labels} is at most $O(n^{1-\delta + \sigma_{L}} / n) = O(n^{\sigma_{L} - \delta})$. Also, we condition on the high probability event of \Cref{lem:mixer-vertex-in-tree}. Moreover, since the number of steps in the coupling is $|T(v)| - 1$, the total probability shift is upper bounded by 
    \begin{align*}
\left(1 + O(n^{\sigma_{L}-\delta})\right)^{|T(v)|} \cdot \left(1 + O(n^{\sigma_{L}-\delta}) \right)^{|T(v)|-1}\cdot  \left(1 + O(n^{-1}) \right) \leq \left(1+O(n^{\sigma_{L}-\delta})\right)^{|T(v)|},\\
\end{align*}
where the inequality is followed by the fact that $-1 < \sigma_L - \delta$, which completes the proof.
\end{proof}

\begin{lemma}[Modification of Lemma 6.7 of \cite{BehnezhadRR-FOCS23} with Pseudo Edges]\label{lem:same-tree-different-labels-1}
    Consider the level $\ell$ of the hierarchy. Let $T$ be a rooted tree that is queried by the algorithm where the root of the tree is in $\{A_r, B_r, D_r \}$. Also, suppose that we condition on having a mixer vertex on all paths that contain at least $r-1$ special crossings. Then, the probability of seeing the same tree is equal for all possible roots in $\{A_r, B_r, D_r \}$ up to $(1 + O(n^{\sigma_{\ell - 1} - \delta}))^{|T|}$ multiplicative factor.
\end{lemma}
\begin{proof}
    Proof is the same as \Cref{lem:same-tree-different-labels} with proper parameters for level $\ell$.
\end{proof}

\begin{proof}[Proof of \Cref{lem:large-component-focs-modification}]
    By \Cref{obs:tree-structure-of-unspoiled-ell}, since $v$ is not $\ell$-spoiled, then, the $\ell$-shallow subgraph of $v$ is rooted. With the exact same argument as proof of \Cref{lem:coupling-high-level}, the failure probability of the event in \Cref{lem:mixer-vertex-in-tree} is at most
    \begin{align*}
        \wt{O}\left( \frac{|T(v)|}{d^{r-1}} \right) &= \wt{O}\left( \frac{n^{\delta - 2\sigma_{\ell - 1}}}{d^{r-1}} \right) & (\text{Since }|T(v)| \leq n^{\delta - 2\sigma_{\ell - 1}})\\
        & = \wt{O}\left( \frac{n^{\delta - 2\sigma_{\ell - 1}}}{n^{(\sigma_\ell - \sigma_{\ell-1})\cdot (r-1)}} \right) & (\text{Since }d = \Theta(n^{\sigma_\ell - \sigma_{\ell-1}}))\\
        & = \wt{O}\left( \frac{n^{\delta - 2\sigma_\ell}}{n^{r\sigma_\ell/4}} \right) & (\text{Since }\sigma_\ell - \sigma_{\ell-1} \geq \sigma_\ell/ 2 \text{ and } r-1 \geq r/2)\\
        & = \wt{O}\left( n^{\delta - 2\sigma_\ell - 10 /(4\delta)} \right) & (\text{Since } r \sigma_\ell \geq 10/\delta)\\
        & = O\left( n^{\delta - 2\sigma_\ell - 2} \right)  & (\text{Since } \delta \leq 1 \text{ and } n \text{ is sufficiently large enough})\\
        & = O(n^{-1})
    \end{align*}
 Note that we condition on the labels of all vertices except those in the $\ell$-shallow subgraph of vertex $v$. However, in the coupling described in \Cref{lem:same-tree-different-labels-1}, there is no such conditioning on vertex labels. Since the total number of vertices whose labels we are conditioning on is $O(n^{1-\delta + \sigma_{\ell - 1}})$, the probability shift in each step of the coupling in \Cref{lem:same-tree-different-labels} is at most $O(n^{1-\delta + \sigma_{\ell - 1}} / n) = O(n^{\sigma_{\ell -1} - \delta})$. Also, we condition on the high probability event of \Cref{lem:mixer-vertex-in-tree}. Moreover, since the number of steps in the coupling is $|T(v)| - 1$, the total probability shift is upper bounded by  
    \begin{align*}
\left(1 + O(n^{\sigma_{\ell-1}-\delta})\right)^{|T(v)|} \cdot \left(1 + O(n^{\sigma_{\ell - 1}-\delta}) \right)^{|T(v)|-1}\cdot  \left(1 + O(n^{-1}) \right) \leq \left(1+O(n^{\sigma_{\ell - 1}-\delta})\right)^{|T(v)|},\\
\end{align*}
where the inequality is followed by the fact that $-1< \sigma_{\ell - 1} - \delta$, which completes the proof.
\end{proof}

\begin{lemma}[Lemma 6.49 of \cite{behnezhad2024approximating}]\label{lem:small-diameter-small-comp}
    Let $C$ be a connected component consisting of inner edges, forming a tree, and assume that $C$ contains no edges from $E^{\text{inner}}_\ell$. Then, with high probability, the longest path in $C$, considering its undirected edges, has a length of at most $r-1$.
\end{lemma}
\begin{proof}
    Proof is the same as \Cref{lem:same-tree-different-labels} with proper parameters for level $\ell$.
\end{proof}

\begin{lemma}[Modification of Lemma 6.7 of \cite{BehnezhadRR-FOCS23} with Pseudo Edges]\label{lem:same-tree-different-labels-2}
Let $C$ be a connected component of inner edges corresponding to the edges of level lower than $\ell$ that is a tree such that there is no edge of $E^{inner}_\ell$ in $C$. Also, suppose that we condition on having a mixer vertex on all paths that contain at least $r-1$ special crossings. Then, the probability of seeing the same component is equal for both distributions up to $(1 + O(n^{\sigma_{\ell - 1} - \delta}))^{|C|}$ multiplicative factor.
\end{lemma}

\begin{proof}[Proof of \Cref{lem:small-component-focs-modification}]
    Similar to the reasoning in the proofs of \Cref{lem:coupling-high-level} and \Cref{lem:large-component-focs-modification}, the total shift in the probability of the coupling is bounded above by
    \begin{align*}
\left(1 + o(n^{2\delta - 3\sigma_{\ell - 1} - 1})\right)^{|\overline{C}|} \cdot \left(1 + O(n^{\sigma_{\ell - 1}-\delta}) \right)^{|\overline{C}|-1}\cdot  \left(1 + O(n^{\sigma_{\ell - 1}-\delta}) \right)&\leq \left(1+O(n^{\sigma_{\ell - 1}-\delta})\right)^{|\overline{C}|},
\end{align*}
which completes the proof.
\end{proof}

\section{Proof of \Cref{lem:deterministic-lower-bound}}

In this section, we complete the proof of \Cref{lem:deterministic-lower-bound}.

\begin{claim}\label{clm:no-detected-edge-base}
    With probability $1 - O(1/n)$, we have $|E^{\text{inner}}_1| = 0$.
\end{claim}
\begin{proof}
    First, note that by statement (iii) of \Cref{obs:g-function-properties}, we have $1 - g(\ell) < 0$. Therefore, applying statement (ii) of \Cref{lem:advantage-from-top}, we obtain
    \begin{align*}
        \Pr[E^{\text{inner}}_1 = \emptyset] &\geq 1 - O(n^{1-g(1)})\\
        &\geq 1 - O(1/n) & (\text{Since }g(1) > 2 \text{ by statement (iii) of \Cref{obs:g-function-properties}}) \\
        &= 1 - o(1).
    \end{align*}
\end{proof}

\begin{claim}\label{clm:forest-base-level}
    With probability $1 - o(1)$, every connected component formed by queried edges at the base level of the hierarchy is a tree.
\end{claim}
\begin{proof}
    By \Cref{clm:no-detected-edge-base}, we know that $|E^{\text{inner}}_1| = 0$ with probability $1 - O(1/n)$. Conditioning on this event, we apply \Cref{clm:no-cycle-in-small-components}, which states that with probability $1 - O(n^{-\delta + 2\sigma_1 + 10 \sigma_1/\sigma_2}) = 1 - o(1)$, all connected components formed by queried edges at the base level of the hierarchy are trees. This completes the proof.
\end{proof}

\begin{claim}\label{clm:no-edge-between-comp-base}
Let $V_B$ be the set of vertices for which the algorithm detects at least one incident edge at the base level. Let $v$ be an arbitrary vertex in $V_B$. The total number of level-1 edges between $v$ and other vertices of $V_B$, excluding level-1 edges discovered by the algorithm, is at most $\wt{O}(1)$ with high probability.
\end{claim}
\begin{proof}
    By \Cref{clm:max-inner-edges}, there are at most $O(n^{1-\delta + \sigma_1})$ vertices that have at least one level-1 edge. Hence, we have $|V_B| = O(n^{1-\delta + \sigma_1})$. Consider all pairs of vertices where one element of pair is $v$ and the other one is from $V_B \setminus \{v\}$, excluding those the algorithm already queried. There are at most $O(n^{1-\delta + \sigma_1})$ such pairs, and each of them has $\rho n$ parallel ground edges, which means there are at most $O(\rho n^{2-2\delta + 2\sigma_1})$ total ground edges between these pairs.

    For a ground edge to be marked as a level-1 edge, it must be marked as a level-1 pseudo edge. Each ground edge is independently marked as a level-1 pseudo edge with probability $\rho_1/(\rho n)$. Thus, the expected number of level-1 pseudo edges between the corresponding pairs is at most 
    \begin{align*}
        O\left( \rho n^{2-\delta + \sigma_1} \cdot \frac{\rho_1}{\rho n}\right) & = O\left( \rho_1 n^{1-\delta + \sigma_1} \right)\\
        & = O\left(n^{-\delta + 2\sigma_1} \right) & (\text{Since } \rho_1 = O(n^{\sigma_1 - 1}))\\
        & \ll 1 & (\text{Since } \delta > 2\sigma_1).
    \end{align*}
    Since the expected number of such edges is smaller than one, the total number of level-1 edges between $v$ and other vertices of $V_B$, excluding level-1 edges discovered by the algorithm, is at most $\wt{O}(1)$ with high probability.
\end{proof}

\begin{lemma}\label{lem:same-tree-yndist}
    Suppose we condition on the event in \Cref{clm:no-detected-edge-base}, where every connected component formed by queried edges at the base level of the hierarchy is a tree. Let $C_1, C_2, \ldots, C_k$ be the components of the forest found by the algorithm at the base level of the construction on a graph drawn from $\yesdist$. Then, the probability of querying the same forest in a graph drawn from $\nodist$ is at least nearly as large, up to a multiplicative factor of $1 + o(1)$.
\end{lemma}
\begin{proof}
By \Cref{lem:small-diameter-small-comp}, the longest path in any component has length at most $r-1$. As a result, no path within any component contains $r-1$ special crossings, as defined in \Cref{def:special-crossing}. We aim to show that the probability of observing the same set of components in \yesdist{} and \nodist{} differs by at most a $1 + o(1)$ multiplicative factor. Let $\mathcal{L}$ represent a labeling in \yesdist{}. We construct a corresponding labeling $\mathcal{L'}$ in \nodist{} and demonstrate that the probability of obtaining $\mathcal{L'}$ is nearly the same as that of $\mathcal{L}$. 

We proceed by iterating over the components sequentially. Consider a component $C$. At each step, we condition on the labels that have already been revealed in $\mathcal{L'}$. If no edge within $C$ connects two vertices in $A_r$, we assign the same labels in $\mathcal{L'}$, as all other edges remain identical in both \yesdist{} and \nodist{}. Now, suppose there exists an edge $(u,v)$ such that $u, v \in A_r$. Removing this edge results in two separate components, $C_u$ and $C_v$. In $\mathcal{L'}$, we assign $u$ to $A_r$ and $v$ to $B_r$. The labels of $C_u$ and $C_v$ are then coupled according to \Cref{lem:same-tree-different-labels-2}. This follows the same as the proofs of \Cref{lem:small-component-focs-modification} and \Cref{lem:large-component-focs-modification}. Since each vertex in the component is connected to at most $\widetilde{O}(1)$ vertices with revealed labels, as stated in \Cref{clm:no-edge-between-comp-base}, each coupling step contributes at most a $\widetilde{O}(1/n)$ shift in probability. Since the total number of steps corresponds to the number of edges queried by the algorithm at the base level of construction, which is $O(n^{1-\delta + \sigma_1})$, the overall shift in probability remains $o(1)$. Therefore, we can couple the two distributions in such a way that the probability of querying the same forest in both remains nearly identical, differing by at most a $1 + o(1)$ multiplicative factor.
\end{proof}

Now we have all the tools to finish the proof of \Cref{lem:deterministic-lower-bound}. 

\deterministicLowerbound*

\begin{proof}
    By \Cref{lem:matching-size-final-graph}, any algorithm that estimates the size of the maximum matching within an additive error of $\epsilon n$ must be capable of distinguishing whether the input graph is drawn from $\yesdist$ or $\nodist$. Moreover, \Cref{lem:same-tree-yndist} states that the distribution of outcomes observed by the algorithm differs by at most $o(1)$ in total variation distance between $\yesdist$ and $\nodist$. Consequently, with a probability of at least 1/3 over the randomness of the input distribution, the algorithm cannot differentiate between the supports of the two distributions. As a result, any deterministic algorithm that computes an estimate $\widetilde{\mu}$ for the size of the maximum matching in $G$, satisfying 
    $$\mu(G) - \epsilon n \leq \wt{\mu} \leq  \mu(G),$$ with probability 2/3 over the randomness of the input distribution, must have a runtime of at least $\Omega(n^{2-\delta})$.
\end{proof}

\printbibliography

\appendix

\section{Omitted Proofs}
\label{sec:omitted-proofs}

\expectedDegreeBase*
\begin{proof}
    A gadget between vertex sets $X$ and $Y$ with parameter $p$ contributes $p\card{Y}$ to the expected degree of each vertex in $X$.
    For (i), note that the only edge gadget that we use for vertices of $S$ in both $\yesdist^1$ and $\nodist^1$ is the gadget between $S$ and $B_1$. Thus, for $v \in S^j$, we have
    \begin{align*}
        d_{\Phi^1}(v) =  P(S^j, B_1^j) \cdot |B_1^j| = \log^2 n.
    \end{align*}
    To prove (ii), we consider different cases for vertex $v$:
    \begin{itemize}
        \item \textbf{$v \in A_i$ for $1 \leq i  < r$:} For $v \in A_i^j$, we have
        \begin{align*}
            d_{\Phi^1}(v) &= P(A_i^j, B_i^j) \cdot |B_i^j| + P(A_i^j, B_{i+1}^j) \cdot |B_{i+1}^j| + P(A_i^j, D_i^{3-j}) \cdot |D_i^{3-j}|\\
            &\qquad\qquad + \sum_{k = 1}^{i-1} P(A_i^j, D_k^{3-j}) \cdot |D_k^{3-j}|\\
            & = d_1 + \log^2 n + (r-i+1)\gamma d_1 + \sum_{k = 1}^{i-1} \gamma d_1\\
            & = d_1 + r\gamma d_1 + \log^2 n.
        \end{align*}
        \item \textbf{$v \in B_i$ for $1 \leq i < r$:} We abuse notation in this proof and use $A_0^j$ to represent $S^j$. Then, for $v \in B_i^j$, we have
        \begin{align*}
            d_{\Phi^1}(v) &= P(A_i^j, B_i^j) \cdot |A_i^j| + P(B_i^j, A_{i-1}^j) \cdot |A_{i-1}^j| + P(B_i^j, D_i^{j}) \cdot |D_i^{j}|\\
            &\qquad\qquad + \sum_{k = 1}^{i-1} P(A_i^j, D_k^{j}) \cdot |D_k^{j}|\\
            & = d_1 + \log^2 n + (r-i+1)\gamma d_1 + \sum_{k = 1}^{i-1} \gamma d_1\\
            & = d_1 + r\gamma d_1 + \log^2 n.
        \end{align*}
        \item \textbf{$v \in A_r$ in the \yes{} case:} For $v \in A_r^j$ , we have
        \begin{align*}
            d_{\Phi_\yes^1}(v) &= P(A_r^j, B_r^j) \cdot |B_r^j| + P(A_r^j, A_{r}^{3-j}) \cdot |A_{r}^{3-j}| + \sum_{k = 1}^{r} P(A_r^j, D_k^{3-j}) \cdot |D_k^{3-j}|\\
            & = d_1 + \log^2 n  + \sum_{k = 1}^{r} \gamma d_1\\
            & = d_1 + r\gamma d_1 + \log^2 n.
        \end{align*}
        \item \textbf{$v \in A_r$ in the \no{} case:} For $v \in A_r^j$ , we have
        \begin{align*}
            d_{\Phi_\no^1}(v) &= P(A_r^j, B_r^j) \cdot |B_r^j| + \sum_{k = 1}^{r} P(A_r^j, D_k^{3-j}) \cdot |D_k^{3-j}|\\
            & = d_1 + \log^2 n  + \sum_{k = 1}^{r} \gamma d_1\\
            & = d_1 + r\gamma d_1 + \log^2 n.
        \end{align*}
        \item \textbf{$v \in B_r$ in the \yes{} case:} For $v \in B_r^j$, we have
        \begin{align*}
            d_{\Phi_\yes^1}(v) &= P(A_r^j, B_r^j) \cdot |A_r^j| + P(B_r^j, B_{r}^{3-j}) \cdot |B_{r}^{3-j}| + P(B_r^j, A_{r-1}^{j}) \cdot |A_{r-1}^{j}|\\
            & \qquad\qquad + \sum_{k = 1}^{r} P(B_r^j, D_k^{j}) \cdot |D_k^{j}|\\
            & = (1-\xi) d_1 + \xi d_1 + \log^2 n  + \sum_{k = 1}^{r} \gamma d_1\\
            & = d_1 + r\gamma d_1 + \log^2 n.
        \end{align*}
        \item \textbf{$v \in B_r$ in the \no{} case:} For $v \in B_r^j$, we have
        \begin{align*}
            d_{\Phi_\no^1}(v) &= P(A_r^j, B_r^j) \cdot |A_r^j| + P(B_r^j, B_{r}^{3-j}) \cdot |B_{r}^{3-j}| + P(B_r^j, A_{r-1}^{j}) \cdot |A_{r-1}^{j}|\\
            & \qquad\qquad + \sum_{k = 1}^{r} P(B_r^j, D_k^{j}) \cdot |D_k^{j}|\\
            & = (1-\xi) (d_1 + \log^2 n) + \xi d_1 - (1-\xi)\log^2 n + \log^2 n  + \sum_{k = 1}^{r} \gamma d_1\\
            & = d_1 + r\gamma d_1 + \log^2 n.
        \end{align*}
        \item \textbf{$v \in D_i$ for $1 \leq i < r$:} For $v \in D_i^j$, we have
        \begin{align*}
            d_{\Phi^1}(v) &= P(B_i^j, D_i^j) \cdot |B_i^j| + P(A_i^{3-j}, D_{i}^{j}) \cdot |A_{i}^{3-j}| + P(D_i^j, D_i^{3-j}) \cdot |D_i^{3-j}|\\
            & \qquad\qquad + P(A_r^{3-j}, D_i^j)\cdot |A_r^{3-j}| + \sum_{k = i + 1}^{r-1} P(A_k^{3-j}, D_i^j)\cdot |A_k^{3-j}|\\
            & \qquad\qquad + \sum_{k = i + 1}^r P(B_k^j, D_i^j)\cdot |B_k^j| + \sum_{k \neq i} P(D_k^{3-j}, D_i^j)\cdot |D_k^{3-j}|\\
            & = \frac{(r-i+1)\gamma d_1}{\zeta} + \frac{(r-i+1)\gamma d_1}{\zeta} + d_1 + \gamma d_1 + \log^2 n - \frac{(4r - 4i + 2 - \xi)\gamma d_1}{\zeta}\\
            & \qquad\qquad + \frac{(1-\xi)\gamma d_1}{\zeta} + \frac{(r-1-i)\gamma d_1}{\zeta}\\
            & \qquad\qquad + \frac{(r-i)\gamma d_1}{\zeta} + (r-1)\gamma d_1\\
            & = d_1 + r\gamma d_1 + \log^2 n.
        \end{align*}
        \item \textbf{$v \in D_r$:} For $v \in D_r^j$, we have
        \begin{align*}
            d_{\Phi^1}(v) &= P(B_r^j, D_r^j) \cdot |B_r^j| + P(A_r^{3-j}, D_{r}^{j}) \cdot |A_{r}^{3-j}| + P(D_r^j, D_r^{3-j}) \cdot |D_r^{3-j}|\\
            & \qquad\qquad + \sum_{k \neq r} P(D_k^{3-j}, D_r^j)\cdot |D_k^{3-j}|\\
            & = \frac{\gamma d_1}{\zeta} + \frac{(1-\xi)\gamma d_1}{\zeta} + d_1 + \gamma d_1 + \log^2 n - \frac{(2 - \xi)\gamma d_1}{\zeta} + (r-1)\gamma d_1\\
            & = d_1 + r\gamma d_1 + \log^2 n.
        \end{align*}
    \end{itemize}
    Combining all the above cases, we can conclude the proof of (ii).
\end{proof}

\expectedDegreeHigherLevels*
\begin{proof}
A gadget between vertex sets $X$ and $Y$ with parameter $p$ contributes $p\card{Y}$ to the expected degree of each vertex in $X$.
Observe that the non-recursive part is the same for the $\yesdist^\ell$ and $\nodist^\ell$,
and consists only of the gadgets involving the dummies and the gadgets between $A_i^j$ and $B_i^j$. $(i)$ holds because there are no non-recursive gadgets involving $S$.
For the proof of $(ii)$, we examine the expected degree of vertices in each vertex set separately.

\begin{itemize}
    \item $v \in B_i^j$ for $i \in [r]$ and $j \in \{1, 2\}$:
    \begin{align*}
        d_{\Phi^\ell}(v) &=
         P(B_i^j, A_i^j) \cdot \card{A_i^j} 
         + P(B_i^j, D_i^{j}) \cdot \card{D_i^{j}}
         + P(B_i^j, D_i^{j+2}) \cdot \card{D_i^{j+2}} \\
            &\qquad + \sum_{k = 1}^{i-1} P(B_i^j, D_k^{j}) \cdot \card{D_k^{j}} + P(B_i^j, D_k^{j+2}) \cdot \card{D_k^{j+2}}\\
        &= d_\ell + \frac{(r - i + 1) \gamma d_\ell}{2} + \frac{(r - i + 1) \gamma d_\ell}{2} + \frac{(i-1) \gamma d_\ell}{2}
        + \frac{(i-1) \gamma d_\ell}{2} \\
        &= d_\ell + r\gamma d_\ell
    \end{align*}
    \item $v \in A_i^j$ for $i \in [r]$ and $j \in \{1, 2\}$:
    \begin{align*}
        d_{\Phi^\ell}(v) &=
         P(A_i^j, B_i^j) \cdot \card{B_i^j} 
         + P(A_i^j, D_i^{3-j}) \cdot \card{D_i^{3-j}}
         + P(A_i^j, D_i^{5-j}) \cdot \card{D_i^{5-j}} \\
            &\qquad + \sum_{k = 1}^{i-1} P(A_i^j, D_k^{3-j}) \cdot \card{D_k^{3-j}} + P(A_i^j, D_k^{5-j}) \cdot \card{D_k^{5-j}}\\
        &= d_\ell + \frac{(r - i + 1) \gamma d_\ell}{2} + \frac{(r - i + 1) \gamma d_\ell}{2} + \frac{(i-1) \gamma d_\ell}{2}
        + \frac{(i-1) \gamma d_\ell}{2} \\
        &= d_\ell + r\gamma d_\ell
    \end{align*}
    \item $v \in D_i^j$ for $i \in [r]$ and $j \in \{1, 2, 3, 4\}$: We examine $j = 1$. The other cases follow similarly.
    \begin{align*}
        d_{\Phi^\ell}(v) &=
         P(D_i^1, D_i^2) \cdot \card{D_i^2} \\
         & \qquad + \sum_{k \neq i} P(D_i^1, D_k^2) \cdot \card{D_k^2}
         + P(D_i^1, D_k^4) \cdot \card{D_k^4} \\
         &\qquad + P(D_i^1, B_i^1) \cdot \card{B_i^1} 
         + P(D_i^1, A_i^2) \cdot \card{A_i^2} \\
         & \qquad + \sum_{k = i + 1}^r P(D_i^1, B_k^1) \cdot \card{B_k^1} 
         + P(D_i^1, A_k^2) \cdot \card{A_k^2} \\ 
        &=
         d_\ell + \gamma d_\ell  - \frac{(2r - 2i + 1)\gamma d_\ell}{\zeta} \\
         & \qquad + (r-1)\frac{\gamma d_\ell}{2}
         + (r-1)\frac{\gamma d_\ell}{2} \\
         &\qquad + \frac{(r - i + 1) \gamma d_\ell}{2\zeta}
         + \frac{(r - i + 1) \gamma d_\ell}{2\zeta} \\
         & \qquad + \frac{(r - i) \gamma d_\ell}{2\zeta} 
         + \frac{(r - i) \gamma d_\ell}{2\zeta} \\ 
        &= d_\ell + r\gamma d_\ell \qedhere
    \end{align*}
\end{itemize}
\end{proof}

\gFunctionProperties*
\begin{proof}
    We prove each statement separately.
        \paragraph{Proof of \textit{(i)}:} From the definition of $g$, we can express $g(\ell - 1)$ as:
        \begin{align*}
            g(\ell - 1) 
            &= (L - \ell + 3) \cdot \delta - 5 \left( \sum_{i=\ell-1}^{L} \frac{\sigma_i}{\sigma_{i+1}} \right) - 5 \left( \sum_{i=\ell-1}^{L-1} \sigma_i \right) \\
            &= \left[ (L - \ell + 2) \cdot \delta - 5 \left( \sum_{i=\ell}^{L} \frac{\sigma_i}{\sigma_{i+1}} \right) - 5 \left( \sum_{i=\ell}^{L-1} \sigma_i \right) \right] 
            + \delta - 5 \frac{\sigma_{\ell-1}}{\sigma_\ell} - 5 \sigma_{\ell-1} \\
            &= g(\ell) + \delta - 5 \frac{\sigma_{\ell-1}}{\sigma_\ell} - 5 \sigma_{\ell-1}.
        \end{align*}

        \paragraph{Proof of \textit{(ii)}:} Using the expression for $g(\ell - 1)$ derived in part (i), we get:
        \begin{align*}
            1 - g(\ell - 1) - 3\sigma_{\ell-1} 
            &= 1 - g(\ell) - \delta + 5 \frac{\sigma_{\ell-1}}{\sigma_\ell} + 2 \sigma_{\ell-1} \\
        \end{align*}
        \paragraph{Proof of \textit{(iii)}:} Evaluating $g(1)$ with the given definition and using the values of parameters in \Cref{tbl:parameters}, we obtain:
        \begin{align*}
            g(1) 
            &= (L + 1) \cdot \delta - 5 \left( \sum_{i=1}^{L} \frac{\sigma_i}{\sigma_{i+1}} \right) - 5 \left( \sum_{i=1}^{L-1} \sigma_i \right) \\
            &= (L + 1) \cdot \delta - \frac{\delta L}{2} - 5 \left( \sum_{i=1}^{L-1} \left( \frac{\delta}{10} \right)^{L+1-i} \right)\\
            &> (L + 1) \cdot \delta - \frac{\delta L}{2} - \delta \\
            &= 2.
        \end{align*}

        \paragraph{Proof of \textit{(iv)}:} If there exists some $\ell$ such that $g(\ell) = 0$, the parameters in \Cref{tbl:parameters} can be slightly perturbed to satisfy all conditions while ensuring $g(\ell) \neq 0$.
\end{proof}

\maxInnerEdgesDiscovered*
\begin{proof}
Since there are $\rho n$ ground edges between every pair of vertices and the algorithm makes $O(n^{2 - \delta})$ pair queries, the total number of ground edges queried by the algorithm is $O(\rho n^{3 - \delta})$. Each queried ground edges is classified as an inner edge with an independent probability of at most $\rho_{\ell-1}/(\rho n)$, as it must be marked as a pseudo edge of level $\ell - 1$, which occurs with probability $\rho_{\ell-1}/(\rho n)$ independently for all queried ground edges. Therefore, the expected number of such edges found by the algorithm is $O(\rho_{\ell-1} n^{2 - \delta})$.  

Let $X_i$ be an indicator random variable for the event that the $i$-th queried ground edge is a pseudo or real edge, and let $X = \sum X_i$ denote the total number of such edges found by the algorithm. We have $E[X] \leq O(\rho_{\ell-1} n^{2 - \delta})$. Moreover, the random variables $X_i$ are independent. Applying the Chernoff bound, we get:  
\begin{align*}    
\Pr\left[ |X - E[X]| \geq 2 \sqrt{E[X] \log n} \right] \leq 2 \exp \left( -\frac{(2 \sqrt{E[X] \log n})^2}{3 E[X]} \right) < \frac{1}{n}.
\end{align*}  

Thus, with probability at least $1 - 1/n$, the total number of pseudo or real edges discovered by the algorithm is $O(\rho_{\ell-1} n^{2 - \delta})$. Finally, substituting $\rho_{\ell-1} = n^{\sigma_{\ell - 1} - 1}$ completes the proof.  
\end{proof}

\vertexInShallowCountExtended*
\begin{proof}
The proof follows a similar structure to that of \Cref{lem:belong-shallow-count-vertex}. Let $v$ be an arbitrary vertex in the graph. Let $V_i$ be the set of vertices at a distance $i$ from $v$ for $i \in [10 \log n]$ using inner edges in the reverse direction. We will show that, with high probability, $|V_i| \leq 3i\sqrt{\log n}$ using induction.

For the base case of the induction, $i = 1$, the claim holds by \Cref{clm:max-in-deg-top}. Suppose the claim holds for all $i' < i$. Suppose that each vertex in $V_{i-1}$ has at least one incoming edge. Note that this only increases the size of $V_i$. Let $u \in V_{i-1}$. Also, let $\hat{V}$ be the set of vertices that the algorithm has found at least one edge in the queried subgraph. If $u$ has more than one incoming  inner edge, it should be between a vertex that already has an edge because of the way we defined the direction of edges in \Cref{def:directing}. Therefore, there are $|\hat{V}|$ possible pairs between $\hat{V}$ and $u$, each containing $\rho n$ ground edges, and each being marked as a pseudo edge of level $\ell - 1$ with probability at most $p_{\ell-1}/(n\rho)$. Thus, the expected number of inner edges between $\hat{V}$ and $u$ is $|\hat{V}|\cdot \rho_{\ell-1} = O(n^{\sigma_L + \sigma_{\ell - 1} - \delta})$,
since $|\hat{V}| = O(n^{1-\delta+\sigma_L})$ by \Cref{lem:total-edges-discovered}.  Additionally, $|V_{i-1}| \leq 3(i - 1)\sqrt{\log n}$ by induction hypothesis. Hence, the expected number of inner edges between $\hat{V}$ and $V_{i-1}$ is $\wt{O}(n^{\sigma_L + \sigma_{\ell - 1} - \delta})$.

Let $X_i$ be the indicator random variable for the event that the $i$-th ground edge between $\hat{V}$ and $V_{i-1}$ is a pseudo of level $\ell - 1$, and let $X = \sum X_i$ denote the total number of such edges. We have $E[X] = \wt{O}(n^{\sigma_L + \sigma_{\ell - 1} - \delta}) < 1$ for large enough $n$. Furthermore, the random variables $X_i$ are independent. Applying the Chernoff bound, we obtain:
\begin{align*}
\Pr\left[ |X - E[X]| \geq 3 \sqrt{E[X] \log n} \right] \leq 2 \exp \left( -\frac{(3 \sqrt{E[X] \log n})^2}{3 E[X]} \right) < \frac{1}{n^2}.
\end{align*}
Thus, with probability at least $1 - 1/n^2$, the total number of incoming inner edges to $v$ is $3\sqrt{\log n}$. Moreover, we assume that each vertex in $V_{i-1}$ contains at least one incoming inner edge. Hence, we have $|V_i| \leq |V_{i-1}| + 3\sqrt{\log n} = \leq 3i\sqrt{\log n}$ which completes the induction step. Therefore, we have $|V_i| \leq 3i\sqrt{\log n}$ for all $i \leq 10 \log n$. As a result, the total number of $\ell$-shallow subgraphs that contain $v$ is bounded by
\begin{align*}
    1 + \sum_{i=1}^{10\log n} |V_i| \leq 1 + \sum_{i=1}^{10\log n}3i\sqrt{\log n} \leq (10 \log n)\cdot (30 \log n \sqrt{\log n}),
\end{align*}
which completes the proof.
\end{proof}

\edgeInShallowCountExtended*
\begin{proof}
    For directed inner edge $(u,v)$, since $u$ is in at most $\wt{O}(1)$ shallow subgraphs by \Cref{cor:edge-in-shallow-count-extended}, then $(u,v)$ is in at most $\wt{O}(1)$ $\ell$-shallow subgraphs.
\end{proof}

\newpage

\section{Table of Parameters}\label{sec:tableofparameters}

\setlength\extrarowheight{6pt}
\begin{center}
\begin{table}[h!]
\begin{tabular}{ | >{\centering\arraybackslash} m{4.2em} | >{\centering\arraybackslash} m{7em} |  >{\centering\arraybackslash} m{29em}|}   
  \hline
  Parameter & Value & Definition\\ 
  \hline
   \hline
   $\delta$ & - & Parameter that controls the running time of the algorithm. More specifically, the algorithm has $O(n^{2-\delta})$ running time.\\
   \hline
   $\rho$ & $2n^{\delta/10 - 1}$ & The probability that each pseudo edge exists in the \ER{} graph of pseudo edges.\\
   \hline
   $\rho_i$ & $2n^{\sigma_i - 1}$ & The probability that each pseudo edge exists in the \ER{} graph of pseudo edges of level $i$.\\
   \hline
   $L$ & $4/\delta$ & Number of \textbf{levels} in the recursive hierarchy for the construction of input distribution.\\
   \hline
   $r$ & $(10/\delta)^{L+1}$ & Number of \textbf{layers} in each level of the hierarchy.\\
   \hline
   $\yesdist^i$ & - & Distribution of level $i$ graphs in the construction hierarchy that have a perfect matching.\\
   \hline
   $\nodist^i$ & - & Distribution of level $i$ graphs that at most $(1-\epsilon)$ fraction of their vertices can be matched in the maximum matching.\\
   \hline
    $\sigma_i$ & $ (\delta/10)^{L+1-i}$ & Parameter that controls the degree of vertices in graphs of level $i$.\\
   \hline
   $d_i$ & $n^{\sigma_i}$ & Parameter that controls the degree of vertices in graphs of level $i$.\\
   \hline
   $\zeta$ & $1/r^2$ & Fraction of vertices that are delusive in each level.\\
   \hline
   $\xi$ & $1/r^2$ & The gap between the size of $A_r$ and $B_r$ in the base construction.\\
   \hline
   $\gamma$ & $1/r^4$ & Degree to delusive vertices is $\gamma d$.\\
   \hline
   $N_i$ & $N_i = n_{i - 1}/(2\zeta)$ &  Parameter that controls the number of vertices in graphs of level $i$.\\
   \hline
   $n_i$ & $(8 + 16r + 4\zeta r)N_i$ & Total number of vertices in a graph of level $i$.\\
   \hline
   $n$ & $(1+\tau)\cdot n_L$& Number of vertices in a graph drawn from the final distribution.\\
\hline
\end{tabular}
\captionsetup{justification=centering}
\caption{Variables used throughout the paper.}
\label{tbl:parameters}
\end{table}
\end{center}

\end{document}